\documentclass[11pt]{article}
\usepackage{smile}

\usepackage{fullpage}
\usepackage{framed}
\usepackage{mdframed}
\usepackage{enumerate}
\usepackage[inline]{enumitem}
\usepackage[T1]{fontenc}
\usepackage{moresize}
\usepackage{bm}
\usepackage{dsfont}
\usepackage{amsmath}
\usepackage{amssymb}
\usepackage{amsthm}
\usepackage{amsfonts}
\usepackage{stmaryrd}
\usepackage{array}
\usepackage{mathrsfs}
\usepackage{mathtools} 
\usepackage{extarrows}
\usepackage{stackrel}
\usepackage[nodisplayskipstretch]{setspace}
\usepackage{color}
\usepackage[usenames,dvipsnames]{xcolor}
\usepackage{cancel}
\usepackage{soul}
\usepackage{xfrac}
\usepackage{siunitx}
\usepackage{graphicx}
\usepackage{float}
\usepackage{rotating}
\usepackage{subcaption}
\usepackage{overpic}
\usepackage[all]{xy}
\usepackage{tikz}
\usetikzlibrary{arrows,matrix,positioning,calc,automata,patterns}
\usepackage{booktabs}
\usepackage{dcolumn}
\usepackage{multirow}
\usepackage{diagbox}
\usepackage{verbatim}
\usepackage{listings}
\usepackage{algorithm}
\usepackage{algpseudocode}
\usepackage{fancyvrb}
\usepackage{hyperref}
\usepackage[round]{natbib}
\usepackage{sectsty}

\DeclareCaptionType{Pseudocode}

\hypersetup{
    bookmarks=true,         
    unicode=false,          
    pdftoolbar=true,        
    pdfmenubar=true,        
    pdffitwindow=false,     
    pdfstartview={FitH},    
    pdftitle={My title},    
    pdfauthor={Author},     
    pdfsubject={Subject},   
    pdfcreator={Creator},   
    pdfproducer={Producer}, 
    pdfkeywords={key1, key2}, 
    pdfnewwindow=true,      
    colorlinks=true,        
    linkcolor=blue,         
    citecolor=blue,         
    filecolor=blue,         
    urlcolor=cyan           
}

\def\tr{\mathop{\text{tr}}\kern.2ex}

\def\cov{{\rm Cov}}
\def\var{{\rm Var}}

\renewcommand{\Pr}{{P}}

\newcolumntype{L}[1]{>{\raggedright\let\newline\\\arraybackslash\hspace{0pt}}m{#1}}
\newcolumntype{C}[1]{>{  \centering\let\newline\\\arraybackslash\hspace{0pt}}m{#1}}
\newcolumntype{R}[1]{>{ \raggedleft\let\newline\\\arraybackslash\hspace{0pt}}m{#1}}
\newcolumntype{d}[1]{D{.}{.}{#1}}
\newcolumntype{H}{>{\setbox0=\hbox\bgroup}c<{\egroup}@{}}
\newcolumntype{Z}{>{\setbox0=\hbox\bgroup}c<{\egroup}@{\hspace*{-\tabcolsep}}}
\numberwithin{equation}{section}
\newtheorem{theorem}{Theorem}[section]
\newtheorem{lemma}{Lemma}[section]
\newtheorem{proposition}{Proposition}[section]
\newtheorem{assumption}{Assumption}[section]
\newtheorem{corollary}{Corollary}[section]
\newtheorem{definition}{Definition}[section]
\newtheorem{innercustomgeneric}{\customgenericname}
\newtheorem{fact}[theorem]{Fact}
\providecommand{\customgenericname}{}
\newcommand{\newcustomtheorem}[2]{%
  \newenvironment{#1}[1]
  {%
   \renewcommand\customgenericname{#2}%
   \renewcommand\theinnercustomgeneric{##1}%
   \innercustomgeneric
  }
  {\endinnercustomgeneric}
}
\newcustomtheorem{customtheorem}{Theorem}
\newcustomtheorem{customassumption}{Assumption}
\newcustomtheorem{customlemma}{Lemma}
\newcustomtheorem{customexample}{Example}
\theoremstyle{definition}

\newtheorem{remark}{Remark}[section]

\allowdisplaybreaks

\begin{document}

\setlength{\abovedisplayskip}{5pt}
\setlength{\belowdisplayskip}{5pt}
\setlength{\abovedisplayshortskip}{5pt}
\setlength{\belowdisplayshortskip}{5pt}
\hypersetup{colorlinks,breaklinks,urlcolor=blue,linkcolor=blue}

\title{\LARGE Fisher-Pitman permutation tests based on nonparametric Poisson mixtures with application to single cell genomics}

\author{
Zhen Miao\thanks{Department of Statistics, University of Washington, Seattle; e-mail: {\tt zhenm@uw.edu}},~
Weihao Kong\thanks{Google Inc; e-mail: {\tt kweihao@gmail.com}},~Ramya Korlakai Vinayak\thanks{Department of Electrical and Computer Engineering, University of Wisconsin-Madison; e-mail: {\tt ramya@ece.wisc.edu}}, ~Wei Sun\thanks{Public Health Science Division, Fred Hutchinson Cancer Research Center; e-mail: {\tt wsun@fredhutch.org}},~and~Fang Han\thanks{Department of Statistics, University of Washington, Seattle; e-mail: {\tt fanghan@uw.edu}}
}

\date{\today}

\maketitle


\begin{abstract} 
This paper investigates the theoretical and empirical performance of Fisher-Pitman-type permutation tests for assessing the equality of unknown Poisson mixture distributions. Building on nonparametric maximum likelihood estimators (NPMLEs) of the mixing distribution, these tests are theoretically shown to be able to adapt to complicated unspecified structures of count data and also consistent against their corresponding ANOVA-type alternatives; the latter is a result in parallel to classic claims made by Robinson \citep{robinson1973large}. The studied methods are then applied to a single-cell RNA-seq data obtained from different cell types from brain samples of autism subjects and healthy controls; 
empirically, they unveil genes that are  differentially expressed between autism and control subjects  yet are missed using common tests. For justifying their use,  rate optimality of NPMLEs is also established in settings similar to nonparametric Gaussian \citep{wu2020optimal} and binomial mixtures \citep{tian2017learning,vinayak2019maximum}. 
\end{abstract}


\section{Introduction}

Considering an experiment with multiple samples drawn from multiple populations, distinguishing possible difference among them in one or more dimensions is a fundamental statistical task. In the classical test of the null hypothesis of no mean differences, one-way analysis of variance (ANOVA, cf. \cite{fisher1925statistical}) $F$-test is perhaps the most commonly used tool, 
and is the uniformly most powerful invariant one under additional normal assumption, c.f. \citet[page 50]{ScheffeHenry1959}.

Despite its popularity, one-way ANOVA has its competing alternatives. In the context of randomized experiments, Fisher \citep{fisher1935design} initialized an ingenious permutation approach as an alternative to performing ANOVA $F$-test. This idea was later developed further by Pitman \citep{pitman1938significance}. The resulting procedures, often termed the Fisher-Pitman permutation tests in literature, achieve the appealing property of being exactly distribution-free and have been suggested in various contexts as, e.g., when the distributional assumptions of $F$-tests no longer hold \citep{marascuilo1977nonparametric,still1981approximate,berry1983moment}. Robustness properties have been further studied empirically \citep{boik1987fisher} and theoretically \citep{chung2013exact}; power analyses were also performed in \cite{hoeffding1952large} and \cite{robinson1973large}.

Although being originally defined in Euclidean spaces, it is by now well understood that the ANOVA $F$-tests and especially their permutation-type alternatives are able to  adapt to an arbitrary metric space. This is via the approach of ``interpoint" distance functions \citep{mielke1976multi, mielke1984meteorological} that uses an alternative representation of the $F$ statistic as a function of between- and within-group pairwise distances. Thus, through replacing the original Euclidean distance by any properly defined distance function, the idea of Fisher-Pitman permutation tests is now implementable in many complicated metric spaces beyond the Euclidean \citep{anderson2001new,mielke2007permutation,petersen2019frechet}.

Our study of Fisher-Pitman-type permutation tests stems from the analysis of single-cell RNA-seq (scRNA-seq) data, and particularly, a framework that was recently promoted in \cite{sarkar2020separating}. There, the authors described how a separation of measurement and expression models is able to clarify confusion in modeling scRNA-seq data, and accordingly advocated using the terminology of Poisson mixtures to unify many existing models (cf. Table 1 in \cite{sarkar2020separating}). In detail, thinking about $X_{ij}^{(k)}$ to be the absolute expression of a specific gene in cell $i\in [N_{jk}]:=\{1,2,\ldots,N_{jk}\}$ of subject $j\in [n_k]$ of population $k\in [K]$, we are interested in studying the following model of $X_{ij}^{(k)}$ that is a slight simplification to Sarkar and Stephens's Equation (1): 
\begin{align}
X_{ij}^{(k)}~|~\lambda_{ij}^{(k)}  \sim {\rm Poisson}\big(r_{ij}^{(k)}\lambda_{ij}^{(k)}\big);& \quad\quad \text{(measurement model)} \label{model:measure}\\
\lambda_{ij}^{(k)} \sim Q_j^{(k)}.& \quad\quad \text{(expression model)} \label{model:expression}
\end{align}
Here $r_{ij}^{(k)}>0$ adjusts the cell ``read depth'' (cf. \citet[Page 1]{Zhang2020}) and in this paper is assumed to be known; $Q_j^{(k)}$ is a properly defined distribution that describes the ``expression level" of the gene in population $k$ and is assumed to have a compact support on the nonnegative real line. Adopting the statistical terminology, for each $k\in[K]$ and $j\in[n_k]$, $\{X_{ij}^{(k)}, i=1,\ldots,N_{jk}\}$ then independently follow Poisson mixture distributions of point mass functions (PMFs)
\[
h_{ij}^{(k)}(x) := \int_0^{\infty} e^{-\lambda r_{ij}^{(k)}}\frac{\{\lambda r_{ij}^{(k)}\}^x}{x!}{\sf d}Q_j^{(k)}(\lambda), ~~~x=0,1,2,\ldots
\]
and a mixing distribution $Q_j^{(k)}$ that has to be characterized by a nonparametric model; see \citet[Section ``Modeling scRNA-seq data"]{sarkar2020separating} for a discussion of why a nonparametric model of $Q_j^{(k)}$ is preferred in single-cell genomics, though \cite{sarkar2020separating} did not employ such Poisson mixtures for individual level differential expression testing, which however is the main focus of this work. 

Based on the observations $\big\{X_{ij}^{(k)}, i\in [N_{jk}], j\in [n_k], k\in[K]\big\}$ as well as the measure/expression models \eqref{model:measure}-\eqref{model:expression}, a natural question to ask is whether there exists any population-level gene expression difference among the $K$ groups. For this, we propose to leverage a Fisher-Pitman-type permutation test based on consistent estimators $\big\{\tilde Q_j^{(k)}, j\in [n_k], k\in[K] \big\}$ of the mixing distributions $\big\{ Q_j^{(k)}, j\in [n_k], k\in[K] \big\}$ under Wasserstein metrics, which have received much attention in recent mixture distribution estimation literature (see, among many others, \cite{nguyen2013convergence}, \cite{tian2017learning}, \cite{vinayak2019maximum}, \cite{wu2020optimal}, and the references therein). Particularly appealing choices to us include the NPMLE $\hat Q_j^{(k)}$ and its Poisson-smoothed one $h_{\hat Q_j^{(k)}}$ (notation to be introduced by the end of this section); see Section \ref{sec:test} ahead for the detailed description of the testing procedure. 

Many methods have been developed for differential expression analysis of scRNA-seq data \citep{chen2019single}. However, their focus is differential expression between two groups of cells instead of two groups of individuals. For individual level testing, a standard approach is to add up gene expression across all the cells (of a particular cell type) of an individual to create a pseudo-bulk sample, and then apply the methods for differential expression analysis using bulk RNA-seq data, such as DESeq2 \citep{love2014moderated}. 
The novelty of our proposed procedure is that we assess differential expression across individuals using cell level data instead of pseudo-bulk data. Furthermore, the proposed tests are shown to be consistent against their ANOVA-type alternatives, i.e., they are able to asymptotically distinguish the null from any fixed alternative where the ``between-group'' variation is larger than the ``within-group'' variation, a result that sheds insight to the power of the developed tests and is in line with classic observations \citep{hoeffding1952large,robinson1973large}\footnote{In addition to developing a more flexible non-parametric model, another route to boost the power of differential expression analysis is to de-noise the scRNA-seq data; see \cite{ideas} for a proposal along that track.}.

As a byproduct of our theoretical study, this paper further justifies the use of NPMLEs via establishing their rate-optimality in estimating the Poisson mixing distribution under the Wasserstein-1 ($W_1$) metric. Although the consistency of the NPMLEs has been established in the literature for different nonparametric mixture models (cf. \cite{simar1976maximum} for nonparametric Poisson mixtures; and \cite{chen2017consistency} and the references therein for more general models), NPMLEs' rates of convergence and their matching to a minimax lower bound are long standing until very recently. Built on the breakthroughs in binomial \citep{tian2017learning,vinayak2019maximum} and Gaussian mixtures \citep{wu2020optimal} (see also \cite{jiang2019rate} for a related study on the nonparametric
likelihood ratio test) as well as the new analytical techniques devised in
 \cite{jiao2015minimax}, \cite{wu2016minimax}, 
 \cite{jiao2018minimax}, and \cite{han2020optimality}, we are now able to further the optimality of NPMLEs to the nonparametric Poisson mixtures under minimal assumptions on the true mixing distribution function. These results 
yield additional theoretical support for the use of NPMLEs in our developed tests.

The rest of this paper is organized as follows. Section \ref{sec:test} describes the model setup and studies the size and power of the proposed permutation tests. Section \ref{sec:algorithm} discusses implementation of the developed test. The finite-sample performance of the developed (smoothed or not) NPMLE-based permutation tests is investigated in Section \ref{sec:sim}. Section \ref{sec:app} applies the studied tests to a real scRNA-seq data containing single brain nuclei from autism subjects and healthy controls  \citep{velmeshev2019single} and discover significantly differentially expressed genes that cannot be detected using the benchmark DESeq2 method applied on pseudo-bulk data \citep{love2014moderated}. In Section \ref{sec:optimality}, we justify the use of NPMLEs in the permutation tests outlined in Section \ref{sec:test} by providing minimax optimality results for the NPMLE for nonparametric mixture of Poissons. In the last section, Section \ref{Section:Proofs} we provide outline of proofs. All the technical details of the proofs are relegated to a supplement.

{\it Notation.} For any two distributions $P,Q$ on the real line, the Wasserstein-1 distance is defined to be $W_1(P,Q):=\sup_{\ell\in {\rm Lip}_1}\int \ell({\sf d}P-{\sf d}Q)$, where ${\rm Lip}_1$ represents all $1$-Lipschitz functions. For any distribution $P$ on the nonnegative real line, we define its Poisson smoothed version as $$h_Q(x):=\int_0^{\infty}e^{-\lambda}\frac{\lambda^x}{x!}{\sf d}Q(\lambda),\ x = 0, 1, 2, ....$$ For any two constants $a,b$, we denote $a\vee b := \max\{a,b\}$ and $a\wedge b := \min\{a,b\}$. 

\section{Permutation tests}\label{sec:test}

\subsection{Setup} 

Throughout this section, it is assumed that the observations are heterogeneous count data $\{X_{ij}^{(k)}, i\in [N_{jk}], j\in[n_k], k\in [K]\}$ with $N_{jk}=N_{jk,n}\to \infty$ and $n_k=n_{k,n}\to \infty$ as $n:=\sum n_k\to \infty$. In contrast, $K\geq 2$ is assumed to be a fixed integer. It is further assumed that the probability measures $Q_j^{(k)}$'s in \eqref{model:measure} have a common support $[0,B]$ for some $B>0$ that is known a priori (cf. appendix Section \ref{sec:app-details} for a real implementation) and kept to be fixed in this section; later in Section \ref{sec:optimality} we will explore a more general setting where $B=B_n$ is allowed to increase with $n$.

To facilitate the approach to distinguishing differences among the $K$ groups, in addition to the measurement model \eqref{model:measure} and the expression model \eqref{model:expression},  a third-layer ``population model" is introduced to encourage independent and identically distributed (i.i.d.) randomness among each $n_k$ within-group expression models:
\begin{align}\label{model:population}
\text{for each }k\in [K]: \quad Q_1^{(k)}, \ldots, Q_{n_k}^{(k)}~ \stackrel{i.i.d.}{\sim} ~\mathcal{Q}_k.~~~ \text{(population model)}
\end{align}
Here $\mathcal{Q}_k$ is understood to be a probability measure over the Prohorov-metric topology of the space of probability measures that are defined on the Borel $\sigma$-field of $[0,B]$; details about constructing Prohorov-metric topology are referred to Pages 72-73 in \cite{billingsley1999convergence}. Following the discussions in \citet[Section ``Modeling scRNA-seq data"]{sarkar2020separating}, we do not specify $\mathcal{Q}_k$ except for assuming boundedness and well-definedness. 

To wrap up, the model considered in this manuscript, summarizing the three layers (\eqref{model:measure}, \eqref{model:expression}, \eqref{model:population}), is: 
\begin{align}\label{eq:model-all}
&\Big\{X_{ij}^{(k)}, i\in [N_{jk}], j\in[n_k], k\in[K]\Big\} \text{ are independently distributed with PMFs } \notag\\
&\quad\quad\quad\quad\quad \int \Big[\int_0^B e^{-\lambda r_{ij}^{(k)}}\frac{\{\lambda r_{ij}^{(k)}\}^x}{x!}{\sf d}Q(\lambda)\Big]{\sf d}\mathcal{Q}_k(Q), ~~~x=0,1,2,\ldots.
\end{align}

Under the above model, it is understood that $\mathcal{Q}_1,\ldots,\mathcal{Q}_K$ and $K\geq 2$ are fixed, all of which won't change with $n$. Besides $\mathcal{Q}_1,\ldots,\mathcal{Q}_K$ and accordingly the random measures $Q_j^{(k)}$'s, the observations $X_{ij}^{(k)}$'s also depend on the read depths $r_{ij}^{(k)}=r_{ij,n}^{(k)}$'s that are allowed to change with $n$. We are hence faced with a triangular array of possibly highly heterogeneous observations.

\subsection{Tests} 

Under Model \eqref{eq:model-all}, we are interested in testing the following null hypothesis,
\begin{align}\label{eq:H0}
H_0: \mathcal{Q}_1=\mathcal{Q}_2=\cdots=\mathcal{Q}_K,
\end{align}
and aim to detect any population-level difference between groups. Note that here, due to the incorporation of read depths $r_{ij}^{(k)}$'s, the measurements themselves even within each group are generally not identically distributed; thus, a naive empirical distribution function based test could be substantially biased.

The main interest of this paper is to explore how robust a Fisher-Pitman-type test can be when each unobserved subject-level random measure $Q_j^{(k)}$ is replaced by a plug-in-type estimate $\tilde Q_j^{(k)}$ and its Poisson-smoothed version $h_{\tilde Q_j^{(k)}}$ calculated from the measurements $X_{1j}^{(k)},\ldots,X_{N_{jk}j}^{(k)}$. To this end, let's regulate $\tilde Q_j^{(k)}$ as follows. 

\begin{definition}
For any $j\in[n_k]$ and any $k\in[K]$, an estimator $\tilde Q_j^{(k)}$ of $Q_j^{(k)}$ is said to be subject-specific conditionally $W_1$-consistent (shorthanded as ``conditionally $W_1$-consistent'') if it is (i) a function of $X_{1j}^{(k)},\ldots,X_{N_{jk}j}^{(k)}$; (ii) of support $[0,B]$; and (iii) satisfying 
\begin{align}\label{eq:key}
E\Big\{W_1\Big(\tilde Q_j^{(k)}, Q_j^{(k)}\Big)~\Big |~Q_j^{(k)}\Big\} \to 0 \text{ as } N_{jk}=N_{jk,n}\to\infty
\end{align}  
for almost all $Q_j^{(k)}$ with regard to the measure $\mathcal{Q}_k$.
\end{definition}

We next consider the Poisson-smoothed mixing distribution estimator  
\[
h_{\tilde Q_j^{(k)}}:= \int_0^{\infty}e^{-\lambda}\frac{\lambda^x}{x!}{\sf d}\tilde Q_j^{(k)}(\lambda)
\]
based on any conditionally $W_1$-consistent estimator $\tilde Q_j^{(k)}$. It justifies the use of smoothed NPMLEs as an alternative to directly using the original ones; see also Proposition 3.1 in \cite{lambert1984asymptotic} for more results as read depths are all forced to be equal. 

\begin{theorem} \label{theorem:ConsistencyinMixingImpliesConsistencyinMixture}
Suppose $\tilde Q_j^{(k)}$ is conditionally $W_1$-consistent. Then 
\[
E\Big\{W_1\Big(h_{\tilde Q_j^{(k)}}, h_{Q_j^{(k)}}\Big)~\Big |~Q_j^{(k)}\Big\} \to 0 \text{ as } N_{jk}=N_{jk,n}\to\infty
\]
for almost all $Q_j^{(k)}$ with regard to the measure $\mathcal{Q}_k$.
\end{theorem}

A particularly appealing candidate estimator of the mixing distribution is the following NPMLE $\hat Q_j^{(k)}$ with read depth incorporated:
\begin{align}\label{eq:npmle}
\hat Q_j^{(k)} \in \argmax_{Q \text{ of support }[0,B]} \sum\limits_{i\in [N_{jk}]}\log\int_0^{\infty} e^{-\lambda r_{ij}^{(k)}}\frac{\{\lambda r_{ij}^{(k)}\}^{X_{ij}^{(k)}}}{X_{ij}^{(k)}!}{\sf d}Q(\lambda).
\end{align}
Note that here $\hat Q_j^{(k)}$ may not be unique due to read depths, and if there are multiple choices, pick any one of them (cf. Remark \ref{remark:non-uniqueness}). We shall discuss the calculation of $\hat Q_j^{(k)}$ in Section \ref{sec:algorithm}. The next theorem shows that NPMLEs are conditionally $W_1$-consistent under no further assumptions on the population measures $\mathcal{Q}_k$'s except for the already imposed bounded support one. 

\begin{theorem}[Conditionally $W_1$-consistency of NPMLEs] \label{theorem:W1consistency of NPMLEs} Assume $N_{jk}=N_{jk,n}\to \infty$ as $n\to\infty$, $r_{ij}^{(k)}=r_{ij, n}^{(k)}\in [\gamma_0,\gamma_1]$ are uniformly upper and lower bounded by two positive universal constants $\gamma_0,\gamma_1$, and $\mathcal{Q}_k$'s have a common fixed support $[0,B]$. We then have the NPMLEs $\hat Q_j^{(k)}$'s are all conditionally $W_1$-consistent.
\end{theorem}

\begin{remark}
In the literature, consistency of NPMLEs of mixing distributions under the classical i.i.d. mixture distribution setup (corresponding to the case with all read depths identical to each other) has been studied in depth. Notable results include \cite{kiefer1956consistency}, \cite{simar1976maximum}, \cite{pfanzagl1988consistency}; note also the survey by Chen \citep{chen2017consistency}. However, although arising naturally from single-cell genomics modeling, read-depth-incorporated nonparametric mixture distributions have not received much attention in mathematical statistics and, to our knowledge, Theorem \ref{theorem:W1consistency of NPMLEs} delivers the first consistency result for NPMLEs under this heterogeneous setting. 
\end{remark}

Based on any conditionally $W_1$-consistent estimators $\{\tilde Q_j^{(k)}\}$ of $\{Q_j^{(k)}\}$ and their Poisson-smoothed versions $h_{\tilde Q_j^{(k)}}$'s, the proposed ANOVA-type (pseudo-$F$) test statistics are
\[
\tilde F: = \frac{\frac{1}{n}\sum\limits_{k_1,k_2 \in [K]}\sum\limits_{j_1\in [n_{k_1}],j_2\in [n_{k_2}]}W_1\Big(\tilde Q_{j_1}^{(k_1)},\tilde Q_{j_2}^{(k_2)}\Big)^2-\sum\limits_{k\in [K]}\frac{1}{n_k}\sum\limits_{j_1,j_2\in [n_k]}W_1\Big(\tilde Q_{j_1}^{(k)},\tilde Q_{j_2}^{(k)}\Big)^2}{\sum\limits_{k\in [K]}\frac{1}{n_k}\sum\limits_{j_1,j_2\in [n_k]}W_1\Big(\tilde Q_{j_1}^{(k)},\tilde Q_{j_2}^{(k)}\Big)^2}
\]
and
\[
\tilde F_h:= \frac{\frac{1}{n}\sum\limits_{k_1,k_2 \in [K]}\sum\limits_{j_1\in [n_{k_1}],j_2\in [n_{k_2}]} W_1\Big(h_{\tilde Q_{j_1}^{(k_1)}},h_{\tilde Q_{j_2}^{(k_2)}}\Big)^2-\sum\limits_{k\in [K]}\frac{1}{n_k}\sum\limits_{j_1,j_2\in [n_k]}W_1\Big(h_{\tilde Q_{j_1}^{(k)}},h_{\tilde Q_{j_2}^{(k)}}\Big)^2}{\sum\limits_{k\in [K]}\frac{1}{n_k}\sum\limits_{j_1,j_2\in [n_k]}W_1\Big(h_{\tilde Q_{j_1}^{(k)}},h_{\tilde Q_{j_2}^{(k)}}\Big)^2}.
\]
It is ready to check that these two test statistics both reduce to the original one-way ANOVA statistic if the examined space is the real space equipped with the Euclidean norm. The studied statistics then generalize the one-way ANOVA statistics to the $W_1$-metric measure space with different inputs (mixing distribution smoothed or not); similar generalizations have been made in various other (non-)Euclidean spaces \citep{anderson2001new,mielke2007permutation,petersen2019frechet}.

We then move on to introduce the corresponding permuted ANOVA-type test statistics. To this end, for each permutation $\pi: [n]\to [n]$, let $\Pi^{j,k}=(\Pi^{j,k}_1,\Pi^{j,k}_2):=\pi^{\uparrow}(j,k)$ represent the original subject and population indices corresponding to ``the $j$-th subject in the $k$-th group'' after permutation $\pi$. The permuted test statistics are
\[
\tilde F^{\pi}:=\frac{\frac{1}{n}\sum\limits_{k_1,k_2 \in [K]}\sum\limits_{j_1\in [n_{k_1}],j_2\in [n_{k_2}]}W_1\Big(\tilde Q_{j_1}^{(k_1)},\tilde Q_{j_2}^{(k_2)}\Big)^2-\sum\limits_{k\in [K]}\frac{1}{n_{k}}\sum\limits_{j_1,j_2\in [n_{k}]}W_1\Big(\tilde Q_{\Pi^{j_1,k}_1}^{(\Pi^{{j_1,k}}_2)},\tilde Q_{\Pi^{j_2,k}_1}^{(\Pi^{j_2,k}_2)}\Big)^2}{\sum\limits_{k\in [K]}\frac{1}{n_{k}}\sum\limits_{j_1,j_2\in [n_{k}]}W_1\Big(\tilde Q_{\Pi^{j_1,k}_1}^{(\Pi^{{j_1,k}}_2)},\tilde Q_{\Pi^{j_2,k}_1}^{(\Pi^{j_2,k}_2)}\Big)^2}
\]
and
\[
\tilde F_h^{\pi}:= \frac{\frac{1}{n}\sum\limits_{k_1,k_2 \in [K]}\sum\limits_{j_1\in [n_{k_1}],j_2\in [n_{k_2}]} W_1\Big(h_{\tilde Q_{j_1}^{(k_1)}},h_{\tilde Q_{j_2}^{(k_2)}}\Big)^2-\sum\limits_{k\in [K]}\frac{1}{n_{k}}\sum\limits_{j_1,j_2\in [n_{k}]}W_1\Big(h_{\tilde Q_{\Pi^{j_1,k}_1}^{(\Pi^{j_1,k}_2)}},h_{\tilde Q_{\Pi^{j_2,k}_1}^{(\Pi^{j_2,k}_2)}}\Big)^2}{\sum\limits_{k\in [K]}\frac{1}{n_{k}}\sum\limits_{j_1,j_2\in [n_{k}]}W_1\Big(h_{\tilde Q_{\Pi^{j_1,k}_1}^{(\Pi^{j_1,k}_2)}},h_{\tilde Q_{\Pi^{j_2,k}_1}^{(\Pi^{j_2,k}_2)}}\Big)^2}.
\]
The following are the Fisher-Pitman-type permutation tests with nominal level  $\alpha$:
\[
\tilde T_\alpha := 
\begin{cases} 
1, & \text{ if } P(\tilde F^{\pi}<\tilde F~|~\tilde Q_j^{(k)}\text{'s})\geq 1-\alpha,\\
0, & \text{otherwise},
\end{cases}
\] 
and
\[
\tilde T_{h,\alpha}:=
\begin{cases} 
1, & \text{ if } P(\tilde F_h^{\pi}<\tilde F_h~|~\tilde Q_j^{(k)}\text{'s})\geq 1-\alpha,\\
0, & \text{otherwise},
\end{cases}
\]
where the probability here is only with respect to the random permutation $\pi$.

As the (Poisson smoothed-)NPMLEs are chosen, the corresponding tests $\tilde T_\alpha$ and $\tilde T_{h,\alpha}$ are specified as $\hat T_\alpha$ and $\hat T_{h,\alpha}$.

\subsection{Theory}

This subsection provides the necessary theoretical support on the presented tests $\tilde F^{\pi}$ and $\tilde F_h^\pi$. Particular focus is on the asymptotic size and consistency against Robinson-type ANOVA alternatives (cf. Theorem 3 in \cite{robinson1973large}). To minimize assumptions and for presentation clearness, we are focused on the following balanced design case:
\begin{assumption}\label{ass:balance}
The design is balanced so that $n_k=n/K$ and $N_{jk}=N$ for $j\in[n_k]$, $k\in[K]$. In addition, it is assumed that the sets $\{r_{ij}^{(k)}, i\in[N]\}$ are invariant with respect to both $j$ and $k$.
\end{assumption}

\begin{remark}
We note that Assumption \ref{ass:balance} can be weakened in a straightforward manner to allow for $n_k/n\to 1/K$, $N_{jk}$'s asymptotically comparable, and the sets $\{r_{ij}^{(k)}, i\in[N_{jk}]\}$ all weakly converge to a same probability measure that does not depend on the particular choice of $j$ and $k$ (see \citet[Proposition 2.2]{shi2020distribution} as well as \cite{deb2019multivariate} for a similar setup in the recent independence testing literature). We however do not pursue these tracks but rather leave them to the readers of interest to verify. 
\end{remark}

Our first result concerns with the sizes of proposed tests, is of a finite-sample nature, and is a direct consequence of a long line of literature on permutation-based tests.

\begin{theorem}[Size validity]\label{theorem:size_validity}
We have, for any finite $N$ and $n$, as long as $H_0$ in \eqref{eq:H0} and Assumption \ref{ass:balance} hold, 
\[
P(\tilde T_\alpha=1 | H_0) \leq \alpha~~~{\rm and}~~~P(\tilde T_{h,\alpha}=1 | H_0) \leq \alpha.
\]
\end{theorem}

In the following, we are focused on asymptotic results with the balanced design and let $N=N_{n}\to \infty$ as $n \to \infty$. The next theorem is the main result of this subsection. 

\begin{theorem}[Test consistency]\label{theorem:TestConsistency}
Consider $\tilde Q_j^{(k)}$'s to be conditionally $W_1$-consistent estimators of $Q_j^{(k)}$'s. If Assumption \ref{ass:balance} holds, then the following two statements are true. 
\begin{itemize}
\item[(a)] Under any fixed alternative regarding $\mathcal{Q}_1,\ldots,\mathcal{Q}_K$ such that
\begin{eqnarray}\label{eq:ConsistencyH1mixing}
H_1: \frac{1}{K}\sum\limits_{k\in[K]}E\Big\{W_1\Big(Q^{(k)}_1,Q^{(k)}_2\Big)^2\Big\}<\sum\limits_{k_1\neq k_2\in[K]}\frac{E\{W_1(Q_{1}^{(k_1)},Q_{1}^{(k_2)})^2\}}{K(K-1)},
\end{eqnarray}
we have $\lim\limits_{n\to\infty}P(\tilde T_\alpha=1 | H_1)= 1$ for each $\alpha\in(0,1)$.
\item[(b)] Under any fixed alternative regarding $\mathcal{Q}_1,\ldots,\mathcal{Q}_K$ such that
\begin{eqnarray}\label{eq:ConsistencyH1mixture}
H_{1,h}: \frac{1}{K}\sum\limits_{k\in[K]}E\Big\{W_1\Big(h_{Q^{(k)}_1},h_{Q^{(k)}_2}\Big)^2\Big\}<\sum\limits_{k_1\neq k_2\in[K]}\frac{E\Big\{W_1\Big(h_{Q_{1}^{(k_1)}},h_{Q_{1}^{(k_2)}}\Big)^2\Big\}}{K(K-1)},
\end{eqnarray}
we have $\lim\limits_{n\to\infty}P(\tilde T_{h,\alpha}=1 | H_{1,h})= 1$  for each $\alpha\in(0,1)$.
\end{itemize}
\end{theorem}

Specific to (smoothed-)NPMLEs, the following theorem is a direct consequence of Theorems \ref{theorem:W1consistency of NPMLEs}-\ref{theorem:TestConsistency}.

\begin{corollary} Suppose Assumption \ref{ass:balance} and all conditions in Theorem \ref{theorem:W1consistency of NPMLEs} hold. Then the following are true for any $\alpha\in (0,1)$.
\begin{itemize}
\item[(a)] For any finite $N$ and $n$, as long as $H_0$ in \eqref{eq:H0} holds, we have
\[
P(\hat T_\alpha=1 | H_0) \leq \alpha~~~{\rm and}~~~P(\hat T_{h,\alpha}=1 | H_0) \leq \alpha.
\]
\item[(b)] Concerning any fixed alternative $H_1$ (or $H_{1h}$), we have
\[
\lim\limits_{n\to\infty}P(\hat T_\alpha=1~|~H_1) = 1 ~~~{\rm and}~~~\lim\limits_{n\to\infty}P(\hat T_{h,\alpha}=1~|~H_{1,h})= 1.
\]
\end{itemize}
\end{corollary}


\section{Algorithms}\label{sec:algorithm}

This section presents three algorithms to calculate \eqref{eq:npmle}, 
\begin{itemize}
\item[(1)] the vertex direction method (VDM), cf. \cite{fedorov1972theory}, \cite{simar1976maximum}, \cite{wu1978some}, \cite{wu1978someII}, \cite{bohning1982convergence}, and \cite{lindsay1983geometry}; 
\item[(2)] the vertex exchange method (VEM), cf. \cite{bohning1985numerical} and \cite{bohning1986vertex};
\item[(3)] the intra simplex direction method (ISDM), cf. \cite{lesperance1992algorithm}.
\end{itemize}

 

To simplify the notation, in this section we remove $j,k$ from the subscript and use $\{X_i,i\in[N]\}$ and $\{r_i,i\in[N]\}$ to denote the sample points and the corresponding read-depths. Moreover, we use $\hat Q$ to denote the NPMLE defined in \eqref{eq:npmle} based on $\{X_i,i\in[N]$ and $\{r_i,i\in[N]\}$.
For a discrete measure $G$ on $[0,B]$ with support points $\{\lambda_m,m\in[M]\}$,  let $G(\lambda_m)$ stand for the mass $G$ assigned at $\lambda_m$ for each $m\in[M]$. We define 
\[
\Phi(G):=\frac{1}{N}\sum_{i\in[N]}\log\left(\sum_{m\in[M]}G(\lambda_m)e^{-\lambda_mr_i}(\lambda_mr_i)^{X_i}\right)
\]
and its directional derivative from $G$ to $\delta_\lambda$ as
\[
 \Phi^\prime(G,\delta_\lambda)
 :=\lim\limits_{\epsilon\to0^+}\epsilon^{-1}\Big\{\Phi\{(1-\epsilon)G\oplus\epsilon\delta_\lambda\}-\Phi(G)\Big\}
 =\frac{1}{N}\sum_{i\in[N]}\frac{e^{-\lambda r_i}(\lambda r_i)^{X_i}}{\sum_{m\in[M]}G(\lambda_m)e^{-\lambda_mr_i}(\lambda_mr_i)^{X_i}}-1.
 \]
Here $\delta_\lambda$ represents the unit measure at $\lambda\in[0,B]$. Lastly, for any two signed measures $\nu_1$ and $\nu_2$ on the real line, we denote $\nu_1\oplus \nu_2$ as the sum of $\nu_1$ and $\nu_2$, and $\nu_1\ominus \nu_2$ as the sum of $\nu_1$ and $-\nu_2$.

With these notation, we are now ready to present the VDM, VEM, and ISDM algorithms for calculating $\hat Q$. 

\begin{center}
The VDM Algorithm
\begin{enumerate}[label=\textit{Step} {\arabic*}]
\setcounter{enumi}{-1}
\item
(\textit{Initialization}).
Select a point $\lambda_1\in(0,B]$.
Let $G_1=\delta_{\lambda_1}$ be the initial value.
Set the loop index $L=1$.
\item
If $\max\limits_{\lambda\in[0,B]}\Phi^\prime(G_L,\delta_\lambda)=0$, then stop and return $G_L$. Otherwise, find $\lambda_{\text{max}}=\argmax\limits_{\lambda\in[0,B]}\Phi^\prime(G_L,\delta_\lambda)$.
\item Find $\alpha_{\max}=\argmax_{\alpha\in[0,1]}\Phi\Big\{(1-\alpha)G_L\oplus \alpha\delta_{\lambda_{\text{max}}}\Big\}$.
\item
Set $G_{L+1}=(1-\alpha)G_L\oplus \alpha_{\max}\delta_{\lambda_{\text{max}}}$. 
 Set $L=L+1$ and go to Step 1.
 \end{enumerate}
\end{center}
\noindent
\begin{center}
The VEM Algorithm
\begin{enumerate}[label=\textit{Step} {\arabic*}]
\setcounter{enumi}{-1}
\item
(\textit{Initialization}).
Select a point $\lambda_1\in(0,B]$.
Let $G_1=\delta_{\lambda_1}$ be the initial value.
Set the loop index $L=1$.
\item If $\max\limits_{\lambda\in[0,B]}\Phi^\prime(G_L,\delta_\lambda)=0$, then stop and return $G_L$. Otherwise, find $\lambda_{\text{max}}=\argmax\limits_{\lambda\in[0,B]}\Phi^\prime(G_L,\delta_\lambda)$ and $\lambda_{\text{min}}=\argmin\limits_{\lambda\in\text{supp}(G_L)}\Phi^\prime(G_L,\delta_\lambda)$, where $\text{supp}(G_L)$ stands for the support of $G_L$.
\item Find $\alpha_{\max}=\argmax_{\alpha\in[0,1]}\Phi\Big\{G_L\oplus\Big(\alpha G_L(\lambda_{\text{min}})(\delta_{\lambda_{\text{max}}}\ominus\delta_{\lambda_{\text{min}}})\Big)\Big\}$.
\item Set $G_{L+1}=G_L\oplus\Big(\alpha_{\max} G_L(\lambda_{\text{min}})(\delta_{\lambda_{\text{max}}}\ominus\delta_{\lambda_{\text{min}}})\Big)$. 
Set $L=L+1$ and go to Step 1.
\end{enumerate}
\end{center}
\vspace{0.2cm}
\begin{center}
The ISDM Algorithm
\begin{enumerate}[label=\textit{Step} {\arabic*}]
\setcounter{enumi}{-1}
\item
(\textit{Initialization}).
Select a point $\lambda_1\in(0,B]$.
Let $G_1=\delta_{\lambda_1}$ be the initial value.
Set the loop index $L=1$.
\item If $\max\limits_{\lambda\in[0,B]}\Phi^\prime(G_L,\delta_\lambda)=0$, then stop and return $G_L$. Otherwise, find all local maxima $\lambda_{\max,1},\ldots,\lambda_{\max,\mathcal{N}}$ of  $\lambda\mapsto\Phi^\prime(G_L,\delta_\lambda)$ on $[0,B]$, where $\mathcal{N}$ represents the number of local maxima. 
\item Find $(\alpha_{\max,0},\ldots,\alpha_{\max,\mathcal{N}})=\argmax\limits_{\alpha_0,\ldots,\alpha_\mathcal{N}}\Phi\Big\{(1-\alpha_0)G_L\oplus\alpha_1\delta_{\lambda_{\max,1}}\oplus \cdots\oplus\alpha_\mathcal{N}\delta_{\lambda_{\max,\mathcal{N}}}\Big\}$ subject to $\alpha_0\geq 0,\alpha_1\geq 0,\cdots,\alpha_\mathcal{N}\geq 0$ and $\alpha_0+\alpha_1+\cdots+\alpha_\mathcal{N}=1$.
\item Set $G_{L+1}=(1-\alpha_{\max,0})G_L\oplus\alpha_{\max,1}\delta_{\lambda_{\max,1}}\oplus \cdots\oplus\alpha_{\max,\mathcal{N}}\delta_{\lambda_{\max,\mathcal{N}}}$. 
Set $L=L+1$ and go to Step 1.
\end{enumerate}
\end{center}
\medskip

\par\noindent
The convergence of VDM, VEM, and ISDM is guaranteed by the following theorem.
 
\begin{theorem}\label{thm:ConvergenceOfAlgorithm}
Assuming $r_i>0$ for each $i\in[N]$. For each of VDM, VEM and ISDM, if it stops for some $L$, then we have $\Phi(G_L)=\Phi(\hat Q)$;
otherwise, $\Phi(G_L)\to\Phi(\hat Q)$ as $L\to\infty$.
\end{theorem}

\begin{remark}\label{remark:non-uniqueness}
Unlike in the traditional setting where all read depths are identical, when heterogeneous read depths are incorporated, although $G\mapsto \Phi(G)$ is still a concave function, there is no theoretical guarantee about the uniqueness of $\hat Q$'s that maximize the objective function and whether the maximizer is unique or not is still open. This issue of computational uniqueness shall be compared to the parallel result in Theorem \ref{theorem:W1consistency of NPMLEs}, which provides theoretical guarantee for the consistency of an arbitrary maximizer of the objective function as the sample size increases to infinity. 
\end{remark}

\section{Simulation studies}\label{sec:sim}

This section aims to show that the two NPMLE-based (smoothed or not) tests presented in Section \ref{sec:test} cannot dominate each other. Throughout the whole section, we fix $K=2$ and consider the following three designs across with several cases of population models.
\medskip

\par\noindent
\textbf{Designs.}
\begin{enumerate}[itemsep=-.5ex,label=(\Alph*)]
\item\label{scenario(a)} Balanced designs with all read depths set to be 1,  $n_1=n_2=10$, and $N_{jk}=50$, $100$, and $500$ for each $j,k$.
\item\label{scenario(b)} Balanced designs with read-depth effects with
$n_1=n_2=10$ and $N_{jk}=50$, $100$ and $500$ for each $j,k$. In addition, 
in each round of the simulation, $\{r_{i1}^{(1)},i\in[N_{11}]\}$ are i.i.d. generated from $\text{Uniform}(0.5,1.5)$ and then let $r_{ij}^{(k)}=r_{i1}^{(1)}$ for each $j,k$.
\item\label{scenario(c)} A particular unbalanced design motivated by the single-cell RNA-seq data in Section~\ref{sec:app} ahead, with $n_1=10,n_2=13$ and $N_{jk}$ be as in Table~\ref{table:DesignOfN_jk}. 
For each round of the simulation, $\{r_{ij}^{(k)},i\in[N_{jk}],j\in[n_k],k\in[K]\}$ are i.i.d. generated from $\text{Uniform}(0.5,1.5)$.
\end{enumerate}

\begin{table}[h]
\caption{$N_{jk}$ in the unbalanced design (Design~\ref{scenario(c)}) \vspace{-0.5cm}}
\center
\begin{tabular}{|l|l|l|l|l|l|l|l|l|l|l|l|}
\hline
$N_{1,1}$ &$N_{2,1}$  &$N_{3,1}$  &$N_{4,1}$  &$N_{5,1}$  &$N_{6,1}$ & $N_{7,1}$ &$N_{8,1}$  &$N_{9,1}$  &$N_{10,1}$  &$N_{1,2}$  &$N_{2,2}$ \\ \hline
388 & 1142 & 162  &391  &215  &278 &284  &193  &542  &106  &202  & 759\\ \hline
$N_{3,2}$ &$N_{4,2}$  &$N_{5,2}$  &$N_{6,2}$  &$N_{7,2}$  &$N_{8,2}$ &$N_{9,2}$  &$N_{10,2}$  &$N_{11,2}$  &$N_{12,2}$  &$N_{13,2}$  & \\ \hline
415 & 69  &327  &431  &414  &451 &  275&733  &422  & 65 & 362 & \\ \hline
\end{tabular}
\label{table:DesignOfN_jk}
\end{table}

We then move on to specify the population model \eqref{model:population} used in our simulation studies. Hereafter, let $\text{Gam}(a,b;B)$ denote a truncated Gamma distribution with a shape parameter $a>0$, a rate parameter $b>0$, and with any realization larger than $B$ shrunken to $B$. 
Let $\{\Delta_j^{(k)},j\in[n_k],k\in[K]\}$ be i.i.d. generated from $\text{Uniform}(-1,1)$. 
\medskip

\par\noindent
\textbf{Population models.}
\begin{enumerate}[itemsep=-.5ex,leftmargin=*]
\item \label{model1}
\begin{enumerate}[itemsep=-.5ex,leftmargin=*,label=(\alph*)]
\item\label{model(1a)} $Q_j^{(k)} \sim \text{Gam}(14+\Delta_j^{(k)},7/4;50)$ for each $j\in[n_k],k\in[2]$.
\item\label{model(1b)} $Q_j^{(k)} \sim \text{Gam}(14+\Delta_j^{(k)},7;50)$ for each $j\in[n_k],k\in[2]$.
\item\label{model(1c)}  $Q_j^{(k)}\sim \text{Gam}(  6+\Delta_j^{(k)},1;50)$ for each $j\in[n_k],k\in[2]$.
\end{enumerate}
\item \label{model2}
\begin{enumerate}[itemsep=-.5ex,leftmargin=*,label=(\alph*)]
\item\label{model(2a)} 
$Q_j^{(1)}\sim\text{Gam}(14+\Delta_j^{(1)},7/4;50)$ for $j\in[n_1]$ and $Q_j^{(2)}\sim\text{Gam}(6+\Delta_j^{(2)},3/4;50)$ for $j\in[n_2]$.
\item\label{model(2b)} 
$Q_j^{(1)}\sim\text{Gam}(14+\Delta_j^{(1)},7/3;50)$ for $j\in[n_1]$ and $Q_j^{(2)}\sim\text{Gam}(6+\Delta_j^{(2)},1;50)$ for $j\in[n_2]$.
\item\label{model(2c)} 
$Q_j^{(1)}\sim\text{Gam}(14+\Delta_j^{(1)},7/2;50)$ for $j\in[n_1]$ and $Q_j^{(2)}\sim\text{Gam}(6+\Delta_j^{(2)},3/2;50)$ for $j\in[n_2]$.
\end{enumerate}
\item \label{model3}
\begin{enumerate}[itemsep=-.5ex,leftmargin=*,label=(\alph*)]
\item\label{model(3a)} 
$Q_j^{(1)}\sim\text{Gam}(4+\Delta_j^{(1)},1;20)$ for $j\in[n_1]$ and $Q_j^{(2)}\sim\text{Gam}(5+\Delta_j^{(2)},1;20)$ for $j\in[n_2]$.
\item\label{model(3b)}
$Q_j^{(1)}\sim\text{Gam}(5+\Delta_j^{(1)},1;20)$ for $j\in[n_1]$ and $Q_j^{(2)}\sim\text{Gam}(6+\Delta_j^{(2)},1;20)$ for $j\in[n_2]$.
\item\label{model(3c)}
$Q_j^{(1)}\sim\text{Gam}(6+\Delta_j^{(1)},1;20)$ for $j\in[n_1]$ and $Q_j^{(2)}\sim\text{Gam}(7+\Delta_j^{(2)},1;20)$ for $j\in[n_2]$.
\end{enumerate}
\item \label{model4}
\begin{enumerate}[itemsep=-.5ex,leftmargin=*,label=(\alph*)]
\item\label{model(4a)}
$Q_j^{(1)}\sim\text{Gamma}(11+\Delta_j^{(1)},1;50)$ for $j\in[n_1]$ and $Q_j^{(2)}\sim\text{Gamma}(12+\Delta_j^{(2)},1;50)$ for $j\in[n_2]$.
\item\label{model(4b)}
$Q_j^{(1)}\sim\text{Gamma}(12+\Delta_j^{(1)},1;50)$ for $j\in[n_1]$ and $Q_j^{(2)}\sim\text{Gamma}(13+\Delta_j^{(2)},1;50)$ for $j\in[n_2]$.
\item\label{model(4c)}
$Q_j^{(1)}\sim\text{Gamma}(13+\Delta_j^{(1)},1;50)$ for $j\in[n_1]$ and $Q_j^{(2)}\sim\text{Gamma}(14+\Delta_j^{(2)},1;50)$ for $j\in[n_2]$.
\end{enumerate}
\end{enumerate}

Our focus is on examining as well as comparing the empirical performance of the tests $\hat T_{\alpha}$ and $\hat T_{h,\alpha}$ with NPMLE calculated using the oracle $B$.
Both of them are based on an exact critical value approximated by 1,000 Monte Carlo simulations. 
The underlying nominal significance level is $0.05$.
For each setting, $1,000$ rounds of simulations were performed. We use VEM to compute NPMLEs with a stop tolerance 0.01.
Optimization in Step 1 and Step 2 in VEM is implemented by the default interior-point algorithm in Matlab; see the support page of function `\textit{fmincon}' for further details.

Table \ref{tab:power1} shows the empirical sizes and powers (rejection frequencies) of tests $\hat T_{\alpha}$ and $\hat T_{h,\alpha}$.
In short, the results confirm our earlier theoretical claims on the sizes and powers of $T_{\alpha}$ and $\hat T_{h,\alpha}$ in the different models and balanced designs (Designs~\ref{scenario(a)} and \ref{scenario(b)}).
Moreover, even under the unbalanced design (Design \ref{scenario(c)}), $T_{\alpha}$ and $\hat T_{h,\alpha}$ still perform well in terms of their empirical sizes and powers.

Some more detailed comparisons between $T_{\alpha}$ and $\hat T_{h,\alpha}$ are in line. The following observations depend on the ``signal strengths'' $D$ and $D_h$, defined as follows:
\begin{eqnarray}
D:=E\{W_1(Q_1^{(1)},Q_1^{(2)})^2\}-\left(E\{W_1(Q_1^{(1)},Q_2^{(1)})^2\}+E\{W_1(Q_1^{(2)},Q_2^{(2)})^2\}\right)/2
\label{Formula:DifferencesMixing}
\end{eqnarray}
and
\begin{eqnarray}
D_h:=E\{W_1(h_{Q_1^{(1)}},h_{Q_1^{(2)}})^2\}-\left(E\{W_1(h_{Q_1^{(1)}},h_{Q_2^{(1)}})^2\}+E\{W_1(h_{Q_1^{(2)}},h_{Q_2^{(2)}})^2\}\right)/2.
\label{Formula:DifferencesMixture}
\end{eqnarray}



First, empirical results for Model \ref{model1} illustrates that under $H_0$, empirical powers are close to the nominal level $\alpha=0.05$, confirming the size validity of $\hat T_{\alpha}$ and $\hat T_{h,\alpha}$.  In addition, even under the unbalanced design (Design~\ref{scenario(c)}), empirical powers are stable and close to the nominal level $\alpha=0.05$, indicating the robustness of the studied tests.

Second, we compare the empirical powers using Models \ref{model2}, \ref{model3}, and \ref{model4}. In Model \ref{model2}, $D$ is significantly larger than $D_h$ and the corresponding empirical powers of $\hat T_{\alpha}$ are all larger than these of $\hat T_{h,\alpha}$ in all three considered designs (Designs~\ref{scenario(a)},~\ref{scenario(b)}, and \ref{scenario(c)}).
This phenomenon is not surprising to us as the difference between variation between groups and variation within groups in mixing distributions is much larger than that in mixture distributions. 
Therefore,  $\hat T_{\alpha}$ is more powerful than $\hat T_{h,\alpha}$.

In Model \ref{model3}, $D$ is approximately equal to $D_h$ and the empirical power of $\hat T_{\alpha}$ is smaller than the empirical power of $\hat T_{h,\alpha}$ when $N$ is small (e.g., $50$ and $100$).  However, the empirical powers of $\hat T_{\alpha}$ and $\hat T_{h,\alpha}$ are close when $N$ is large. 
Similar observation applies to Model \ref{model4}, where $D$ is also approximately equal to $D_h$. However, compared to Model \ref{model3}, the mixing distributions in Model \ref{model4} have larger $B$ and thus the empirical powers of $\hat T_{h,\alpha}$ are higher than the empirical powers of $\hat T_{\alpha}$ even for $N=500$, especially under Design~\ref{scenario(a)}. Some pilot studies to explain this phenomenon will be put in Section \ref{sec:optimality}, where we analyze the finite-sample behavior of the NPMLE under an exploratory simplified setting where all read depths are fixed to be 1. There, the rate of convergence of NPMLE, at the worst case, is showed to be $O(\log\log N/\log N)$; in contrast, \citet[Lemma 4.1 and Theorem 4.1]{lambert1984asymptotic} showed that the Poisson-smoothed NPMLE attains a near-root-n rate of convergence to the mixture distribution. 

{
\renewcommand{\tabcolsep}{1.5pt}
\renewcommand{\arraystretch}{1.0}
\begin{table}[t]
\caption{Empirical sizes and powers of $\hat T_{\alpha}$ and $\hat T_{h,\alpha}$; here $D$ and $D_h$ are defined in \eqref{Formula:DifferencesMixing} and \eqref{Formula:DifferencesMixture}}
\centering
{{
\begin{tabular}{cC{.45in}C{.45in}C{.45in}C{.45in}C{.45in}C{.45in}
         C{.45in}C{.45in}C{.45in}C{.45in}C{.45in}C{.45in}C{.45in}}
\toprule
  \multicolumn{2}{c}{Model}   &  \ref{model1}\ref{model(1a)}  & \ref{model1}\ref{model(1b)} &   \ref{model1}\ref{model(1c)}   &   \ref{model2}\ref{model(2a)}   & \ref{model2}\ref{model(2b)} 
     &\ref{model2}\ref{model(2c)} & \ref{model3}\ref{model(3a)}  & \ref{model3}\ref{model(3b)} &   \ref{model3}\ref{model(3c)}   &   \ref{model4}\ref{model(4a)}   & \ref{model4}\ref{model(4b)} &\ref{model4}\ref{model(4c)} \\
  & $D$&   0   &   0   &   0   &   0.59   &   0.32        & 0.15&   0.99   &   0.99   &   0.99   &   0.99   &   0.99   &  0.99   \\
&$D_h$  &   0   &   0   &   0   &   0.22   &   0.10        & 0.03 &   0.99   &   0.99   &   0.99   &   0.99   &   0.99   &  0.99   \\
\midrule
  &N     &  \multicolumn{12}{c}{Empirical sizes/powers for $\hat T_{\alpha}$ under Design \ref{scenario(a)}}   \\
 & 50   &   0.054  	&    0.050  &   0.045  &   0.644  &  0.595    & 0.356	&   0.811	&   0.772	&   0.698	&   0.538	&  0.501	&  	0.502\\
& 100  &   0.043  	&    0.055  &   0.053 	&   0.901  &   0.835	&0.583	&   0.870	&   0.872	&  0.843	&  0.723	& 0.680	&   	0.668\\
& 500  &   0.049	&    0.049 	&   0.060	&    0.996	&   0.999  &0.965	&   0.952	&   0.958	&   0.941	&   0.850	& 0.831	&	0.829\\
\vspace{-.5em}\\
 &N     &  \multicolumn{12}{c}{Empirical sizes/powers for $\hat T_{h,\alpha}$ under Design \ref{scenario(a)}}   \\
 & 50  &   0.054   	&   0.045  &   0.049	&   0.284	&   0.210	& 0.111	&  0.833	&   0.816	&   0.767	&   0.650	& 0.635	&	0.624\\
& 100 &   0.038		&   0.063	&   0.049	&   0.371  	&   0.264	& 0.138	&   0.896	&  0.892	&   0.882	&   0.797	& 0.788	&	0.771\\
& 500 &   0.042		&   0.047	&   0.055	&   0.492  &   0.309	&0.186	&   0.951	&   0.961	&   0.947	&   0.944	&   0.924	&	0.921\\
\vspace{-.5em}\\
   &N  &  \multicolumn{12}{c}{Empirical sizes/powers for $\hat T_{\alpha}$ under Design \ref{scenario(b)}}   \\
&  50 &   0.044		&    0.048	&  0.058	&   0.644	& 0.508	& 0.338	&0.796	&  0.763	&0.729   	& 0.559	&  0.522	&  0.520	\\
& 100&    0.053		&   0.050	&  0.062	&   0.863	&  0.779	& 0.518	&  0.878	&   0.862	& 0.846	&   0.714	&   0.735	&  0.679	\\
& 500&    	0.036	&    0.052	&   0.054	&   1.000	&  0.998	&0.972	& 0.958	&  0.952	& 0.939	&   0.922	&   0.920	& 0.913	\\
\vspace{-.5em}\\
&N  	&  \multicolumn{12}{c}{Empirical sizes/powers for $\hat T_{h,\alpha}$ under Design \ref{scenario(b)}}   \\
&50  &   0.044		&   0.050	&  0.054	&   0.262	&   0.193	&0.100	& 0.821	& 0.806	& 0.772	&  0.632	&  0.619	&  0.602	\\
&100&   0.058		&  0.041	&  0.053	&  0.350	&  0.276	&0.132	&   0.885	&  0.877	&  0.858	&  0.772	&  0.788	& 0.759	\\
&500&   0.036		& 0.045	&  0.057	&  0.501	&  0.414	&0.187	&  0.956	& 0.950	&  0.943	&   0.932	&  0.928	& 0.924	\\
\vspace{-.5em}\\
    &N     &  \multicolumn{12}{c}{Empirical sizes/powers for $\hat T_{\alpha}$ under Design \ref{scenario(c)}}   \\
\multicolumn{2}{c}{Table~\ref{table:DesignOfN_jk}}  
	&    	0.048	& 0.050	&  0.051	&  0.994	&  0.988	& 0.900	& 0.962	& 0.940	& 0.951	&0.910	&  0.904	&   0.907\\
 \vspace{-.5em}\\
  &N     &  \multicolumn{12}{c}{Empirical sizes/powers for $\hat T_{h,\alpha}$ under Design \ref{scenario(c)}}   \\
\multicolumn{2}{c}{Table~\ref{table:DesignOfN_jk}} 
	&   	0.047	&0.051	&0.052	&0.452	&0.346	&0.173	&0.966	&0.947	&0.952	&0.929	&0.920	&0.922	\\
\bottomrule 
\end{tabular}}}
\label{tab:power1}
\end{table}
}

\section{Applications to single-cell genomics}\label{sec:app}

This section applies the studied permutation tests to a scRNA-seq data. There has been a large literature studying fitting RNA-seq data using Poisson mixtures including, e.g., over-dispersed Poisson model \citep{robinson2010edgeR},  Poisson-Gamma model \citep{love2014moderated, huang2018saver}, Poisson-Beta model \citep{vu2016beta}, Poisson-log normal model \citep{silva2019multivariate}, Poisson mixture model with K-clusters \citep{rau2015coexpression},  finite Poisson mixture models \citep{wu2013pm}, zero-inflated mixture Poisson linear models \citep{liu2019modelling}, Poisson mixture models with unimodal mixing distributions \citep{lu2018generalized}.
Compared to parametric Poisson mixture models, nonparametric Poisson mixture models haven't received much attention; some notable exceptions include \cite{bi2013npebseq}, \cite{dadaneh2018bnpseq}, \cite{sarkar2020separating}, the latter of which was closely followed by us.


\subsection{Data set description}
The scRNA-seq data used in this paper is obtained from \cite{velmeshev2019single}, 
which focused on autism spectrum disorder (ASD) and recorded gene expression of 23 subjects (13 ASD v.s. 10 control) and 18,041 genes for each subject from 17 different cell types and 2 different brain regions. Here we focus on the brain region prefrontal cortex, which is more relevant to autism disease etiology. 
Moreover, each subject has 7 covariates including age, sex, diagnosis, capbatch, seqbatch, post-mortem interval (PMI), and RNA integrity number (RIN). 

We focus on a pre-selected subset including 100 genes (names of the genes put in Table~\ref{Table:ApplicationTable1}) that were documented to be related to body height; for relation between ASD and body height, see, e.g., \cite{fukumoto2011head} and \cite{chawarska2011early}.
In addition to permutation testing with either estimated mixing distributions or mixture distributions, we also consider DESeq2 \citep{love2014moderated} as a benchmark. In implementing the two considered permutation tests, we adopt a common strategy to incorporate four covariates age, sex, seqbatch, and RIN. The other two covariates PMI and capbatch are not significantly associated with gene expression given the other covariates, since their p-value distributions across all genes are uniform. The corresponding tests were denoted as $\hat T_Z$ (with the original NPMLE) and $\hat T_{h, Z}$ (with the Poisson-smoothed NPMLE). Details of the implementation were put in appendix Section \ref{sec:app-details}.

\subsection{Implementation results}

Using $\hat T_Z$, 9 genes are significant under the threshold of false discovery rate (FDR) 0.05 after multiple testing correction by the Benjamini-Hochberg procedure.
Replacing $\hat T_Z$ by $\hat T_{h,Z}$, 8 genes are significant under the same threshold of FDR and 7 genes are coincident with significant genes found by $\hat T_Z$. This shows some consistency between $\hat T_Z$ and $\hat T_{h,Z}$.


Furthermore, by DESeq2 there are 7 significant genes under the same threshold of FDR and all of them are coincident with significant genes found by $\hat T_Z$.
In other words, among significant genes found by $\hat T_Z$, 78\% significant genes are coincident with genes found by DESeq2 and 22\% are new which means $\hat T_Z$ could enrich the set of significant genes found by the standard method DESeq2.

Similarly, 6 genes are coincident with significant genes found by $\hat T_{h,Z}$.
In other words, among significant genes found by $\hat T_{h,Z}$, 75\% significant genes are coincident with genes found by DESeq2 and 25\% are new which means $\hat T_{h,Z}$ could enrich the set of significant genes found by the standard method DESeq2.
In one word, both $\hat T_{Z}$ and $\hat T_{h,Z}$ could enrich the set of significant genes found by DESeq2.
Further details are summarized in Figure~\ref{Figure:ThreeCircles}.

\begin{figure}[]
\centering
\includegraphics[width=0.65\textwidth]{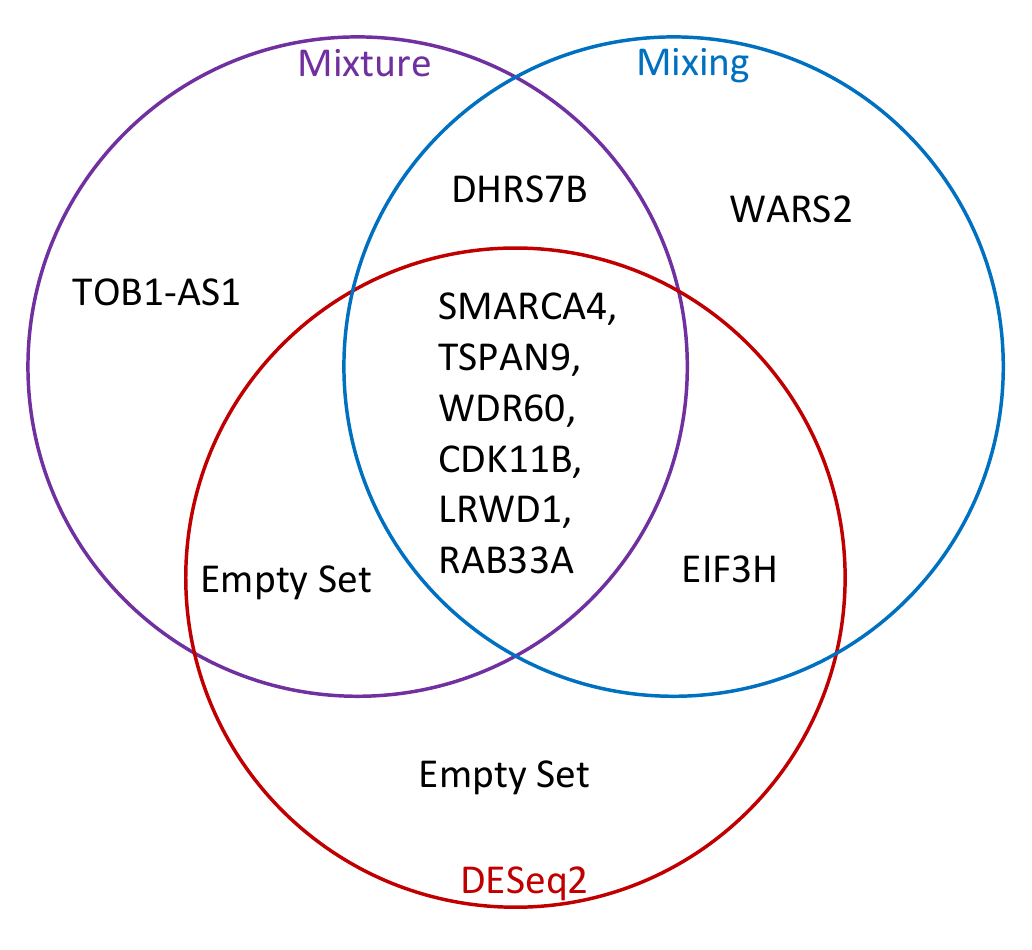}
\caption{Significant genes selected using Mixing ($\hat T_{Z}$), Mixture ($\hat T_{h,Z}$), and DESeq2 methods.}
\label{Figure:ThreeCircles}
\end{figure}

Our results can also be justified by functions of significant genes.
For example, fasting blood glucose measurement is not only one of functions of gene DHRS7B, but also related to ASD \citep{hoirisch2019autism}. More such results are summarized in Table~\ref{Table:FunctionsOfSignificantGenes}.

\begin{table}[]
\caption{All genes used in Section~\ref{sec:app}}
\centering
\begin{tabular}{*{8}{l}}
\hline
DST      		&CHSY3  			&TSC2			&EHD4		&HERC1		&KIF16B   		&DLGAP1			&PIK3CG		\\ 
ELL    		&ODF2L 			&FBXL5			&LNX1		&ERGIC3	 	&CBFA2T2 		&FAM20A	  		&STAT2		\\ 
DAP     		&SSH2 			&WDR60			&SAXO1		&FOXP2 		&SAMD4A    		&TSPAN9			&ARAP3			\\ 
GHR			&KCNK9			&RGL1			&SOCS5		&ZNF76		&ADAMTS2    		&DHRS7B		&PNMA8C	\\ 
KIZ	     		&SHPRH			&RBMS3			&MFSD2B	&NR4A3 		&CCDC171     		&RAB33A			&WDR70	\\
IL16        		&MTMR3			&CDK10			&ZNF628		&CAPZB		& ATXN7L3      		&PSKH1			&FGFRL1		\\ 
BST2 		&UMAD1 			&CPED1			&ESYT2		&LRRC43		&SMARCA4     		&MYO18A		&IL17RD	\\ 
LHX2     		&FBP2 			&ZC3H13			&SRRM2		&NOTCH1 	&HSD17B3      		&SBNO1			&EIF3H		\\ 
RLF      		&LAYN			&SUSD5			&DOT1L		&WARS2 		&RPS4XP13     	&PHF11			&CDK11B		\\ 
DAZL     		&CYFIP2 			&ST7L			&CWC27		&C9orf152	&TOB1-AS1       	&HIF1AN			&KLHL28	\\ 
BCL9		&LRWD1			&LMO7			&PTENP1		&CEP112		&LINC01572		&PPP4R2			&UBE2Z	\\
NRK			&GCLC			&PPM1H			&ITGA9		&HIP1R		&PPP1R16A		&POLR3E			&TANC2 \\	
\multicolumn{2}{l}{ANKDD1A}	 	&\multicolumn{2}{l}{ZNF710-AS1}	&\multicolumn{2}{l}{ZRANB2-AS2}	&\multicolumn{2}{l}{DNAJC27-AS1}	\\ \hline		
\end{tabular}
\label{Table:ApplicationTable1}
\end{table}

\vspace{0.2cm}

\begin{table}[]
\centering
\caption{Significant genes on ASD with some literature support. 
The first column includes names of genes, the second column includes functions potentially related to ASD, and the third column includes literatures supports}
\begin{tabular}{lll}
\hline
gene name				&related functions												&literatures								\\   \hline
DHRS7B					&fasting blood glucose measurement								&\cite{hoirisch2019autism}					\\ 
WDR60					&abnormality of refraction											&\cite{ezegwui2014refractive}					\\ 
EIF3H					&reaction time measurement										&\cite{baisch2017reaction}					\\ 
LRWD1					&insomnia measurement											&\cite{hohn2019insomnia} 					\\ 
RAB33A					&bipolar disorder												&\cite{joshi2012examining}					\\ 
TSPAN9					&creatinine measurement											&\cite{cameron2017variability}					\\ 
WARS2, CDK11B			&heel bone mineral density										&\cite{calarge2017bone} 						\\ 
SMARC4, TOB1-AS1		&cholesterol measurement										&\cite{benachenhou2019implication} 			\\ 
\hline
\end{tabular}
\label{Table:FunctionsOfSignificantGenes}
\end{table}

\section{Minimax optimality of the Poisson NPMLEs}\label{sec:optimality}

This section provides additional theoretical support for the use of NPMLEs in forming up the tests $\hat T_\alpha$ and $\hat T_{h,\alpha}$ in Section \ref{sec:test}. To this end, due to the technical challenges, focus is restricted to a simplified setting of \eqref{eq:model-all}, where the observations $\{X_i,i\in[N]\}$  independently follow a distribution of PMF
\begin{eqnarray}\label{Formula:DefinitionOfh_Q}
h_Q(x)=\int_0^B e^{-\lambda}\frac{\lambda^x}{x!}{\sf d}Q(\lambda), ~~~x=0,1,2,\ldots,
\end{eqnarray}
where $Q$ is a deterministic measure supported on $[0,B]$ that cannot be characterized by a simple parametric model. This is exactly the classic nonparametric Poisson mixture setup, and we study the nonasymptotic behavior of the following NPMLE 
\begin{eqnarray}\label{Formula:DefinitionOfwidehatQ}
\hat Q=\argmax\limits_{Q \text{ of support }[0,B]} \sum\limits_{i\in [N]}\log h_Q(X_i). 
\end{eqnarray}
Note that, the above NPMLE is the simplified version of \eqref{eq:npmle} with all read depths there forced to be one.

There has been an enormous literature studying the NPMLE \eqref{Formula:DefinitionOfwidehatQ} under the nonparametric Poisson mixture model \eqref{Formula:DefinitionOfh_Q}. Earlier results on the existence, discreteness (of the NPMLE support), and computation include, among many others, \cite{simar1976maximum}, \cite{laird1978nonparametric}, \cite{jewell1982mixtures}, \cite{lindsay1983geometry}, \cite{lindsay1983geometryII}, and \cite{lindsay1993uniqueness}; see also \cite{lindsay1995mixture} for a survey. Consistency of NPMLEs were established in, among many others, 
\citet{kiefer1956consistency}, \citet{simar1976maximum}, and 
\cite{pfanzagl1988consistency}; see also \cite{chen2017consistency} for a survey.

Beyond these important results, there has been another track of substantial research that is focused on establishing the minimax rate in estimating the mixing distribution (mostly on the density function) of nonparametric Poisson mixtures. Notable results there include, e.g., \cite{zhang1995estimating}, \cite{loh1996global}, \cite{VanDeGeer1996rates}, \cite{hengartner1997adaptive},  \cite{VanDeGeer2003asymptotic}, \cite{roueff2005nonparametric}, and \cite{rebafka2015nonparametric}. However, to our knowledge, a study on the minimax optimality and the corresponding convergence rates for NPMLEs under a fully nonparametric Poisson mixture model is still absent from the literature.

We would love to highlight again that, due to the nature of nonasymptotic analysis, all the parameters in the model, including $B$, are allowed to change with $N$. This is a strict generalization of the ``asymptotic'' setting in Section \ref{sec:test}, where, due to the additional hardness of handling the read depth as well as for simplifying notation and assumptions, we do not intend to establish similar nonasymptotic results.

Our first theorem concerns with the NPMLE's rate of convergence. 
\begin{theorem}[Upper bound of NPMLEs]\label{theorem:UpperBoundOfMixing}~
\begin{itemize}
\item[(a)] Suppose there exists a universal constant $c_0>0$ such that $B\leq c_0\log N$. Then there exists a positive constant $C=C(c_0)$ such that for all sufficiently large $N$ ($>N_0(c_0)$) we have
\[
\sup_{Q \text{ of support }[0,B]}E\Big\{W_1(\hat{Q},Q)\Big\}\leq C\frac{B}{\log N}\log\left(\frac{\log N}{B}\vee e\right).
\]
\item[(b)] Suppose there exist universal strictly positive constants $c_0, C_0$ and $\epsilon_0\in(0,1/3)$ such that $B\in[c_0\log N,C_0N^{1/3-\epsilon_0}]$. Then there exists a strictly positive constant $C=C(\epsilon_0,c_0)$ such that for all sufficiently large $N$ ($>N_0(c_0,C_0,\epsilon_0)$) we have
\[
\sup_{Q \text{ of support }[0,B]}E\Big\{W_1(\hat{Q},Q)\Big\}\leq C\sqrt{\frac{B}{\log N}}.
\]
\end{itemize}
\end{theorem}

Our second theorem concerns with minimax lower bounds in estimating mixing distributions in model \eqref{Formula:DefinitionOfh_Q}. Combined with Theorem \ref{theorem:UpperBoundOfMixing}, it confirms the NPMLE's minimax optimality. 

\begin{theorem}[Minimax lower bound of mixing distribution estimation]\label{theorem:LowerBoundOfMixing} ~
\begin{itemize}
\item[(a)] Supposing there exists $c_0>0$ such that $B\leq c_0\log N$,
it follows that for any $N\geq 3$,
\[
\inf_{\tilde{Q}}\sup_Q E\{W_1(\tilde{Q},Q)\}\geq\frac{B}{24e\log N}\log\Big( \frac{16c_0\log N}{B}\Big).
\]
\item[(b)] Supposing there exists $c_0>0$ such that $B\geq  c_0\log N$, it follows that for any $N\geq 1$,
\[
\inf_{\tilde{Q}}\sup_Q E\{W_1(\tilde{Q},Q)\}\geq \frac{3}{40e^4}\sqrt{\frac{B}{c_0\log N}}.
\]
\end{itemize}
In the above, the infimum and supremum are understood to be taken over all estimators and all distributions of support $[0,B]$
\end{theorem}

\begin{remark}
Under fully nonparametric binomial mixture models, minimax optimal convergence rates for NPMLEs of mixing distributions were obtained by \citet[Section 3]{vinayak2019maximum} in terms of the $W_1$ distance.
Under fully nonparametric binomial and Gaussian mixture models, \citet[Theorem 1]{tian2017learning} and \citet[Page 1985]{wu2020optimal} obtained optimal convergence rates for moment-based estimators in terms of $W_1$ distance; see also \citet[Remark 2]{polyanskiy2020self}.
\citet[Theorems 1 and 2]{nguyen2013convergence} upper bounded the Wasserstein distance between mixing distributions by the divergence between the corresponding mixture distributions under general mixture models, with normal mixture models as an example in Example 2. 
However, their results cannot be applied here since Theorem 1 restricts the mixing distribution being discrete and Theorem 2 is only for convolution mixture models.
\end{remark}

\section{Proofs}
\label{Section:Proofs}

\subsection{Proofs of theorems in Section~\ref{sec:test}}

\begin{proof}[Proof of Theorem~\ref{theorem:ConsistencyinMixingImpliesConsistencyinMixture}]
To simplify notations, we temporarily drop the subject index $j$ and the group index $k$ in this proof.
A restatement of this theorem is then as follows:
\begin{description}
\item ~~~~~~~~\textit{suppose there exists an estimator $\tilde Q=\tilde Q_N$ on $[0,B]$ such that $E\{W_1(\tilde Q,Q)~|~Q\}\to0$ as $N=N_{n}\to \infty$ for almost all $Q$ with regard to the measure $\mathcal{Q}$. Then we have $E\{W_1(h_{\tilde Q}, h_{Q})~|~Q\} \to 0 \text{ as } N=N_{n}\to\infty$ for almost all $Q$ with regard to the measure $\mathcal{Q}$.
}
\end{description}
This proof consists of three steps.
In the first step, we assume both $Q$ and $\tilde Q$'s are ordinary distributions with no randomness and prove that $W_1(\tilde Q,Q)\to0$ implies $W_1(h_{\tilde Q},h_Q)\to0$.
In the second step,  we temporarily forget the third-layer ``population model'' \eqref{model:population} and prove that $E\{W_1(\tilde Q,Q)\}\to0$ implies $E\{W_1(h_{\tilde Q},h_Q)\}\to0$ where the expectation is with respect to randomness from the ``measurement model'' \eqref{model:measure} and ``expression model'' \eqref{model:expression}.
In the third step, the third-layer ``population model'' \eqref{model:population} gets involved and we complete this proof.
\medskip

\par\noindent
\textbf{Step 1.}
Suppose $\{\tilde Q\}$ is a sequence of ordinary distributions with no randomness. 
To prove that $W_1(\tilde Q,Q)\to 0$ implies $W_1(h_{\tilde Q},h_Q)\to 0$,
note that $W_1(\tilde Q,Q)\to 0$ is equivalent to $\tilde Q\stackrel{d}{\to}Q$ supplemented with $E\{|\tilde Q|\}\to E\{|Q|\}$ \citep[Section 2.3]{doi:10.1146/annurev-statistics-030718-104938}.
Moreover, it follows from Skorokhod's representation theorem that we can assume $\tilde Q\stackrel{a.s.}{\to}Q$.
To prove $W_1(h_{\tilde Q},h_{Q})\to 0$, it suffices to prove that $h_{\tilde Q}\stackrel{d}{\to}h_{Q}$ and $E\{h_{\tilde Q}\}\to E\{h_{Q}\}$, where the second part follows immediately from $E\{h_{\tilde Q}\}=E\{\tilde Q\}$ and $E\{h_{Q}\}=E\{Q\}$.
For the first part, it follows from $e^{-\lambda}\lambda^x\leq (x/e)^x$ for all $\lambda\in\mathbb{R}^+$ and the dominated convergence theorem that 
\[
h_{\tilde Q}(x)=\int_0^B e^{-\lambda}\frac{\lambda^x}{x!}d\tilde Q(\lambda)\to \int_0^B e^{-\lambda}\frac{\lambda^x}{x!}dQ(\lambda)=h_{Q}(x).
\]
\medskip

\par\noindent
\textbf{Step 2.}
Now suppose $\{\tilde Q\}$ is a sequence of estimators for $Q$ with randomness from the ``measurement model'' \eqref{model:measure} and ``expression model'' \eqref{model:expression}.
Then it can be proved that $W_1(\tilde Q,Q)\stackrel{p}{\to}0$ implies $W_1(h_{\tilde Q},h_Q)\stackrel{p}{\to}0$ based on the result in \textbf{Step 1} and the fact that a sequence converging in probability is equivalent to that its every subsequence has a further subsequence that converges almost surely.
To prove $E\{W_1(\tilde Q,Q)\}\to0$ implies $E\{W_1(h_{\tilde Q},h_Q)\}\to0$, it suffices to verify that $E\{W_1(h_{\tilde Q},h_Q)^2\}$ is bounded which follows immediately from Proposition~\ref{Prop:W1BetweenTwoMixture}, or specifically, 
\[
W_1(h_{\tilde Q},h_Q)\leq E\{h_{\tilde Q}\}+E\{h_Q\}=E\{\tilde Q\}+E\{Q\}\leq 2B.
\]
\medskip

\par\noindent
\textbf{Step 3.}
Suppose $\mathbb{Q}_B$ is a set consisting of all distributions on $[0,B]$.
For any $Q_0\in\mathbb{Q}_B$ with $E\{W_1(\hat{Q},Q)~|~Q=Q_0\}\to 0$, it follows from \textbf{Step 2} that 
\[
E\{W_1(h_{\tilde{Q}},h_{Q})~|~Q=Q_0\}=E\{W_1(h_{\tilde{Q}},h_{Q_0})~|~Q=Q_0\}=E\{W_1(h_{\tilde{Q}},h_{Q_0})\}\to 0,
\]
where the expectation in the last term $E\{W_1(h_{\tilde{Q}},h_{Q_0})\}$ is with respect to the randomness from the ``measurement model'' \eqref{model:measure} and ``expression model'' \eqref{model:expression} only.
Then we can complete this proof by noting that $P(Q\in \mathbb{Q}_B)=1$.
\end{proof}

\begin{proof}[Proof of Theorem~\ref{theorem:W1consistency of NPMLEs}] 
To simplify notations, we temporarily drop the subject index $j$ and the group index $k$ in this proof.
A restatement of this theorem is accordingly as follows:
\begin{description}
\item ~~~~\textit{assume $N=N_{n}\to \infty$ as $n\to\infty$, $r_{i}=r_{i,n}\in [\gamma_0,\gamma_1]$ are uniformly upper and lower bounded by two positive universal constants $\gamma_0,\gamma_1$, and $\mathcal{Q}$ is supported on $[0,B]$. We then have $E\{W_1(\hat Q, Q)~|~Q\} \to 0 \text{ as } n\to\infty$ for almost all $Q$ with regard to the measure $\mathcal{Q}$.}
\end{description}
This proof consists of two steps.
In the first step, we temporarily drop further the third-layer ``population model'' \eqref{model:population} and prove that for each fixed distribution $Q$ supported on $[0,B]$ we have $E\{W_1(\hat Q, Q)\} \to 0 \text{ as } n\to\infty$, where the expectation is with respect to randomness from the ``measurement model'' \eqref{model:measure} and ``expression model'' \eqref{model:expression}.
In the second step, the third-layer ``population model'' \eqref{model:population} gets involved and we complete the proof of the conditional $W_1$-consistency of $\hat Q$.
\medskip

\par\noindent
\textbf{Step 1.}
The first step consists of three substeps.
In the first substep, we prove that the set containing all distributions which are at least $\delta>0$ far from $Q$ can be covered by finite open balls in $W_1$ distance.
In the second substep, with the aid of finite balls, we prove that, with probability converging to 1, no distributions that are at least $\delta$ far away from $Q$ can maximize the likelihood function and hence the $W_1$ distance between $\hat{Q}$ and $Q$ is less than $\delta$.
Then the $W_1$ consistency of $\hat Q$ follows immediately from picking a arbitrarily small $\delta$.
 
\textbf{Step 1(a).}
Let $\mathbb{Q}_B$ be a metric space consisting of all distributions supported on $[0,B]$ with the $W_1$ distance.
For any $\delta>0$, define 
\[
\mathcal{B}_\delta(Q):=\Big\{Q^\prime\in\mathbb{Q}_B: W_1(Q^\prime,Q)<\delta\Big\}
\]
and its complement is denoted by $\mathcal{B}_\delta^c(Q)$.

In the sequel, fix $\epsilon$ to be a small positive number.
Suppose $Q_1,Q_2$ are two distributions on $[0,B]$ and $F_{Q_1},F_{Q_2}$ are their distribution functions.
It then follows from 
\[
e^{-B}\int_0^B|F_{Q_1}-F_{Q_2}|\leq \int_0^B|F_{Q_1}(\lambda)-F_{Q_2}(\lambda)|e^{-\lambda}{\sf d}\lambda\leq\int_0^B|F_{Q_1}-F_{Q_2}|
\]
that Kiefer-Wolfowitz distance \citet[page 51]{chen2017consistency} and Wasserstein-1 distance induce the same topology on $\mathbb{Q}_B$.
Hence it follows from \citet[page 54]{chen2017consistency} that that  
there exists a finite number of distributions $Q_j\in\mathbb{Q}_B,j\in[J],$
such that 
\[
\mathcal{B}^c_\delta(Q)\subset \bigcup_{j\in[J]}\mathcal{B}_\epsilon(Q_j).
\]
Without loss of generality, it is assumed that $Q_j$ is neither a deterministic distribution at $0$ (in other words, degenerate distribution at $0$) nor $Q$ for each $j\in[J]$.

\textbf{Step 1(b).}
Let $Y_{j,\epsilon}(r):=\log\Big\{1+u\left(h_{r,Q}(H_{r,Q}^{-1}(U))/h_{r,\mathcal{B}_\epsilon(Q_j)}(H_{r,Q}^{-1}(U))-1\right)\Big\}$,
where $r\in[\gamma_0,\gamma_1]$, $u\in(0,1)$, $h_{r,Q}(x):=\int_0^Be^{-r\lambda}\frac{(r\lambda)^{x}}{x!}dQ(\lambda)$, $h_{r,\mathcal{B}_\epsilon(Q_j)}(x):=\sup\limits_{Q^\prime\in\mathcal{B}_\epsilon(Q_j)}h_{r,Q^\prime}(x)$, $U$ is a uniform random variable on $[0,1]$, and $H_{r,Q}^{-1}(\cdot)$ is a function such that $H_{r,Q}^{-1}(U)\sim h_{r,Q}$.

\textbf{(i)} We first prove that there exist constants $\epsilon_j>0$ and $c_j>0$ such that for all $\epsilon\leq \epsilon_j$ and $r\in[\gamma_0,\gamma_1]$ we have $E\{Y_{j,\epsilon}(r)\}\geq c_j>0$.

Note that $Y_{j,\epsilon}(r)\geq\log(1-u)$ and $\lim\limits_{\epsilon\to0^+}Y_{j,\epsilon}(r)=Y_{j,0^+}(r)$ almost surely, where 
$$
Y_{j,0^+}(r):=\log\Big\{1+u\left(h_{r,Q}(H_{r,Q}^{-1}(U))/h_{r,Q_j}(H_{r,Q}^{-1}(U))-1\right)\Big\}.
$$
Since $Y_{j,\epsilon}(r)$ is monotonically decreasing with respect to $\epsilon$, it follows from the monotone convergence theorem that 
$\lim\limits_{\epsilon\to 0^+}E\{Y_{j,\epsilon}(r)\}=E\{Y_{j,0^+}(r)\}$ for each $r\in[\gamma_0,\gamma_1]$.
Moreover, it follows from Lemma 2.5 in \citet{chen2017consistency} that $E\{Y_{j,0^+}(r)\}>0$ for each $r\in[\gamma_0,\gamma_1]$.
Since $E\{Y_{j,0^+}(r)\}$ is a continuous function with respect to $r$ and $r\in[\gamma_0,\gamma_1]$, there exists a positive constant $c_j$ such that 
$E\{Y_{j,0^+}(r)\}\geq 2c_j$ for all $r\in[\gamma_0,\gamma_1]$.
Furthermore, since $E\{Y_{j,\epsilon}(r)\}$ is a monotonically decreasing function with respect to $\epsilon$, then it follows from Dini's theorem that $E\{Y_{j,\epsilon}(r)\}$ uniformly converges to $E\{Y_{j,0^+}(r)\}$ as $\epsilon\to0^+$ on $r\in[\gamma_0,\gamma_1]$ and hence there exists a $\epsilon_j$ which doesn't depend on $r$ such that for all $\epsilon\leq \epsilon_j$ and $r\in[\gamma_0,\gamma_1]$ we have $|E\{Y_{j,\epsilon}(r)\}-E\{Y_{j,0^+}(r)\}|\leq c_j$.
Hence $E\{Y_{j,\epsilon}(r)\}\geq c_j>0$ for all $\epsilon\leq\epsilon_j$ and $r\in[\gamma_0,\gamma_1]$.
Replacing $r$ by $r_{i,n}$ for all $\epsilon\leq\epsilon_j$ we have
\[
E\{Y_{j,\epsilon}(r_{i,n})\}\geq c_j>0\text{ for all }i\in[N_n].
\]
In the following arguments, set $\epsilon=\min\limits_{j\in [J]}\{\epsilon_j\}$ and let $\mathcal{B}_j:=\mathcal{B}_\epsilon(Q_j)$ for simplicity.

\textbf{(ii)} We then prove that there exists a constant $C_j$ such that $\var\{Y_{j,\epsilon}(r)\}\leq C_j<\infty$ for all $r\in[\gamma_0,\gamma_1]$.

Since $h_{r,\mathcal{B}_j}(x)\geq h_{r,Q_j}(x)$ and $\log\{1+u(h_{r,Q}(x)/h_{r,\mathcal{B}_j}(x)-1)\}\geq \log (1-u)$, we have $E\{Y^2_{j,\epsilon}(r)\}<\infty$ or equivalently
\[
E\{\left(\log\{1+u(h_{r,Q}(X_{r})/h_{r,\mathcal{B}_j}(X_{r})-1)\}\right)^2\}<\infty,
\]
where $X_{r}:=H_{r,Q}^{-1}(U)\sim h_{r,Q}$,
as long as  $E\{(h_{r,Q}(X_{r})/h_{r,Q_j}(X_{r})-1)^2\}<\infty$, or simply, 
\[
E\{(h_{r,Q}(X_{r})/h_{r,Q_j}(X_{r}))^2\}<\infty.
\]
To prove it, note that $Q_j$ is not a deterministic distribution at $0$ and hence there exist $\lambda_j\in(0,B]$ such that $F_{Q_j}(\lambda_j)<1$.
Then we have
\[
h_{r,Q}(x)\leq e^{-B\gamma_1}\frac{(B\gamma_1)^x}{x!}\text{ and } h_{r,Q_j}(x)\geq \left(1-F_{Q_j}(\lambda_j)\right)e^{-\gamma_0\lambda_j}\frac{(\gamma_0\lambda_j)^x}{x!}
\]
for sufficiently large $x$,
and hence
\[
\frac{h_{r,Q}(x)}{h_{r,Q_j}(x)}\leq \frac{e^{\gamma_0\lambda_j-B\gamma_1}}{\left(1-F_{Q_j}(\lambda_j)\right)}\left(\frac{B\gamma_1}{\gamma_0\lambda_j}\right)^x.
\]
Then $E\{(h_{r,Q}(X_{r})/h_{r,Q_j}(X_{r}))^2\}<\infty$ follows immediately from the existence of the moment generating function of Poisson distribution.
Since $E\{Y^2_{j,\epsilon}(r)\}$ is a continuous function with respect to $r$ and $r\in[\gamma_0,\gamma_1]$, then there exists a uniform constant $C_j$ such that 
\[
E\{Y^2_{j,\epsilon}(r)\}\leq C_j
\] 
for all $r\in[\gamma_0,\gamma_1]$. Replacing $r$ by $r_{i,n}$ it follows that $E\{Y^2_{j,\epsilon}(r_{i,n})\}\leq C_j$ for all $i\in[N_n]$.

\textbf{(iii)}
Suppose $X_{i,n},i\in[N_n]$ is a sequence of independent random variables with $X_{i,n}\sim h_{r_{i,n},Q}$.
Define $Z_{ij,n}:=\log\Big\{1+u\left(h_{r_{i,n},Q}(X_{i,n})/h_{r_{i,n},\mathcal{B}_\epsilon(Q_j)}(X_{i,n})-1\right)\Big\}$.
Note that $Z_{ij,n}\stackrel{d}{=}Y_{j,\epsilon}(r_{i,n})$.
Built on (i) and (ii), we have $E\{Z_{ij,n}\}\geq c_j>0$ and $\var\{Z_{ij,n}\}\leq C_j<\infty$ for $i\in[N_n]$ and $j\in[J]$.
Therefore,
\[
\var\left\{\sum\limits_{i\in[N_n]}Z_{ij,n}\right\}\leq N_nC_j
\]
 and hence $\var\left\{\sum\limits_{i\in[N_n]}Z_{ij,n}\right\}/N_n^2\to0$ as $n\to\infty$.
Then it follows from Markov's inequality that
\[
\frac{1}{N_n}\sum_{i\in[N_n]}\left(Z_{ij,n}-E\{Z_{ij,n}\}\right)\stackrel{p}{\to}0
\]
for each $j\in[J]$.
In other words, for any positive number $\xi$ and events 
\[
A_{j,n}:=\left|\frac{1}{N_n}\sum_{i\in[N_n]}\left(Z_{ij,n}-E\{Z_{ij,n}\}\right)\right|\leq \xi,
\]
we have $\lim\limits_{n\to\infty}P(A_{j,n})=1$.
Combined with $E\{Z_{ij,n}\}\geq c_j$, we have, under $A_{j,n}$ with $\xi\leq c_j/2$,
\[
\frac{1}{N_n}\sum_{i\in[N_n]}\log\{1+u\left(h_{r_{i,n},Q}(X_{i,n})/h_{r_{i,n},\mathcal{B}_j}(X_{i,n})-1\right)\}> c_j/2,
\]
and hence
\begin{eqnarray*}
0&<&\sum_{i\in[N_n]}\log\{1+u\left(h_{r_{i,n},Q}(X_{i,n})/h_{r_{i,n},\mathcal{B}_j}(X_{i,n})-1\right)\}\\
&\leq&\inf_{Q^\prime\in\mathcal{B}_j}\sum_{i\in[N_n]}\log\{1+u\left(h_{r_{i,n},Q}(X_{i,n})/h_{r_{i,n},Q^\prime}(X_{i,n})-1\right)\}
\end{eqnarray*}
for each $j\in[J]$.
Noting that $\mathcal{B}^c_\delta(Q)\subset \bigcup_{j=1}^J\mathcal{B}_j$, the last display implies that under events $A_{n}:=\bigcap\limits_{j\in[J]}A_{j,n}$ with $\xi\leq \min\limits_{j\in[J]}c_j/2$ we have
\[
0<\inf_{Q^\prime\notin\mathcal{B}_\delta(Q)}\sum_{i\in[N_n]}\log\{1+u\left(h_{r_{i,n},Q}(X_{i,n})/h_{r_{i,n},Q^\prime}(X_{i,n})-1\right)\},
\]
or equivalently,
\[
l_{n}(uQ+(1-u)Q^\prime)>l_{n}(Q^\prime)
\]
for all $Q^\prime\in\mathcal{B}^c_\delta(Q)$, where for each $Q^\prime\in\mathbb{Q}_B$
\[
l_{n}(Q^\prime):=\sum_{i\in[N_n]}\log\int_0^Be^{-r_{i,n}\lambda}\frac{(r_{i,n}\lambda)^{X_{i,n}}}{X_{i,n}!}dQ^\prime(\lambda).
\]
Therefore, under events $A_{n}$ with $\xi\leq \min\limits_{j\in[J]}c_j/2$, the maximum likelihood estimator $\hat{Q}$ must belong to $\mathcal{B}_\delta(Q)$ and hence $W_1(\hat{Q},Q)\leq \delta$.
Since $P(A_n)\to1$, we have $P(W_1(\hat{Q},Q)\leq \delta)\to1$, or equivalently, $W_1(\hat{Q},Q)\stackrel{p}{\to}0$ as $n\to\infty$.
It further follows from $W_1(\hat{Q},Q)\leq B$ that $E\{W_1(\hat{Q},Q)\}\to0$ as $n\to\infty$.
\medskip

\par\noindent
\textbf{Step 2.}
For any $Q_0\in\mathbb{Q}_B$, it follows from \textbf{Step 1} that 
\[
E\{W_1(\hat{Q},Q)~|~Q=Q_0\}=E\{W_1(\hat{Q},Q_0)~|~Q=Q_0\}=E\{W_1(\hat{Q},Q_0)\}\to 0,
\]
where the expectation in the last term $E\{W_1(\hat{Q},Q_0)\}$ is with respect to the randomness from the ``measurement model'' \eqref{model:measure} and ``expression model'' \eqref{model:expression} only.
Then it follows from $P(Q\in\mathbb{Q}_B)=1$ that $E\{W_1(\hat{Q},Q)~|~Q\}\to 0$ for almost all $Q$ with regard to the measure $\mathcal{Q}$.
\end{proof}

\begin{proof}[Proof of Theorem~\ref{theorem:size_validity}] 
By the construction of the population model \eqref{model:population}, under the $H_0$ in \eqref{eq:H0} $Q_j^{(k)}$'s are independent and identically distributed.
Furthermore, since the sets $\{r_{ij}^{(k)}, i\in[N]\}$ are invariant with respect to $j\in[n_k]$ and $k\in[K]$ for each $i\in[N]$, the random vectors $(X_{1j}^{(k)},\ldots,X_{Nj}^{(k)})^\top$'s are independent and identically distributed.
Therefore, $\tilde Q_j^{(k)}$'s are independent and identically distributed.
As a consequence, $\tilde F$ is uniformly distributed over 
\[
\mathcal{\tilde{F}}^{\pi}:=\Big\{\tilde F^{\pi}:\pi\in\text{ all permutations of }[n]\to[n]\Big\}
\]
and hence 
 \[
 P(\tilde F>\text{ the $1-\alpha$ quantile of $\mathcal{\tilde F}^{\pi}$ }| H_0) \leq \alpha.
 \]
 Note that the event 
 \[
 \tilde F>\text{ the $1-\alpha$ quantile of $\mathcal{\tilde F}^{\pi}$ }
 \]
 is identical to the event
 \[
 P(\tilde F^{\pi}<\tilde F~|~\tilde Q_j^{(k)}\text{'s})\geq 1-\alpha,\text{ where probability here is with respect to the random permutation $\pi$,}
 \]
 and hence
$P(\tilde T_\alpha=1 | H_0) \leq \alpha$.
The proof of $P(\tilde T_{h,\alpha}=1 | H_0) \leq \alpha$ is analogous and hence omitted.
\end{proof}

\begin{proposition}\label{Prop:W1BetweenTwoMixture}
Suppose $P$ and $Q$ are two distributions supported on $[0,\infty)$.
Then $W_1(P,Q)\leq E\{P\}+E\{Q\}$ and $W_1(P,Q)\geq |E\{P\}-E\{Q\}|$.
\end{proposition}
\begin{proof}
Denote distribution functions of $P$ and $Q$ by $F_{P}$ and $F_{Q}$ respectively.
Then it follows from the triangle inequality that
$
W_1(P,Q)
=\int_0^\infty |(1-F_{P})-(1-F_{Q})|
\leq\int_0^\infty \left(1-F_{P}\right)+\int_0^\infty \left(1-F_{Q}\right)
=E\{P\}+E\{Q\}
$
and
$
W_1(P,Q)
=\int_0^\infty |(1-F_{P})-(1-F_{Q})|
\geq|E\{P\}-E\{Q\}|.
$
\end{proof}

Define
\begin{eqnarray}\label{eq:DefinitionOfF}
F:=\Big(SS_T-\sum\limits_{k\in[K]}SS_{k}\Big)\Big/\sum\limits_{k\in[K]}SS_{k},
\end{eqnarray}
where 
\begin{align*}
SS_T:=\frac{1}{n}\sum\limits_{k_1,k_2 \in [K]}\sum\limits_{j_1\in [n_{k_1}],j_2\in [n_{k_2}]}W_1(Q_{j_1}^{(k_1)},Q_{j_2}^{(k_2)})^2~~~{\rm and}~~~
SS_{k}:=\frac{1}{n_k}\sum\limits_{j_1,j_2\in [n_k]}W_1(Q_{j_1}^{(k)},Q_{j_2}^{(k)})^2.
\end{align*}
For any permutation $\pi: [n]\to [n]$, define
\begin{eqnarray}\label{eq:DefinitionOfFpi}
F^\pi:=\Big(SS_T-\sum\limits_{k\in[K]}SS_{k}^\pi\Big)\Big/\sum\limits_{k\in[K]}SS_{k}^\pi,
\end{eqnarray}
where 
\begin{eqnarray*}
SS_k^\pi&:=&\frac{1}{n_{k}}\sum\limits_{j_1,j_2\in [n_{k}]}W_1\Big(Q_{\Pi^{j_1,k}_1}^{(\Pi^{{j_1,k}}_2)},Q_{\Pi^{j_2,k}_1}^{(\Pi^{j_2,k}_2)}\Big)^2\text{ for each }k\in[K].
\end{eqnarray*}
\medskip

\par\noindent
\begin{proof}[Proof of Theorem~\ref{theorem:TestConsistency}(a)]
Throughout this proof, unless conditioning on certain events, the probability refers to randomness from all three layers as well as the permutation. 
Without loss of generality, it is assumed that $\mathcal{Q}_k$ is non-degenerate for each $k\in[K]$.
Otherwise, the proof is analogous and omitted.


This proof consists of five steps. In the first step, we prove that if the following Equation \eqref{eq:ConsistencyGoal} is true,
\begin{eqnarray}
\lim_{n\to\infty}P\Big(\tilde F>\tilde F^\pi|H_1\Big)=1,
\label{eq:ConsistencyGoal}
\end{eqnarray}
then $\lim\limits_{n\to\infty}P(\tilde T_\alpha=1~|~H_1)=1$ for any $\alpha\in(0,1)$. The rest four steps are devoted to proving Equation \eqref{eq:ConsistencyGoal}. 
Note that $\tilde{F}-\tilde F^\pi=(\tilde{F}-F)+(F-F^\pi)+(F^\pi-\tilde{F}^\pi)$, where $F$ is defined in \eqref{eq:DefinitionOfF} and $F^\pi$ is defined in \eqref{eq:DefinitionOfFpi}.
The second step proves that $\tilde{F}-F\stackrel{p}{\to}0$ as $n\to\infty$.  
The third step proves that $F^\pi-\tilde{F}^\pi\stackrel{p}{\to}0$ as $n\to\infty$. 
The fourth step proves that $F-F^\pi\stackrel{p}{\to}$ some strictly positive constant. In the fifth step, we combine results in Steps 2-4 to prove \eqref{eq:ConsistencyGoal} and hence finish the proof of Theorem \ref{theorem:TestConsistency}(a).

In the following the notion $H_1$ in the probability is abandoned as long as no confusion is possible.

\medskip

\par\noindent
\textbf{Step 1.}
Note that 
\[
P(\tilde T_\alpha=1)=P\Big\{P(\tilde F^{\pi}<\tilde F~|~\tilde Q_j^{(k)}\text{'s})\geq 1-\alpha\Big\}=1-P\Big\{P(\tilde F^{\pi}<\tilde F~|~\tilde Q_j^{(k)}\text{'s})< 1-\alpha\Big\},
\]
\begin{eqnarray*}
P\Big\{P(\tilde F^{\pi}<\tilde F~|~\tilde Q_j^{(k)}\text{'s})< 1-\alpha\Big\}
=P\Big\{P(\tilde F^{\pi}\geq \tilde F~|~\tilde Q_j^{(k)}\text{'s})> \alpha)\Big\}
\leq \frac{E\Big\{P(\tilde F^{\pi}\geq \tilde F~|~\tilde Q_j^{(k)}\text{'s})\Big\}}{\alpha},
\end{eqnarray*}
and
\[
E\Big\{P(\tilde F^{\pi}\geq \tilde F~|~\tilde Q_j^{(k)}\text{'s})\Big\}
=P(\tilde F^{\pi}\geq \tilde F)
=1-P(\tilde F^{\pi}<\tilde F).
\]
Therefore, we have $\lim\limits_{n\to\infty}P\{P(\tilde F^{\pi}<\tilde F~|~\tilde Q_j^{(k)}\text{'s})< 1-\alpha\}=0$ and hence $\lim\limits_{n\to\infty}P(\tilde T_\alpha=1)=1$ as long as \eqref{eq:ConsistencyGoal} holds.
\medskip

\par\noindent 
\textbf{Step 2.}
In this step, we prove that $\tilde{F}-F\stackrel{p}{\to}0$ as $n\to\infty$, where the probability here refers to randomness from all three layers.

Note that
\begin{eqnarray*}
 \left|\tilde{F}-F\right|
=\left|\frac{\tilde{SS}_T-\sum\limits_{k\in[K]}\tilde{SS}_k}{\sum\limits_{k\in[K]}\tilde{SS}_k}-\frac{SS_T-\sum\limits_{k\in[K]}SS_k}{\sum\limits_{k\in[K]}SS_k}\right|
=\left|\frac{SS_T}{\sum\limits_{k\in[K]}SS_k}-\frac{\tilde{SS}_T}{\sum\limits_{k\in[K]}SS_k}+\frac{\tilde{SS}_T}{\sum\limits_{k\in[K]}SS_k}-\frac{\tilde{SS}_T}{\sum\limits_{k\in[K]}\tilde{SS}_k}\right|,
\end{eqnarray*}
where $SS_k$ and $SS_T$ are defined in \eqref{eq:DefinitionOfF}.
It then follows from the triangle inequality that
\begin{eqnarray}
\left|\tilde{F}-F\right|
&\leq&\left|\frac{SS_T-\tilde{SS}_T}{\sum\limits_{k\in[K]}SS_k}\right|+\left|\frac{\tilde{SS}_T}{\sum\limits_{k\in[K]}\tilde{SS}_k}\right|\left|\frac{\sum\limits_{k\in[K]}(\tilde{SS}_k-SS_k)}{\sum\limits_{k\in[K]}SS_k}\right|\notag\\
&\leq&\left|\frac{SS_T-\tilde{SS}_T}{\sum\limits_{k\in[K]}SS_k}\right|+\left|\frac{B^2}{\frac{1}{n-1}\sum\limits_{k\in[K]}\tilde{SS}_k}\right|\left|\frac{\sum\limits_{k\in[K]}(\tilde{SS}_k-SS_k)}{\sum\limits_{k\in[K]}SS_k}\right|\label{eq:DiffInF-TildeF},
\end{eqnarray}
where $\frac{1}{n-1}\tilde{SS}_T\leq B^2$ follows from $W_1(\tilde Q_{j_1}^{(k_1)},\tilde Q_{j_2}^{(k_2)})\leq B^2$ for each $j_1,j_2,k_1,k_2$.

\textbf{Step 2(a).}
We first prove that $\left|\frac{\tilde{SS}_k}{n_k-1}-\frac{SS_{k}}{n_k-1}\right|\stackrel{p}{\to}0$ and $\frac{1}{n-1}\left|\tilde{SS}_T-SS_T\right|\stackrel{p}{\to}0$.

It follows from the triangle inequality that for each $k\in[K]$
\begin{align*}
\left|\frac{\tilde{SS}_k}{n_k-1}-\frac{SS_{k}}{n_k-1}\right|
=&\frac{1}{n_k(n_k-1)}\left|\sum\limits_{j_1,j_2\in [n_k]}\left(W_1(\tilde{Q}_{j_1}^{(k)},\tilde{Q}_{j_2}^{(k)})^2-W_1(Q_{j_1}^{(k)},Q_{j_2}^{(k)})^2\right)\right|\\
\leq&\frac{2B}{n_k(n_k-1)}\sum\limits_{j_1,j_2\in [n_k]}\left(W_1(\tilde{Q}_{j_1}^{(k)},Q_{j_1}^{(k)})+W_1(\tilde{Q}_{j_2}^{(k)},Q_{j_2}^{(k)})\right)\cdot I(j_1\neq j_2)\\
=&\frac{4B}{n_k}\sum\limits_{j\in [n_k]}W_1(\tilde{Q}_{j}^{(k)},Q_{j}^{(k)})
\end{align*}
and analogously we have
\[
\left|\frac{\tilde{SS}_{T}}{n-1}-\frac{SS_{T}}{n-1}\right|\leq\frac{4B}{n}\sum\limits_{k\in[K]}\sum\limits_{j\in[n_k]}W_1(\tilde{Q}^{(k)}_j,Q^{(k)}_j).
\]
Therefore,
\begin{eqnarray*}
E\left\{\left|\frac{\tilde{SS}_k}{n_k-1}-\frac{SS_{k}}{n_k-1}\right|\right\}
\leq4B\cdot \frac{1}{n_k}\sum\limits_{j\in [n_k]}E\left\{W_1(\tilde{Q}_{j}^{(k)},Q_{j}^{(k)})\right\}
=4B\cdot E\left\{W_1(\tilde{Q}_{1}^{(k)},Q_{1}^{(k)})\right\},
\end{eqnarray*}
where the last equality follows from Assumption \ref{ass:balance}.
Noting that $W_1(\tilde{Q}_{1}^{(k)},Q_{1}^{(k)})\leq B$ and 
\[
E\Big\{W_1(\tilde{Q}_{1}^{(k)},Q_{1}^{(k)})~|~Q_{1}^{(k)}\Big\}\stackrel{a.s.}{\to}0,
\]
we have using
\[
E\left\{W_1(\tilde{Q}_{1}^{(k)},Q_{1}^{(k)})\right\}=E\left[E\{W_1(\tilde{Q}_{1}^{(k)},Q_{1}^{(k)})|Q_{1}^{(k)}\}\right]
\]
that 
\[
E\left\{W_1(\tilde{Q}_{1}^{(k)},Q_{1}^{(k)})\right\}\to 0 \text{ as }n\to\infty.
\]
Analogously, we have $E\left\{\frac{\left|\tilde{SS}_T-SS_T\right|}{n-1}\right\}\to0$ as $n\to\infty$.
Therefore, in \eqref{eq:DiffInF-TildeF} we have $\frac{\left|\tilde{SS}_T-SS_T\right|}{n-1}\stackrel{p}{\to}0$ and $\sum\limits_{k\in[K]}\frac{\left|\tilde{SS}_{k}-SS_{k}\right|}{n_k-1}\stackrel{p}{\to}0$.

\textbf{Step 2(b).}
We then prove that $\frac{1}{n_k-1}SS_{k}\stackrel{p}{\to}$ some strictly positive constant.

It follows from $|W_1(Q_{j_1}^{(k_1)},Q_{j_2}^{(k_2)})|\leq B$ and the strong law of large numbers for U-statistics 
\citet[Chapter 5.4]{Serfling1980} 
that
\begin{eqnarray}\label{eq:LimitOfSS_k}
\frac{1}{n_k-1}SS_{k}
&\stackrel{a.s.}{\to}&E\{W_1(Q^{(k)}_1,Q^{(k)}_2)^2\}.
\end{eqnarray}
If $E\{W_1(Q^{(k)}_1,Q^{(k)}_2)^2\}=0$, then $W_1(Q^{(k)}_1,Q^{(k)}_2)=0$ almost surely and hence $Q^{(k)}_1=Q^{(k)}_2$ almost surely, which implies that 
$\mathcal{Q}_k$ is degenerate.

\textbf{Step 2(c).}
Built on Step 2(a) and Step 2(b), we have
\[
\frac{1}{n_k-1}\tilde{SS}_{k}=\frac{1}{n_k-1}(\tilde{SS}_{k}-SS_{k})+\frac{1}{n_k-1}SS_{k}\stackrel{p}{\to}E\Big\{W_1(Q^{(k)}_1,Q^{(k)}_2)^2\Big\}.
\]
Therefore, Slutsky's theorem guarantees $|\tilde F-F|\stackrel{p}{\to}0$, where the  probability here refers to randomness from all three layers.
\medskip

\par\noindent
\textbf{Step 3.}
In this step, we prove that $F^\pi-\tilde{F}^\pi\stackrel{p}{\to}0$ as $n\to\infty$, where the probability here refers to randomness from all three layers as well as the permutation.

To prove it, it suffices to show that 
$\frac{1}{n_k-1}\left|\tilde{SS}^\pi_{k}-SS_{k}^\pi\right|\stackrel{p}{\to}0$ and $\frac{1}{n_k-1}SS_{k}^\pi\stackrel{p}{\to}$ some strictly positive constant since 
\begin{eqnarray}
 \left|\tilde{F}^\pi-F^\pi\right|
\leq\left|\frac{SS_T-\tilde{SS}_T}{\sum\limits_{k\in[K]}SS^\pi_k}\right|+\left|\frac{B^2}{\frac{1}{n-1}\sum\limits_{k\in[K]}\tilde{SS}^\pi_k}\right|\left|\frac{\sum\limits_{k\in[K]}(\tilde{SS}^\pi_k-SS^\pi_k)}{\sum\limits_{k\in[K]}SS^\pi_k}\right|\label{eq:DiffInF^pi-TildeF^pi},
\end{eqnarray}
where $SS^\pi_k$ and $SS^\pi_T$ are defined in \eqref{eq:DefinitionOfFpi}.

\textbf{Step 3(a).}
We first prove $\frac{1}{n_k-1}\left|\tilde{SS}^\pi_{k}-SS_{k}^\pi\right|\stackrel{p}{\to}0$.

To prove it, note that with similar arguments in Step 2(a) we have
\begin{eqnarray*}
E\left\{\left|\frac{\tilde{SS}^\pi_k}{n_k-1}-\frac{SS^\pi_{k}}{n_k-1}\right|\right\}
&\leq&4B\cdot \frac{1}{n_k}\sum\limits_{j\in [n_k]}E\left\{W_1(\tilde{Q}_{\Pi^{j,k}_1}^{(\Pi^{j,k}_2)},Q_{\Pi^{j,k}_1}^{(\Pi^{j,k}_2)})\right\}.
\end{eqnarray*}
Let $n^\pi_{k,k^\prime}$ represent the number of indices exchanged between group $k$ and $k^\prime$ after the specific permutation $\pi$.
Note that $n^\pi_{k,k^\prime}=n^\pi_{k^\prime,k}$ and $\sum\limits_{k^\prime}n^\pi_{k,k^\prime}=n_k$.
Then, with the aid of the notations $\{n^\pi_{k,k^\prime}\}$, it follows that 
\begin{eqnarray*}
\sum\limits_{j\in [n_k]}E\left\{W_1(\tilde{Q}_{\Pi^{j,k}_1}^{(\Pi^{j,k}_2)},Q_{\Pi^{j,k}_1}^{(\Pi^{j,k}_2)})\Big|\pi\right\}
&=&\sum_{k^\prime\in[K]}\sum\limits_{j\in[n^\pi_{k,k^\prime}]}E\left\{W_1(\tilde{Q}_{j}^{(k^\prime)},Q_{j}^{(k^\prime)})\Big|\pi\right\}\\
&=&\sum_{k^\prime\in[K]}n^\pi_{k,k^\prime}E\left\{W_1(\tilde{Q}_{1}^{(k^\prime)},Q_{1}^{(k^\prime)})\right\}
\end{eqnarray*}
and hence 
\begin{eqnarray*}
\sum\limits_{j\in [n_k]}E\left\{\frac{W_1(\tilde{Q}_{\Pi^{j,k}_1}^{(\Pi^{j,k}_2)},Q_{\Pi^{j,k}_1}^{(\Pi^{j,k}_2)})}{n_k}\right\}
=\sum_{k^\prime\in[K]}E\left\{W_1(\tilde{Q}_{1}^{(k^\prime)},Q_{1}^{(k^\prime)})\right\}E\Big\{\frac{n^\pi_{k,k^\prime}}{n_k}\Big\}
=\sum_{k^\prime\in[K]}E\left\{\frac{W_1(\tilde{Q}_{1}^{(k^\prime)},Q_{1}^{(k^\prime)})}{K}\right\}.
\end{eqnarray*}
Therefore, it follows from $E\left\{W_1(\tilde{Q}_{1}^{(k^\prime)},Q_{1}^{(k^\prime)})\right\}\to 0$ for each $k\in[K]$ that $E\left\{\left|\frac{\tilde{SS}^\pi_k}{n_k-1}-\frac{SS^\pi_{k}}{n_k-1}\right|\right\}\to0$.

\textbf{Step 3(b).}
We then prove $\frac{1}{n_k-1}SS_{k}^\pi\stackrel{p}{\to}$ some strictly positive constant.

To prove it, let $X\stackrel{d}{=}Y$ denote that the two random variables $X,Y$ are identically distributed and note that 
\begin{eqnarray}
\frac{SS_{k}^\pi}{n_k-1}
\stackrel{d}{=}\sum\limits_{k^\prime\in[K]}\underbrace{\sum\limits_{j_1,j_2\in[n^\pi_{k,k^\prime}]}\frac{W_1(Q_{j_1}^{(k^\prime)},Q_{j_2}^{(k^\prime)})^2}{n_k(n_k-1)}}_{(a)}
+\sum\limits_{k_1^\prime\neq k_2^\prime\in[K]}\sum\limits_{j_1\in[n^\pi_{k,k_1^\prime}],j_2\in[n^\pi_{k,k_2^\prime}]}\frac{W_1(Q_{j_1}^{(k_1^\prime)},Q_{j_2}^{(k_2^\prime)})^2}{n_k(n_k-1)}.\label{eq:SS_k^pi/(n_k-1)}
\end{eqnarray}

(i) We first prove that the variance of (a) converges to $0$. 

Note that the variance of (a) equals to
\begin{eqnarray*}
\var\left\{E\left\{\sum\limits_{j_1,j_2\in[n^\pi_{k,k^\prime}]}\frac{W_1(Q_{j_1}^{(k^\prime)},Q_{j_2}^{(k^\prime)})^2}{n_k(n_k-1)}\Big|\pi\right\}\right\}
+E\left\{\var\left\{\sum\limits_{j_1,j_2\in[n^\pi_{k,k^\prime}]}\frac{W_1(Q_{j_1}^{(k^\prime)},Q_{j_2}^{(k^\prime)})^2}{n_k(n_k-1)}\Big|\pi\right\}\right\},
\end{eqnarray*}
where the first term equals to
\[
\var\left\{\frac{n^\pi_{k,k^\prime}(n^\pi_{k,k^\prime}-1)}{n_k(n_k-1)}E\{W_1(Q_{1}^{(k^\prime)},Q_{2}^{(k^\prime)})^2\}\right\}
\leq B^4\cdot\var\left\{\frac{n^\pi_{k,k^\prime}(n^\pi_{k,k^\prime}-1)}{n_k(n_k-1)}\right\}.
\]
For the second term, note that 
\begin{eqnarray*}
& &\var\left\{\sum\limits_{j_1\neq j_2\in[n^\pi_{k,k^\prime}]}W_1(Q_{j_1}^{(k^\prime)},Q_{j_2}^{(k^\prime)})^2~\Big|~\pi\right\}\\
&=&\sum\limits_{j_1\neq j_2\in[n^\pi_{k,k^\prime}]}\sum\limits_{j^\prime_1\neq j^\prime_2\in[n^\pi_{k,k^\prime}]}
\cov\left\{W_1(Q_{j_1}^{(k^\prime)},Q_{j_2}^{(k^\prime)})^2,W_1(Q_{j^\prime_1}^{(k^\prime)},Q_{j^\prime_2}^{(k^\prime)})^2\right\}\\
&\leq&B^4\sum\limits_{j_1\neq j_2\in[n^\pi_{k,k^\prime}]}\sum\limits_{j^\prime_1\neq j^\prime_2\in[n^\pi_{k,k^\prime}]} I\left(\text{At least two of $\{j_1,j_2,j_1^\prime,j_2^\prime\}$ are identical.}\right)\\
&\leq&B^4\cdot 2n^\pi_{k,k^\prime}(n^\pi_{k,k^\prime}-1)^2.
\end{eqnarray*}
Therefore, the variance of (a) is upper bounded by
\begin{eqnarray*}
\frac{2B^4}{\left(n_k(n_k-1)\right)^2}\Big[\var\{n^\pi_{k,k^\prime}(n^\pi_{k,k^\prime}-1)\}+E\left\{(n^\pi_{k,k^\prime})^3\right\}\Big]
\to0,
\end{eqnarray*}
where the convergence follows from $n^\pi_{k,k^\prime}/n_k\stackrel{p}{\to}1/K$, $n^\pi_{k,k^\prime}/n_k\leq 1$, the dominated convergence theorem such that
\[
E\Big\{(n^\pi_{k,k^\prime})^3\Big\}\Big/\left(n_k(n_k-1)\right)^2 \to E\{0\}=0,
\]
and
\begin{eqnarray*}
\frac{\var\{n^\pi_{k,k^\prime}(n^\pi_{k,k^\prime}-1)\}}{\left(n_k(n_k-1)\right)^2}
=E\left\{\left(\frac{n^\pi_{k,k^\prime}(n^\pi_{k,k^\prime}-1)}{n_k(n_k-1)}\right)^2\right\}
-\left(E\left\{\frac{n^\pi_{k,k^\prime}(n^\pi_{k,k^\prime}-1)}{n_k(n_k-1)}\right\}\right)^2
\to\frac{1}{K^4}-\frac{1}{K^4}
=0.
\end{eqnarray*}
To prove $n^\pi_{k,k^\prime}/n_k\stackrel{p}{\to}1/K$, note that $E\{n^\pi_{k,k^\prime}/n_k\}=1/K$ and $\var\{n^\pi_{k,k^\prime}/n_k\}=\var\{n^\pi_{1,1}/n_1\}=\frac{n_1^2(n-n_1)^2}{n_1^2n^2(n-1)}\to0$ (cf. \citet[page 460]{chapuy2007random}).

(ii)
We then prove that the expectation of (a) converges to $E\left\{W_1(Q_{1}^{(k^\prime)},Q_{2}^{(k^\prime)})^2\right\}/K^2$. For this, we have
\begin{align*}
E\{(a)\}
&=E\left\{\frac{1}{n_k(n_k-1)}\sum\limits_{j_1,j_2\in[n^\pi_{k,k^\prime}]}E\left\{W_1(Q_{j_1}^{(k^\prime)},Q_{j_2}^{(k^\prime)})^2|\pi\right\}\right\}\\
&=E\left\{\frac{n^\pi_{k,k^\prime}(n^\pi_{k,k^\prime}-1)}{n_k(n_k-1)}E\left\{W_1(Q_{1}^{(k^\prime)},Q_{2}^{(k^\prime)})^2\right\}\right\},
\end{align*}
which converges to $E\left\{W_1(Q_{1}^{(k^\prime)},Q_{2}^{(k^\prime)})^2\right\}/K^2$ by the dominated convergence theorem.

(iii)
Built on (i) and (ii), it follows from Markov's inequality that $(a)\stackrel{p}{\to}E\left\{W_1(Q_{1}^{(k^\prime)},Q_{2}^{(k^\prime)})^2\right\}/K^2$.
Analogously, we can prove that the second term in \eqref{eq:SS_k^pi/(n_k-1)} converges to a constant in probability, i.e.,
\[
\sum\limits_{k_1^\prime\neq k_2^\prime\in[K]}\frac{1}{n_k(n_k-1)}\sum\limits_{j_1\in[n^\pi_{k,k_1^\prime}],j_2\in[n^\pi_{k,k_2^\prime}]}W_1(Q_{j_1}^{(k_1^\prime)},Q_{j_2}^{(k_2^\prime)})^2\stackrel{p}{\to}\sum\limits_{k_1^\prime\neq k_2^\prime\in[K]}\frac{E\{W_1(Q_{1}^{(k_1^\prime)},Q_{1}^{(k_2^\prime)})^2\}}{K^2}.
\]
As a result, we have 
\begin{eqnarray}\label{eq:LimitOfSS_k^pi}
\frac{1}{n_k-1}SS_{k}^\pi
\stackrel{p}{\to}\sum_{k^\prime\in[K]}\frac{E\left\{W_1(Q_{1}^{(k^\prime)},Q_{2}^{(k^\prime)})^2\right\}}{K^2}+\sum\limits_{k_1^\prime\neq k_2^\prime\in[K]}\frac{E\{W_1(Q_{1}^{(k_1^\prime)},Q_{1}^{(k_2^\prime)})^2\}}{K^2}>0.
\end{eqnarray}
\medskip

\par\noindent
\textbf{Step 4.}
In this step we prove that $F-F^\pi\stackrel{p}{\to}$ some strictly positive constant, where the probability here refers to randomness from all three layers and permutations.

To prove it,
note that 
\[
F-F^\pi
=\frac{SS_T}{\sum\limits_{k\in[K]}SS_k}-\frac{SS_T}{\sum\limits_{k\in[K]}SS_k^\pi}
=\frac{1}{n-1}SS_T\left(\frac{1}{\frac{1}{n-1}\sum\limits_{k\in[K]}SS_k}-\frac{1}{\frac{1}{n-1}\sum\limits_{k\in[K]}SS_k^\pi}\right).
\]

\textbf{Step 4(a).}
We first prove that $\frac{1}{n-1}SS_T\stackrel{p}{\to}$ some strictly positive constant. 
Note that 
\begin{eqnarray*}
\frac{1}{n-1}SS_T
&=&\frac{1}{n(n-1)}\sum\limits_{k_1,k_2 \in [K]}\sum\limits_{j_1\in [n_{k_1}],j_2\in [n_{k_2}]}W_1\Big(Q_{j_1}^{(k_1)},Q_{j_2}^{(k_2)}\Big)^2\\
&=&\sum\limits_{k_1,k_2 \in [K]}\frac{n_{k_1}n_{k_2}}{n(n-1)}\frac{1}{n_{k_1}n_{k_2}}\sum\limits_{j_1\in [n_{k_1}],j_2\in [n_{k_2}]}W_1\Big(Q_{j_1}^{(k_1)},Q_{j_2}^{(k_2)}\Big)^2\\
&\stackrel{p}{\to}&\frac{1}{K^2}\sum\limits_{k_1,k_2 \in [K]}E\Big\{W_1(Q_{1}^{(k_1)},Q_{2}^{(k_2)})^2\Big\}>0,
\end{eqnarray*}
which follows from the strong law of large numbers for U-statistics \citet[Chapter 5.4]{Serfling1980}.

\textbf{Step 4(b).}
Build on Step 3(a), \eqref{eq:LimitOfSS_k}, and \eqref{eq:LimitOfSS_k^pi}, it suffices to prove that $\sum\limits_{k\in[K]}\frac{1}{n-1}(SS_{k}-SS_{k}^\pi)\stackrel{p}{\to}$ some strictly negative constant.
It follows from \eqref{eq:LimitOfSS_k} and \eqref{eq:LimitOfSS_k^pi}
that
\begin{eqnarray*}
& &\sum\limits_{k\in[K]}\frac{1}{n_k-1}(SS_{k}-SS_{k}^\pi)\\
&\stackrel{p}{\to}&\left(1-\frac{1}{K}\right)\sum\limits_{k\in[K]}E\{W_1(Q^{(k)}_1,Q^{(k)}_2)^2\}-\sum\limits_{k_1^\prime\neq k_2^\prime\in[K]}\frac{E\{W_1(Q_{1}^{(k_1^\prime)},Q_{1}^{(k_2^\prime)})^2\}}{K}\\
&=&(K-1)\left(\frac{1}{K}\sum\limits_{k\in[K]}E\{W_1(Q^{(k)}_1,Q^{(k)}_2)^2\}-\sum\limits_{k_1^\prime\neq k_2^\prime\in[K]}\frac{E\{W_1(Q_{1}^{(k_1^\prime)},Q_{1}^{(k_2^\prime)})^2\}}{K(K-1)}\right)
<0,
\end{eqnarray*}
where the last inequality follows from $H_1$ and hence
\[
\sum\limits_{k\in[K]}\frac{1}{n-1}(SS_{k}-SS_{k}^\pi)\stackrel{p}{\to}
\frac{K-1}{K}\left(\frac{1}{K}\sum\limits_{k\in[K]}E\{W_1(Q^{(k)}_1,Q^{(k)}_2)^2\}-\sum\limits_{k_1^\prime\neq k_2^\prime\in[K]}\frac{E\{W_1(Q_{1}^{(k_1^\prime)},Q_{1}^{(k_2^\prime)})^2\}}{K(K-1)}\right),
\]
which is a strictly negative constant.
\medskip

\par\noindent
\textbf{Step 5.}
Building on the previous three steps, we have established that $\tilde{F}-\tilde F^\pi\stackrel{p}{\to}C$, where $C$ is a strictly positive constant. Accordingly, we have
$
\lim_{n\to\infty}P(\tilde F>\tilde F^\pi~|~H_1)=1.
$
\end{proof}
\begin{proof}[Proof of Theorem~\ref{theorem:TestConsistency}(b)]
Noting Theorem~\ref{theorem:ConsistencyinMixingImpliesConsistencyinMixture} and Proposition~\ref{Prop:W1BetweenTwoMixture}, this is analogous to the proof of Theorem~\ref{theorem:TestConsistency}(a) and hence omitted.
\end{proof}

\subsection{Proof of theorems in Section \ref{sec:algorithm}}
\begin{proof}[Proof of Theorem~\ref{thm:ConvergenceOfAlgorithm}]
 We focus on VDM. After understanding the proof of VDM, arguments for VEM and ISDM are straight-forward and hence omitted.
 
 The proof of VDM largely remains the same as \cite{bohning1982convergence} and we include it here only for the completeness of this paper.
 This proof consists of four steps. 
 In the first step, we prove the existence of $\hat Q$.
 In the second step, we prove an important property \eqref{eq:ImportantPropertyOfConvergence} for proving $\Phi(G_L)\to\Phi(\hat Q)$ as $L\to\infty$ if this algorithm doesn't stop.
 In the third step, we complete the proof in the case that this algorithm doesn't stop.
 In the fourth step, we complete the proof in the case that algorithm does stop at some $L$.
 
 \noindent\textbf{Step 1.}
 This step gives a proof of the existence of $\hat Q$, which is an analogue of \citet[Section 3.1]{simar1976maximum}.
 
 Let $\bar{\mathbb{Q}}_B$ be the set of all sub-distributions (total mass less or equal to 1) on $[0,B]$ and let $\bar \Gamma_N:=\{{\bm \mu}(\bar G)|\bar G\in\bar{\mathbb{Q}}_B\}$, where $\bar G\mapsto{\bm \mu}(\bar G):=\left(\mu_1(\bar G),\ldots,\mu_N(\bar G)\right)$ and
 \[
\bar G\mapsto\mu_i(\bar G):=\int_0^B\exp(-\lambda r_i)(\lambda r_i)^{X_i}{\sf d}\bar G(\lambda)\text{ for }i\in[N]\text{ and }\bar G\in\bar{\mathbb{Q}}_B.
\]
We claim that $\bar \Gamma_N$ is convex and compact. Convexity is obvious. 
Compactness follows from the weak compactness of $\bar{\mathbb{Q}}_B$, boundedness and continuity of $\lambda\mapsto\exp(-\lambda r_i)(\lambda r_i)^{X_i}$ on $[0,B]$, and Helly–Bray theorem, see Simar's arguments for further details.
It further follows from the concavity of $(\mu_1,\ldots,\mu_N)\mapsto\Psi(\mu_1,\ldots,\mu_N):=\frac{1}{N}\sum_{i=1}^{N}\log \mu_i$ on $\bar \Gamma_N$ that there exists a unique maximizer $(\hat \mu_1,\ldots,\hat\mu_N)$ of $\Psi$ on $\bar \Gamma_N$.
By the construction of $\bar \Gamma_N$, there exists a sub-distribution $\bar G_{max}\in\bar{\mathbb{Q}}_B$ such that $(\hat \mu_1,\ldots,\hat\mu_N)=(\mu_1(\bar G_{max}),\ldots,\mu_N(\bar G_{max}))$. 
The proof of that $\bar G_{max}$ is actually a distribution follows from exactly same arguments by \citet[Page 1202]{simar1976maximum}.
Now we complete the proof of the existence of $\hat Q$.

\noindent\textbf{Step 2.}
Let $\mathbb{Q}_B$ be the set of all distributions on $[0,B]$ and let $\delta_\lambda$ be the deterministic distribution at $\lambda\in[0,B]$.
Since we have $\Phi(G)>-\infty$ for each $G\in\mathbb{Q}_B\backslash\{\delta_0\}$, we can define the following directional directive 
\[
\Phi^\prime(G,\delta_\lambda):=\lim\limits_{\epsilon\to0^+}\epsilon^{-1}\Big\{\Phi\{(1-\epsilon)G\oplus\epsilon\delta_\lambda\}-\Phi(G)\Big\}
=\frac{1}{N}\sum_{i\in[N]}\frac{e^{-\lambda r_i}(\lambda r_i)^{X_i}}{\mu_i(G)}-1
\]
for $G\in\mathbb{Q}_B\backslash\{\delta_0\}$ and $\lambda\in[0,B]$.

In the second step, we prove that for all $\nu>0,\alpha\in\mathbb{R}$ there exists $\epsilon_0=\epsilon_0(\nu,\alpha)\in(0,1)$ such that 
\begin{eqnarray}\label{eq:ImportantPropertyOfConvergence}
\Phi^\prime(G,\delta_\lambda)\geq\nu\text{ implies }\Phi\{(1-\epsilon)G\oplus\epsilon\delta_\lambda\}-\Phi(G)\geq \epsilon\nu/2
\end{eqnarray}
for all $\epsilon\in[0,\epsilon_0(\nu,\alpha)]$, all $G\in\Delta_\alpha:=\{G\in\mathbb{Q}_B|\Phi(G)\geq \alpha\}$, and all $\lambda\in[0,B]$.

$\Psi,{\bm \mu},\bar{\mathbb{Q}}_B$ and $\bar \Gamma_N$ are defined in Step 1.
Since $\Psi$ is continuously differentiable on $\bar \Gamma_N\backslash\{0\}$, it follows from the mean value theorem that
\begin{align*}
\Phi\{(1-\epsilon)G\oplus\epsilon\delta_\lambda\}-\Phi(G)
&=\Psi\{(1-\epsilon){\bm \mu}(G)+\epsilon{\bm \mu}(\delta_\lambda)\}-\Psi({\bm \mu}(G))\\
&=\epsilon \nabla\Psi\left\{(1-\xi\epsilon){\bm \mu}(G)+\xi\epsilon{\bm \mu}(\delta_\lambda)\right\}^T{\bm \mu}(\delta_\lambda)-1,
\end{align*}
where $\nabla\Psi$ denotes the gradient of $\Psi$, for some $\xi\in[0,1]$.
Therefore, 
\begin{align*}
\Phi\{(1-\epsilon)G\oplus\epsilon\delta_\lambda\}-\Phi(G)-\epsilon \Phi^\prime(G,\delta_\lambda)
=\epsilon\Big\{\nabla\Psi\left\{(1-\xi\epsilon){\bm \mu}(G)+\xi\epsilon{\bm \mu}(\delta_\lambda)\right\}-\nabla\Psi({\bm \mu}(G))\Big\}^T{\bm \mu}(\delta_\lambda).
\end{align*}
Define $\mathcal{L}_{\alpha^\prime}:=\{{\bm \mu}\in\bar \Gamma_N:\Psi\geq \alpha^\prime\}$ for $\alpha^\prime\in\mathbb{R}$.
Note that $\mathcal{L}_{\alpha^\prime}$ is a compact set, on which $\nabla \Psi$ is uniformly continuous, for $\alpha^\prime=\alpha-1$.
Since ${\bm \mu}(G)\in \mathcal{L}_{\alpha}$, we can find a sufficiently small $\epsilon_0=\epsilon_0(\alpha,\nu)$ such that for all $\epsilon\in[0,\epsilon_0]$ we have
$
(1-\xi\epsilon){\bm \mu}(G)+\xi\epsilon{\bm \mu}(\delta_\lambda)\in \mathcal{L}_{\alpha-1}
$
and 
\[
\|\nabla\Psi\left\{(1-\xi\epsilon){\bm \mu}(G)+\xi\epsilon{\bm \mu}(\delta_\lambda)\right\}-\nabla\Psi({\bm \mu}(G))\|\leq \nu/(2S),
\]
where $\|\cdot\|$ denotes the Euclidean norm and $S:=\sup_{{\bm \mu}\in\bar\Gamma_N}\|{\bm \mu}\|$.
Therefore we have 
\begin{eqnarray}\label{eq:ReadyToContradiction}
|\Phi\{(1-\epsilon)G\oplus\epsilon\delta_\lambda\}-\Phi(G)-\epsilon \Phi^\prime(G,\delta_\lambda)|\leq \epsilon\nu/(2S)\cdot S=\epsilon\nu/2.
\end{eqnarray}
If the claim doesn't hold, i.e.
$
\Phi\{(1-\epsilon)G\oplus\epsilon\delta_\lambda\}-\Phi(G)<\epsilon\nu/2,
$
it follows from $-\Phi^\prime(G,\delta_\lambda)\leq-\nu$
that 
\[
\Phi\{(1-\epsilon)G\oplus\epsilon\delta_\lambda\}-\Phi(G)-\epsilon \Phi^\prime(G,\delta_\lambda)<\epsilon\nu/2-\epsilon\nu=-\epsilon\nu/2,
\]
which contradicts \eqref{eq:ReadyToContradiction}.

\noindent\textbf{Step 3.}
In this step, we assume that VDM doesn't stop and we have $\Phi(G_L)\to\Phi(\hat Q)$ as $L\to\infty$.

Note that $\Phi(G_L)$ is monotonically increasing and suppose $\lim_{L\to\infty}\Phi(G_L)=\Phi^+$.
If $\Phi^+<\Phi(\hat Q)$, then we have
\[
\Phi^\prime(G_L,\delta_{\lambda_{\text{max}}})
=\max_{\lambda\in[0,B]}\Phi^\prime(G_L,\delta_\lambda)
\geq\Phi^\prime(G_L,\hat Q)
\geq\Phi(\hat Q)-\Phi(G_L)
\geq\Phi(\hat Q)-\Phi^+\geq \nu>0,
\]
for some $\nu>0$, where the first inequality follows from \citet[Page 1204]{simar1976maximum} and the second inequality follows from the concavity of $\epsilon\mapsto \Phi((1-\epsilon)G_L+\epsilon\hat Q)$ with $\epsilon\in[0,1]$.
Then it follows from the claim in Step 2 that 
\[
\Phi(G_{L+1})-\Phi(G_{L})\geq \Phi\{(1-\epsilon_0)G_L\oplus\epsilon_0\delta_{\lambda_{\text{max}}}\}-\Phi(G_{L})\geq \nu\epsilon_0/2>0,
\]
which contradicts $\lim_{L\to\infty}\Phi(G_L)=\Phi^+$.

\noindent\textbf{Step 4.}
In this step, we prove that if VDM stops at some $L$, then $\Phi(G_L)=\Phi(\hat Q)$.

If $\Phi(G_L)<\Phi(\hat Q)$, we then have
\[
\max_{\lambda\in[0,B]}\Phi^\prime(G_L,\delta_\lambda)\geq \Phi(\hat Q)-\Phi(G_L)>0,
\]
which contradicts the criterion for stopping this algorithm.
 \end{proof}

\subsection{Proof of theorems in Section~\ref{sec:optimality}}
\begin{proof}[Proof of Theorem~\ref{theorem:UpperBoundOfMixing}(a)]
This proof consists of two steps, similar to Section 4 in~\cite{vinayak2019maximum}. In the first step, we prove that $W_1(Q,\hat Q)$ can be upper bounded by three parts, see \eqref{Formula1}.
In the second step, we upper bound these three parts separately with the help of Lemma \ref{Lemma:FirstPartOfUpperBound}, Lemma \ref{Lemma:SecondPartOfUpperBound} and Proposition~\ref{Proposition:BoundOnf1} and complete this proof.
\medskip

\par\noindent
\textbf{Step 1.}
For $x=0,1,\ldots$, let $x\mapsto h^{obs}_Q(x)$ denote the sample proportion, i.e.
$
h^{obs}_Q(x):=\sum_{i=1}^NI(x=X_i)/N,
$
where $I(\cdot)$ is an indicator function.
Recall that 
$
W_1(Q,\hat{Q})
=\sup_{\ell\in {\rm Lip}_1}\int_0^{B}\ell{\sf d}(Q-\hat{Q}),
$
where ${\rm Lip}_1$ represents all $1$-Lipschitz functions on $[0,B]$ and $\ell$ is one of those $1$-Lipschitz functions.
Without loss of generality, it is assumed that $\ell(0)=0$.
The idea is to use the following function 
\[
\lambda\mapsto\hat{\ell}(\lambda):=\sum_{x=0}^\infty b_x\frac{\lambda^xe^{-\lambda}}{x!}, \text{ where }b_x\in\mathbb{R}\text{ and }\lambda\in[0,B],
\]
to approximate the $1$-Lipschitz function $\lambda\mapsto \ell(\lambda)$
and upper bound $W_1(Q,\hat Q)$ by three parts.
It follows from a straight-forward algebra that
\begin{align*}
\int_0^{B}\ell(\lambda){\sf d}\left(Q(\lambda)-\hat{Q}(\lambda)\right)
=&\int_0^{B}\left(\ell(\lambda)-\hat{\ell}(\lambda)\right){\sf d}\left(Q(\lambda)-\hat{Q}(\lambda)\right)
+\int_0^{B}\sum_{x=0}^\infty b_x\frac{\lambda^xe^{-\lambda}}{x!}{\sf d}\left(Q(\lambda)-\hat{Q}(\lambda)\right)\\
\leq&~2\left\|\ell-\hat{\ell}\right\|_{\infty}+\sum_{x=0}^\infty b_x\left(h_Q(x)-h^{obs}_Q(x)\right)+\sum_{x=0}^\infty b_x\left(h^{obs}_Q(x)-h_{\hat{Q}}(x)\right),
\end{align*}
where $\|\ell-\hat{\ell}\|_{\infty}:=\sup_{\lambda\in[0,B]}|\ell(\lambda)-\hat{\ell}(\lambda)|$, 
and hence 
\begin{eqnarray}\label{Formula1}
W_1(Q,\hat{Q})
\leq\sup_{\ell\in \text{Lip}(1)}\left(2\left\|\ell-\hat{\ell}\right\|_{\infty}
+\sum_{x=0}^\infty b_x\left(h_{Q}(x)-h^{obs}_Q(x)\right)
+\sum_{x=0}^\infty b_x\left(h^{obs}_Q(x)-h_{\hat{Q}}(x)\right)\right).
\end{eqnarray}

\noindent
\textbf{Step 2.}
It follows from Lemma \ref{Lemma:FirstPartOfUpperBound} and Lemma \ref{Lemma:SecondPartOfUpperBound} that for an arbitrary $\delta\in(0,1/2)$ and an arbitrary $\epsilon\in(0,1)$ there exists constants $N(\epsilon)$ and $C(\epsilon)$ depending only on $\epsilon$ such that the sum of the last two terms in (\ref{Formula1}) is upper bounded by 
$C(\epsilon)\max_{x\geq 0}|b_x|\sqrt{\frac{B\vee 1}{N^{1-\epsilon}\delta^{1+\epsilon}}}$
for all $N\geq N(\epsilon)$ with probability at least $1-2\delta$.

\textbf{Step 2(a).} Suppose $c_0\leq\min\left\{\sqrt{e}c/C_2,0.001\right\}$, where $c=1/8$ and $C_2$ is a universal constant specified later.
It follows from Proposition~\ref{Proposition:BoundOnf1}(a) that 
 $\ell(\lambda)$ can be approximated by 
$\hat{\ell}(\lambda) = \sum_{x = 0}^{k} b_x\frac{\lambda^xe^{-\lambda}}{x!}$ with an uniform approximation error of
 $C_1 B/k$ with 
 $\max_x |b_x| \leq C_1\left(\sqrt{e}k/B\right)^{k}$ for $k\geq 4(B\vee 1)$, where $C_1>1$ is a universal constant.
Hence we have
\[
W_1(Q,\hat{Q})
\leq
2C_1\frac{B}{k}+C_1C(\epsilon)\left(\frac{\sqrt{e}k}{B}\right)^k\sqrt{\frac{B\vee 1}{N^{1-\epsilon}}\frac{1}{\delta^{1+\epsilon}}},
\]
for $N\geq N(\epsilon)$ and $k\geq 4(B\vee 1)$ with probability at least $1-2\delta$.
Taking $k=k(N,B)$ satisfying $\left(\sqrt{e}k/B\right)^k=N^{c}$ for $c=1/8$, it follows that
\begin{eqnarray}\label{eq:W_1<=B/k}
W_1(Q,\hat{Q})
\leq
2C_1B/k+C_1C(\epsilon)N^{c+\epsilon/2-1/2}\sqrt{B/\delta^{1+\epsilon}}.
\end{eqnarray}
To verify $k(N,B)\geq 4(B\vee 1)$, note that 
$
\left(\sqrt{e}k/B\right)^k=N^{c}
$
is equivalent to 
\[
\log\left(\sqrt{e}k/B\right)\exp\{\log\left(\sqrt{e}k/B\right)\}=(\sqrt{e}c\log N)/B.
\]
It further follows from $(\sqrt{e}c\log N)/B>0$ that
$k(N,B)$, as the solution of $(\sqrt{e}k/B)^k=N^{c}$, can be written using the Lambert $W$ function, i.e. $k(N,B)=\frac{B}{\sqrt{e}}\exp\left(W\left(\frac{\sqrt{e}c\log N}{B}\right)\right)$, where $W$ is the Lambert W function.
It follows from the expansion of $W$ (see Wiki of Lambert W function),
\[
W(x)=\log x-\log\log x+o(1), \text{ as }x\rightarrow\infty,
\]
that there exists a universal constant $C_2>0$ such that 
\[
\exp(W(x))\geq \frac{1}{2}\frac{x}{\log x}, \text{ for }x\geq C_2.
\]
It then follows from that $B\leq c_0\log N$ and $c_0\leq\sqrt{e}c/C_2$ that 
\[
\frac{\sqrt{e}c\log N}{B}
\geq \frac{\sqrt{e}c\log N}{c_0\log N}
\geq C_2
\]
and hence
\begin{eqnarray}
k(N,B)=\frac{B}{\sqrt{e}}\exp\left(W\left(\frac{\sqrt{e}c\log N}{B}\right)\right)
\geq\frac{B}{2\sqrt{e}}\frac{\frac{\sqrt{e}c\log N}{B}}{\log \frac{\sqrt{e}c\log N}{B}}
\geq\frac{c}{2}\frac{\log N}{\log \frac{\log N}{B}}.\label{SolutionOfk(N,B)}
\end{eqnarray} 
It further follows from $B\leq c_0\log N$ with $c_0\leq 0.001$ that 
\[
\frac{k(N,B)}{B}
\geq \frac{c}{2}\frac{\log N}{B}/\log \frac{\log N}{B}
\geq \frac{1}{16}\frac{1000}{\log 1000}
\geq4.
\]
If $\frac{c}{2}\frac{\log N}{\log \frac{\log N}{B}}\geq 4$ doesn't hold, then $E\{W_1(Q,\hat Q)\}\leq B\leq \frac{64 B}{\log N}\log\Big(\frac{\log N}{B}\vee e\Big)$ and hence Theorem~\ref{theorem:UpperBoundOfMixing}(a) is trivial.
Therefore without loss of generality we assume that $\frac{c}{2}\frac{\log N}{\log \frac{\log N}{B}}\geq 4$ and hence $k(N,B)\geq 4$.
As a consequence, we have $k(N,B)\geq 4(B\vee 1)$.

Combining \eqref{eq:W_1<=B/k} with \eqref{SolutionOfk(N,B)} and letting $\epsilon=1/4$, we have
\[
W_1(Q,\hat{Q})
\leq32C_1\frac{B\log \frac{\log N}{B}}{\log N}+C_1C(\epsilon)|_{\epsilon=1/4}\cdot N^{-1/4}\sqrt{\frac{B\vee 1}{\delta^{1+\epsilon}}},
\]
where $C(\epsilon)|_{\epsilon=1/4}$ means the value of the function $\epsilon\mapsto C(\epsilon)$ at $1/4$.
Therefore, for an arbitrary $\delta\in(0,1/2)$, there exists a universal constant $C_3$ such that for sufficiently large $N$ we have
\[
W_1(Q,\hat{Q})
\leq C_3\frac{B}{\log N}\left(\log \frac{\log N}{B}\right)\frac{1}{\delta^{5/8}},
\]
with probability at least $1-2\delta$.
Therefore, for sufficiently large $N$ we have
\[
E\{W_1(Q,\hat{Q})\}\leq 5C_3\frac{B}{\log N}\log \frac{\log N}{B}\leq 5C_3\frac{B}{\log N}\log\left(\frac{\log N}{B}\vee e\right).
\]

\textbf{Step 2(b).} Suppose $c_0>\min\left\{\frac{\sqrt{e}c}{C_2},0.001\right\}$. Then for $B\in[\min\left\{\sqrt{e}c/C_2,0.001\right\}\log N,c_0\log N]$, it follows from Theorem~\ref{theorem:UpperBoundOfMixing}(b) that 
$E\{W_1(Q,\hat{Q})\}\leq C_4\sqrt{B/\log N}\leq C_4\sqrt{c_0}$, where $C_4$ is a universal constant.
On the other hand, in this case $\frac{B}{\log N}\log\left(\frac{\log N}{B}\vee e\right)\geq \min\left\{\sqrt{e}c/C_2,0.001\right\}$ and hence
\[
E\{W_1(Q,\hat{Q})\}\leq \max\left\{5C_3,\frac{C_4\sqrt{c_0}}{\min\left\{\sqrt{e}c/C_2,0.001\right\}}\right\}\frac{B}{\log N}\log\left(\frac{\log N}{B}\vee e\right)
\]
holds for all $B\leq c_0\log N$.
\end{proof}

\begin{proof}[Proof of Theorem~\ref{theorem:UpperBoundOfMixing}(b)]
Since $B\geq c_0\log N$, we have $B\geq 1$ for sufficiently large $N$.
It follows from Step 1 in the proof of Theorem~\ref{theorem:UpperBoundOfMixing}(a), Lemma \ref{Lemma:FirstPartOfUpperBound} and Lemma \ref{Lemma:SecondPartOfUpperBound} that for an arbitrary $\delta\in(0,1/2)$ and an arbitrary $\epsilon\in(0,1)$ there exist constants $N(\epsilon)$ and $C(\epsilon)$ depending only on $\epsilon$ such that the sum of the last two terms in \eqref{Formula1} is upper bounded by 
$C(\epsilon)\max_{x\geq 0}|b_x|\sqrt{\frac{B}{N^{1-\epsilon}\delta^{1+\epsilon}}}$
for all $N\geq N(\epsilon)$ with probability at least $1-2\delta$.

If $c_0\geq 100$, it follows further from Proposition~\ref{Proposition:BoundOnf1}(b) that for sufficiently small $\epsilon$ there exists a constant $C_1=C_1(\epsilon)$ such that 
\begin{eqnarray*}
W_1(Q,\hat{Q})
&\leq&
C_1\left(\sqrt{\frac{B}{\log N}}+B^{3/2}N^{-1/2+2\epsilon} \sqrt{\frac{1}{\delta^{1+\epsilon}}}\right),
\end{eqnarray*}
with probability at least $1-2\delta$.
Since $B^{3}\leq C_0^3N^{1-3\epsilon_0}$, then it follows from choosing $\epsilon=(\epsilon_0/2)\wedge 0.01$ that there exists a constant $C_2=C_2(\epsilon_0)$ such that 
\[
W_1(Q,\hat{Q})\leq C_2\sqrt{\frac{B}{\log N}\frac{1}{\delta^{1+\epsilon}}}\text{~~and hence~~}E\{W_1(Q,\hat{Q})\}\leq C_3\sqrt{\frac{B}{\log N}},\text{ }
\]
where $C_3=C_3(\epsilon_0)$ is a constant.

If $c_0<100$, a $1$-Lipschitz function on $[0,B]$ can also be viewed as a Lipschitz function on $[0,100\log N]$ and hence it follows from letting $B=100\log N$ in  Proposition~\ref{Proposition:BoundOnf1}(b) that for sufficiently small $\epsilon$ there exists a constant $C_4=C_4(\epsilon)$ such that with probability $1-2\delta$
\begin{eqnarray*}
W_1(Q,\hat{Q})
\leq
C_4\left(1+\sqrt{B}N^{-1/2+2\epsilon}\log N \cdot \sqrt{\frac{1}{\delta^{1+\epsilon}}}\right)
\leq
C_4\left(1+\sqrt{C_0}N^{-1/3+2\epsilon}\log N \cdot \sqrt{\frac{1}{\delta^{1+\epsilon}}}\right).
\end{eqnarray*}
Therefore it follows from letting $\epsilon=0.01$ that for sufficiently large $N$
\[
E\{W_1(Q,\hat{Q})\}\leq 2C_4\leq \frac{2C_4}{\sqrt{c_0}}\sqrt{\frac{B}{\log N}},
\]
where the last inequality follows from $B\geq c_0\log N$.
\end{proof}

\begin{proof}[Proof of Theorem~\ref{theorem:LowerBoundOfMixing}(a)]
Suppose $a\geq M\geq 0$ are constants and $P$ and $Q$ are two random variables supported on $[a-M,a+M]$ with $E\{P^j\}=E\{Q^j\}$, $0\leq j\leq L$. Existence of $P$ and $Q$ is guaranteed by Proposition 4.3 in \cite{vinayak2019maximum}.
\medskip

\par\noindent
For $0<B\leq c_0\log N$, setting $a=C_1B$, $M=B$ and 
$$(L+1)/2=Be^2/(2C_1)\cdot\exp(W(4C_1\log(N)/(e^2B))),$$ where $W(\cdot)$ is the Lambert W function, $C_1=\max\{1,4e^2c_0,C_2c_0e^2/4,C_3c_0e^2/4\}$ and $C_2,C_3$ are universal positive constants specified later.
Since $W(x)=\log x-\log\log x+o(1)$ as $x\rightarrow\infty$, there exists a universal constant $C_2$ such that for $x\geq C_2$ we have
$
W(x)\geq \frac{1}{2}\log x.
$
Therefore, it follows from 
\[
4C_1\frac{\log N}{e^2B}
\geq 4C_2c_0\frac{e^2}{4}\frac{\log N}{e^2B}
\geq 4C_2c_0\frac{e^2}{4}\frac{\log N}{e^2c_0\log N}
=C_2
\]
that
\begin{eqnarray*}
L+1=\frac{Be^2}{C_1}\cdot\exp\{W(4C_1\log(N)/(e^2B))\}
\geq \frac{Be^2}{C_1}\sqrt{4C_1\frac{\log N}{e^2B}}
\geq \frac{Be^2}{C_1}\sqrt{4C_1\frac{\log N}{e^2c_0\log N}}
\geq\frac{4Be^2}{C_1}.
\end{eqnarray*}
By $(2eM)^2/a=(2eB)^2/(C_1B)=4e^2B/C_1$,
it follows that $L+1\geq (2eM)^2/a $.
Hence it follows from Proposition~\ref{Lemma:EssentialLemmaInLowerBound} that 
\begin{eqnarray*}
\text{TV}(P,Q)&\leq&2\left(\frac{eB}{\sqrt{C_1B(L+1)}}\right)^{L+1}
=2\left(\frac{e^2B}{2C_1(L+1)/2}\right)^{\frac{L+1}{2}}
=2N^{-2},
\end{eqnarray*}
where the last equality follows from the definition of the Lambert W function (see the proof of Theorem~\ref{theorem:UpperBoundOfMixing}(a) for details).
It follows from the LeCam minimax lower bound that for $N\geq 3$
\begin{eqnarray*}
\inf_{\tilde{Q}}\sup_Q E\{W_1(Q,\tilde{Q})\}
\geq\frac{1}{2}W_1(P,Q)(1-\text{TV}(P_N,Q_N))
\geq\frac{1}{2}W_1(P,Q)\left(1-2NN^{-2}\right)
\geq\frac{1}{6}W_1(P,Q).
\end{eqnarray*}
On the other hand, it follows from Proposition 4.3 in \cite{vinayak2019maximum} that 
$W_1(P,Q)\geq 2M/(2L)=B/L$.
Since $W(x)=\log x-\log\log x+o(1)$ as $x\rightarrow\infty$, there exists a universal constant $C_3$ such that for $x\geq C_3$ we have 
$
W(x)\leq 1+\log x-\log\log x.
$
Therefore, it follows from 
\[
4C_1\frac{\log N}{e^2B}
\geq 4C_3c_0\frac{e^2}{4}\frac{\log N}{e^2B}
\geq 4C_3c_0\frac{e^2}{4}\frac{\log N}{e^2c_0\log N}
=C_3
\]
that 
\begin{eqnarray*}
L\leq
 \frac{Be^2}{C_1}\cdot\exp\{W(4C_1\log(N)/(e^2B))\}
\leq\frac{Be^3}{C_1}\frac{4C_1\log N}{e^2B}/\log \frac{4C_1\log N}{e^2B}
=4e\log N/\log \frac{4C_1\log N}{e^2B}.
\end{eqnarray*}
Therefore,
\begin{eqnarray*}
\inf_{\tilde{Q}}\sup_Q E\{W_1(Q,\tilde{Q})\}
\geq\frac{1}{6}W_1(P,Q)
\geq\frac{1}{6}\frac{B}{4e\log N}\log \frac{4C_1\log N}{e^2B}
\geq\frac{B}{24e\log N}\log \frac{16c_0\log N}{B}.
\end{eqnarray*}
This completes the proof.
\end{proof}

\begin{proof}[Proof of Theorem~\ref{theorem:LowerBoundOfMixing}(b)]
Suppose $a\geq M\geq 0$ are constants and $P$ and $Q$ are two random variables supported on $[a-M,a+M]$ with $E\{P^j\}=E\{Q^j\}$, $0\leq j\leq L$. Existence of $P$ and $Q$ is guaranteed by Proposition 4.3 in \cite{vinayak2019maximum}.

For $B\geq c_0\log N$, setting $a=c_1B/\sqrt{c_0}$, $L=\log N$ and $M=c_1\sqrt{B\log N}$ with $c_1=1/(4e^4\sqrt{c_0})$.
Note that
\[
\frac{a}{M}=\frac{c_1B/\sqrt{c_0}}{c_1\sqrt{B\log N}}=\sqrt{\frac{B}{c_0\log N}}\geq 1
\]
and
\[
\frac{(2eM)^2}{a} =\frac{4e^2c_1^2B\log N}{c_1B/\sqrt{c_0}}=\sqrt{c_0}4e^2c_1\log N=\frac{\sqrt{c_0}4e^2}{4e^4\sqrt{c_0}}\log N=\frac{1}{e^2}\log N\leq 1+\log N=L+1.
\]
Therefore, it follows from the LeCam minimax lower bound and Proposition~\ref{Lemma:EssentialLemmaInLowerBound} that
\begin{eqnarray*}
\inf_{\tilde{Q}}\sup_Q E\{W_1(Q,\tilde{Q})\}
&\geq&\frac{1}{2}W_1(P,Q)(1-\text{TV}(P_N,Q_N))\\
&\geq&\frac{1}{2}W_1(P,Q)\left(1-2N\left(\frac{ec_1\sqrt{B\log N}}{\sqrt{c_1B(1+\log N)/\sqrt{c_0}}}\right)^{1+\log N}\right)\\
&\geq&\frac{1}{2}W_1(P,Q)\left(1-\frac{1}{e}N^{-\log 2}\right)
\geq\frac{3}{10}W_1(P,Q).
\end{eqnarray*}
On the other hand, it follows from Proposition 4.3 in \cite{vinayak2019maximum} that 
\[
W_1(P,Q)\geq \frac{2M}{2L}=\frac{c_1\sqrt{B\log N}}{\log N}=c_1\sqrt{\frac{B}{\log N}}.
\]
Hence
\begin{eqnarray*}
\inf_{\tilde{Q}}\sup_Q E\{W_1(Q,\tilde{Q})\}
\geq\frac{3c_1}{10}\sqrt{\frac{B}{\log N}}
\geq\frac{3}{40e^4\sqrt{c_0}}\sqrt{\frac{B}{\log N}}.
\end{eqnarray*}
This completes the proof.
\end{proof}

\appendix 
\section{Auxiliary proofs}\label{sec:appendix}

\begin{lemma} \label{Lemma:FirstPartOfUpperBound}
Suppose $Q$ is a distribution on $[0,B]$ and $\{X_i,i\in[N]\}$ are $N$ observations generated from $h_Q$ defined in \eqref{Formula:DefinitionOfh_Q}. For an arbitrary $\delta\in(0,1)$, the following inequality 
\[
\left|\sum_{x=0}^\infty b_x\left(h^{obs}_Q(x)-h_{Q}(x)\right)\right|
\leq  \max_{x}|b_x|\sqrt{\frac{\log(2/\delta)}{2N}},
\]
holds with probability at least $1-\delta$, where $b_x\in\mathbb{R}$ and $h^{obs}_Q=\sum_{i=1}^NI(x=X_i)/N$.
\end{lemma}
\begin{proof}[Proof of Lemma~\ref{Lemma:FirstPartOfUpperBound}]
By noting that $E\{h^{obs}_Q(x)\}=h_{Q}(x)$, this proof is basically an application of McDiarmid's inequality.

Let $\phi:\mathbb{R}^N\mapsto\mathbb{R}$  be a function of $(y_1,\ldots,y_N)\in \mathbb{R}^N$ such that 
\[
\phi(y_1,\ldots,y_N)
:=\frac{1}{N}\sum_{i=1}^N\sum_{x=0}^\infty b_xI(x\in\{y_i\}).
\]
Since for any $y_1,\ldots,y_N,y_{i^\prime}\in\mathbb{R}$
\[
|\phi(y_1,\ldots,y_i,\ldots,y_N)-\phi(y_1,\ldots,y_{i^\prime},\ldots,y_N)|
\leq \max_{x\geq 0}|b_x|\frac{1}{N},
\]
it follows from McDiarmid's inequality that for all $\epsilon>0$
\[
P(|\phi(X_1,\ldots,X_N)-E\{\phi(X_1,\ldots,X_N)\}|\geq\epsilon)\leq 2\exp\left(\frac{-2N\epsilon^2}{\max_{x\geq0}|b_x|^2}\right),
\]
or equivalently,
\[
P(|\sum_{x=0}^\infty b_x\left(h_x^{obs}-h_{Q}(x)\right)|\geq\epsilon)\leq 2\exp\left(\frac{-2N\epsilon^2}{\max_{x\geq0}|b_x|^2}\right)
\]
by noting that 
\[
\phi(X_1,\ldots,X_N)-E\{\phi(X_1,\ldots,X_N)\}=\sum_{x=0}^\infty b_x\left(h_x^{obs}-h_{Q}(x)\right).
\] 
Hence for an arbitrary $\delta\in(0,1)$ the following inequality
\[
\left|\sum_{x=0}^\infty b_x\left(h^{obs}_Q(x)-h_Q(x)\right)\right|
\leq \max_{x}|b_x|\sqrt{\frac{\log(2/\delta)}{2N}}
\]
holds with probability at least $1-\delta$.
\end{proof}

\begin{lemma} \label{Lemma:SecondPartOfUpperBound}
Suppose $Q$ is a distribution on $[0,B]$ and $\{X_i,i\in[N]\}$ is a random sample from $h_Q$ defined in \eqref{Formula:DefinitionOfh_Q}.
$\hat{Q}$ defined in \eqref{Formula:DefinitionOfwidehatQ} is a NPMLE of the mixing distribution $Q$. 
For an arbitrary $\delta\in(0,1)$ and an arbitrary $\epsilon\in(0,1)$, there exist constants $N(\epsilon)>0$ and $C=C(\epsilon)>0$ such that for all $N\geq N(\epsilon)$,
\begin{eqnarray*}
\left|\sum_{x=0}^\infty b_x\left(h^{obs}_Q(x)-h_{\hat{Q}}(x)\right)\right|
&\leq&C\max_{x\geq 0}|b_x|\sqrt{\frac{B\vee 1}{N^{1-\epsilon}\delta^{1+\epsilon}}}
\end{eqnarray*}
holds with probability at least $1-\delta$.
\end{lemma}
\begin{proof}[Proof of Lemma~\ref{Lemma:SecondPartOfUpperBound}]
Let ${\bf h}_Q^{obs}:=\left(h^{obs}_Q(0),h^{obs}_Q(1),\ldots\right)^T$, ${\bf h}_{\hat{Q}}:= \left(h_{\hat{Q}}(0),h_{\hat{Q}}(1),\ldots\right)^T$ and ${\bf h}_Q:= \left(h_Q(0),h_Q(1),\ldots\right)^T$.
For simplicity, ${\bf h}^{obs}_Q$, ${\bf h}_{\hat{Q}}$ and ${\bf h}_Q$ also represent distributions with respect to corresponding probability mass functions $x\mapsto h^{obs}_Q(x)$, $x\mapsto h_{\hat{Q}}(x)$ and $x\mapsto h_Q(x)$.

This proof consists of two steps.
In the first step, we prove that 
$
\left|\sum_{x=0}^\infty b_x\left(h^{obs}_Q(x)-h_{\hat{Q}}(x)\right)\right|
$
 can be upper bounded by $\text{KL}({\bf h}^{obs}_Q,{\bf h}_Q)$, where KL is the Kullback–Leibler divergence.
 In the second step, we upper bound $\text{KL}({\bf h}^{obs}_Q,{\bf h}_Q)$ by truncation arguments.
\medskip

\par\noindent\textbf{Step 1.}
It follows from the triangle inequality that 
\begin{eqnarray*}
\left|\sum_{x=0}^\infty b_x\left(h^{obs}_Q(x)-h_{\hat{Q}}(x)\right)\right|
\leq\max_{x\geq 0}|b_x|\sum_{x=0}^\infty \left|h^{obs}_Q(x)-h_{\hat{Q}}(x)\right|
=\max_{x\geq 0}|b_x|\|{\bf h}^{obs}_Q-{\bf h}_{\hat{Q}}\|_1,
\end{eqnarray*}
where $\|{\bf h}^{obs}_Q-{\bf h}_{\hat{Q}}\|_1$ is the total variation distance between distributions ${\bf h}^{obs}_Q$ and ${\bf h}_{\hat{Q}}$.
Then it follows from Pinsker’s inequality (see Proposition~\ref{Prop:PinskerInequality}) that 
\begin{eqnarray*}
\|{\bf h}^{obs}_Q-{\bf h}_{\hat{Q}}\|_1\leq\sqrt{\frac{1}{2} \cdot\text{KL}({\bf h}^{obs}_Q,{\bf h}_{\hat{Q}})},
\end{eqnarray*}
where KL is the Kullback–Leibler divergence,
and hence 
\begin{eqnarray*}
\left|\sum_{x=0}^\infty b_x\left(h^{obs}_Q(x)-h_{\hat{Q}}(x)\right)\right|
\leq\max_{x\geq 0}|b_x|\sqrt{\frac{1}{2} \cdot\text{KL}({\bf h}^{obs}_Q,{\bf h}_{\hat{Q}})}
\leq\max_{x\geq 0}|b_x|\sqrt{\frac{1}{2} \cdot\text{KL}({\bf h}^{obs}_Q,{\bf h}_Q)},
\end{eqnarray*}
by noting that maximum likelihood estimators maximize likelihood functions.
\begin{proposition}
\label{Prop:PinskerInequality}
(Pinsker’s Inequality, see \cite{cover2006elements}.) For discrete distributions $P$ and $Q$, it follows that
\[
\text{KL}(P,Q)\geq2\|P-Q\|^2_1,
\]
where KL$(P,Q)$ is the Kullback–Leibler divergence between $P$ and $Q$, and $\|P-Q\|_1$ is the total variation distance between $P$ and $Q$.
\end{proposition}

\noindent
\textbf{Step 2.}
Let $\{T_i:= X_iI(X_i\leq \lfloor2B\rfloor)+(\lfloor2B\rfloor+1) I(X_i\geq\lfloor2B\rfloor+1),i\in[N]\}$ be a truncated sample of $\{X_i,i\in[N]\}$, where $\lfloor2B\rfloor$ denotes the larger integer that is less or equal to $2B$.
Let $t_Q$ be the probability mass function of $T_1$ and let $t_Q^{obs}$ be the sample version of $t_Q$, i.e. 
for $x=0,\ldots,\lfloor2B\rfloor +1$
\begin{eqnarray*}
x\mapsto t_Q(x):= P(T_1=x)\text{~~~and~~~}x\mapsto t_Q^{obs}(x):= \frac{1}{N}\sum_{i=1}^NI(T_i=x).
\end{eqnarray*}
Note that $t_Q(x)=h_Q(x),t_Q^{obs}(x)=h_Q^{obs}(x)$ for $x=0,\ldots,\lfloor2B\rfloor$ and $t_Q(\lfloor2B\rfloor+1)=\sum_{x\geq \lfloor2B\rfloor+1}h_Q(x)$, $t_Q^{obs}(\lfloor2B\rfloor+1)=\sum_{x\geq \lfloor2B\rfloor+1}h_Q^{obs}(x)$ and hence
\begin{eqnarray*}
\text{KL}({\bf h}^{obs}_Q,{\bf h}_Q)
&=&\sum_{x\geq 0}h^{obs}_Q(x)\log\frac{h^{obs}_Q(x)}{h_Q(x)}\\
&=&\text{KL}({\bf t}_Q^{obs},{\bf t}_Q)-t^{obs}_Q(\lfloor2B\rfloor+1)\log\frac{t^{obs}_Q(\lfloor2B\rfloor+1)}{t_Q(\lfloor2B\rfloor+1)}+\sum_{x\geq \lfloor2B\rfloor+1}h^{obs}_Q(x)\log\frac{h^{obs}_Q(x)}{h_Q(x)},
\end{eqnarray*}
where ${\bf t}_Q^{obs}:= (t_Q^{obs}(0),\ldots,t_Q^{obs}(\lfloor2B\rfloor+1))$ and ${\bf t}_Q:= (t_Q(0),\ldots,t_Q(\lfloor2B\rfloor+1))$ are also viewed as distributions with respect to corresponding probability mass functions $x\mapsto t_Q(x)$ and $x\mapsto t^{obs}_Q(x)$.\\
If $t^{obs}_Q(\lfloor2B\rfloor+1)=0$, then $t^{obs}_Q(\lfloor2B\rfloor+1)\log\frac{t^{obs}_Q(\lfloor2B\rfloor+1)}{t_Q(\lfloor2B\rfloor+1)}=0$.
If not, it follows from 
$
\log(1+x)\leq x\text{ for }x>0
$ 
that 
\begin{eqnarray*}
-t^{obs}_Q(\lfloor2B\rfloor+1)\log\frac{t^{obs}_Q(\lfloor2B\rfloor+1)}{t_Q(\lfloor2B\rfloor+1)}
\leq t_Q(\lfloor2B\rfloor+1)-t^{obs}_Q(\lfloor2B\rfloor+1)
=\sum_{x\geq \lfloor2B\rfloor+1}(h_Q(x)-h_Q^{obs}(x)).
\end{eqnarray*}
Analogously, it can be proved that 
\begin{eqnarray*}
\sum_{x\geq \lfloor2B\rfloor+1}h^{obs}_Q(x)\log\frac{h^{obs}_Q(x)}{h_Q(x)}
\leq\sum_{x\geq \lfloor2B\rfloor+1}\frac{(h^{obs}_Q(x)-h_Q(x))^2}{h_Q(x)}+\sum_{x\geq \lfloor2B\rfloor+1}(h^{obs}_Q(x)-h_Q(x))
\end{eqnarray*}
and hence
\begin{eqnarray*}
-t^{obs}_Q(\lfloor2B\rfloor+1)\log\frac{t^{obs}_Q(\lfloor2B\rfloor+1)}{t_Q(\lfloor2B\rfloor+1)}+\sum_{x\geq \lfloor2B\rfloor+1}h^{obs}_Q(x)\log\frac{h^{obs}_Q(x)}{h_Q(x)}
\leq \sum_{x\geq \lfloor2B\rfloor+1}\frac{(h^{obs}_Q(x)-h_Q(x))^2}{h_Q(x)},
\end{eqnarray*}
where the last term can be upper bounded by
an analog of proof of Proposition 3.1(i) in \cite{lambert1984asymptotic} in the following substep. 

\textbf{Step 2(a).}
In this substep, we upper bound $\sum_{x\geq \lfloor2B\rfloor+1}(h^{obs}_Q(x)-h_Q(x))^2/h_Q(x)$.

Fix a $\epsilon>0$, choose a $\gamma>0$ in $(1-\epsilon,1)$ and an $a=(\sqrt{33}-1)/4\approx1.19>1$, 
where $\approx$ means approximately equal to.
Define $A:= a^{(1-\gamma)/3}$.
By H{\"o}lder's inequality,
\begin{eqnarray*}
& &N^{1-\epsilon}\sum_{x\geq \lfloor2B\rfloor+1}\frac{(h^{obs}_Q(x)-h_Q(x))^2}{h_Q(x)}\\
&=&N^{1-\epsilon}\sum_{x\geq \lfloor2B\rfloor+1}\frac{(h^{obs}_Q(x)-h_Q(x))^2}{h_Q(x)}A^{-x}A^{x}\\
&\leq&N^{1-\epsilon}\left(\sum_{x\geq \lfloor2B\rfloor+1}\frac{(h^{obs}_Q(x)-h_Q(x))^2}{h_Q(x)}A^{-x/\gamma}\right)^\gamma
\left(\sum_{x\geq \lfloor2B\rfloor+1}\frac{(h^{obs}_Q(x)-h_Q(x))^2}{h_Q(x)}A^{x/(1-\gamma)}\right)^{1-\gamma}.
\end{eqnarray*}
Since $A>1$, it follows that 
\[
N\cdot E\left\{\sum_{x\geq \lfloor2B\rfloor+1}\frac{(h^{obs}_Q(x)-h_Q(x))^2}{h_Q(x)}A^{-x/\gamma}\right\}
\leq\sum_{x\geq \lfloor2B\rfloor+1}A^{-x/\gamma}
=\frac{A^{-(\lfloor2B\rfloor+1)/\gamma}}{1-A^{-1/\gamma}}
\leq \frac{A^{-1/\gamma}}{1-A^{-1/\gamma}}
<\infty
\]
and hence for an arbitrary $\delta\in(0,1)$, the following inequality
\[
N\sum_{x\geq \lfloor2B\rfloor+1}\frac{(h^{obs}_Q(x)-h_Q(x))^2}{h_Q(x)}A^{-x/\gamma}\leq \frac{A^{-1/\gamma}}{1-A^{-1/\gamma}}\frac{1}{\delta}
\]
holds with probability at least $1-\delta$.
Therefore, with probability at least $1-\delta$, it follows that 
\[
\left(\sum_{x\geq \lfloor2B\rfloor+1}\frac{(h^{obs}_Q(x)-h_Q(x))^2}{h_Q(x)}A^{-x/\gamma}\right)^{\gamma}\leq \left(\frac{A^{-1/\gamma}}{1-A^{-1/\gamma}}\frac{1}{N\delta}\right)^\gamma
\leq\frac{1}{\left(A^{1/\gamma}-1\right)^\gamma}\frac{1}{(N\delta)^\gamma}.
\]
On the other hand, it follows from straight-forward algebra that
\begin{eqnarray*}
\sum_{x\geq \lfloor2B\rfloor+1}\frac{(h^{obs}_Q(x)-h_Q(x))^2}{h_Q(x)}A^{x/(1-\gamma)}
\leq\sum_{x\geq \lfloor2B\rfloor+1}\frac{(h^{obs}_Q(x))^2}{h_Q(x)}a^{x/3}+\sum_{x\geq \lfloor2B\rfloor+1}h_Q(x)a^{x/3}.
\end{eqnarray*}
The second term on the right is bounded by $1$ by the following arguments. 
Since $Q$ is supported on $[0,B]$, it follows that for any fixed $x\geq 2B$, $\lambda\mapsto f(x|\lambda):=e^{-\lambda}\lambda^x/x!$ is a monotonically increasing function and hence
\begin{eqnarray*}
h_Q(x)=\int_0^Bf(x|\lambda)dQ
\leq\sup_{\lambda\in[0,B]}f(x|\lambda)
=f(x|B)
=e^{-B}B^x/x!.
\end{eqnarray*}
Therefore,
\begin{eqnarray*}
\sum_{x\geq \lfloor2B\rfloor+1}h_Q(x)a^x
\leq\sum_{x\geq 2B}e^{-B}\frac{B^x}{x!}a^x
=e^{aB-B}\sum_{x\geq 2B}e^{-aB}\frac{(aB)^x}{x!}
=e^{aB-B}P(\text{Poi}(aB)\geq 2B),
\end{eqnarray*}
where $\text{Poi}(aB)$ denotes a random variable following from Poisson distribution with a parameter $aB$.
Moreover, it follows from Lemma~\ref{Lemma:PoissonTailInequality} that
\begin{eqnarray*}
P(\text{Poi}(aB)\geq 2B)
\leq\exp\left(-\{(2-a)/a\}^2aB/3\right) 
\end{eqnarray*}
and hence
\[
\sum_{x\geq \lfloor2B\rfloor+1}h_Q(x)a^x\leq e^{aB-B}\exp\left(-\{(2-a)/a\}^2aB/3\right) 
=\exp\{B(2a^2+a-4)/(3a)\}=1
\]
by verifying $2a^2+a-4=0$.

For any fixed $k>0$, define $A_N$ to be the event $\{h^{obs}_Q(x)>kh_Q(x)a^{x/3}\text{ for some }x\geq \lfloor2B\rfloor+1\}$.
Then, by Markov's inequality 
\begin{eqnarray*}
P(A_N)
\leq\sum_{x\geq \lfloor2B\rfloor+1}P(h^{obs}_Q(x)>kh_Q(x)a^{x/3})
\leq\sum_{x\geq \lfloor2B\rfloor+1}\frac{E\{h^{obs}_Q(x)\}}{kh_Q(x)a^{x/3}}
\leq\frac{1}{k(a^{1/3}-1)}.
\end{eqnarray*}
Thus, $P(A_N)$ can be made arbitrarily small by choosing $k$ large enough and on the complement of $A_N$ we have
\[
\sum_{x\geq \lfloor2B\rfloor+1}\frac{(h^{obs}_Q(x))^2}{h_Q(x)}a^{x/3}\leq k^2\sum_{x\geq \lfloor2B\rfloor+1}h_Q(x)a^{x}=k^2.
\]
Therefore, for an arbitrary $\delta\in(0,1)$, with probability at least $1-\delta$, the following inequality
\[
\sum_{x\geq \lfloor2B\rfloor+1}\frac{(h^{obs}_Q(x))^2}{h_Q(x)}a^{x/3}\leq \left(\frac{1}{\delta}\frac{1}{a^{1/3}-1}\right)^2
\]
holds.
Thus, for an arbitrary $\delta\in(0,1)$, with probability at least $1-\delta$, it follows  that 
\[
\left(\sum_{x\geq \lfloor2B\rfloor+1}\frac{(h^{obs}_Q(x)-h_Q(x))^2}{h_Q(x)}A^{x/(1-\gamma)}\right)^{1-\gamma}
\leq\left(\left(\frac{1}{\delta}\frac{1}{a^{1/3}-1}\right)^2+1\right)^{1-\gamma}
\leq\left(\frac{20}{\delta}\right)^{2-2\gamma},
\]
where the last inequality follows from $a=(\sqrt{33}-1)/4$ and $\gamma<1$.
For an arbitrary $\delta\in(0,1/2)$, with probability at least $1-2\delta$, it follows that 
\begin{eqnarray*}
N^{1-\epsilon}\sum_{x\geq \lfloor2B\rfloor+1}\frac{(h^{obs}_Q(x)-h_Q(x))^2}{h_Q(x)}
\leq
N^{1-\epsilon}
\frac{1}{\left(a^{\frac{1-\gamma}{3\gamma}}-1\right)^\gamma}\frac{1}{(N\delta)^\gamma}
\left(\frac{20}{\delta}\right)^{2-2\gamma}
=
N^{1-\epsilon-\gamma}
\frac{20^{2-2\gamma}}{\left(a^{\frac{1-\gamma}{3\gamma}}-1\right)^\gamma}\frac{1}{\delta^{2-\gamma}}
\end{eqnarray*}
and hence by letting $\gamma$ go to $1-\epsilon$,
it follows that 
\begin{eqnarray*}
\sum_{x\geq \lfloor2B\rfloor+1}\frac{(h^{obs}_Q(x)-h_Q(x))^2}{h_Q(x)}
&\leq&
\frac{20^{2\epsilon}}{(a^{\frac{\epsilon}{3(1-\epsilon)}}-1)^{1-\epsilon}}\frac{1}{N^{1-\epsilon}}\frac{1}{\delta^{1+\epsilon}}.
\end{eqnarray*}

\textbf{Step 2(b).}
In this subset, we complete the upper bound of $\text{KL}({\bf h}^{obs},{\bf h}_Q)$.

As a result of Step 2(a), for arbitrary $\delta\in(0,1/2)$ and $\epsilon\in(0,1)$, with probability at least $1-2\delta$, it follows that
\begin{eqnarray*}
\text{KL}({\bf h}^{obs},{\bf h}_Q)
\leq\text{KL}({\bf t}_Q^{obs},{\bf t}_Q)+\frac{20^{2\epsilon}}{(a^{\frac{\epsilon}{3(1-\epsilon)}}-1)^{1-\epsilon}}\frac{1}{N^{1-\epsilon}}\frac{1}{\delta^{1+\epsilon}}.
\end{eqnarray*}

To upper bound $\text{KL}({\bf t}_Q^{obs},{\bf t}_Q)$, the KL divergence between empirical observations and the true distribution for discrete distributions,
it follows from \cite{mardia2019concentration} that with probability $1-\delta$
\begin{eqnarray*}
\text{KL}({\bf t}^{obs},{\bf t}_Q)\leq \frac{2B+1}{2N}\log\frac{4N}{2B+1}+\frac{1}{N}\log\frac{3e}{\delta}
\end{eqnarray*}
and hence for any $\epsilon\in(0,1)$ and $\delta\in(0,1/3)$, with probability at least $1-3\delta$ it follows 
that 
\begin{eqnarray*}
\text{KL}({\bf h}^{obs}_Q,{\bf h}_Q)
\leq
\frac{2B+1}{2N}\log\frac{4N}{2B+1}+\frac{1}{N}\log\frac{3e}{\delta}
+\frac{20^{2\epsilon}}{(a^{\frac{\epsilon}{3(1-\epsilon)}}-1)^{1-\epsilon}}\frac{1}{N^{1-\epsilon}}\frac{1}{\delta^{1+\epsilon}}.
\end{eqnarray*}
Therefore, there exist positive constants $N_1=N_1(\epsilon)$ and $C_1=C_1(\epsilon)$ such that for $N\geq N_1$ 
\[
\text{KL}({\bf h}^{obs}_Q,{\bf h}_Q)\leq
C_1
\frac{B\vee 1}{N^{1-\epsilon}\delta^{1+\epsilon}}
\] 
holds with probability at least $1-3\delta$ for any $\epsilon\in(0,1)$ and $\delta\in(0,1/3)$.
\end{proof}

\begin{proposition}\label{Proposition:BoundOnf1}~
\begin{itemize}
\item[(a)]
For any positive integer $k\geq 4(B\vee 1)$ and any $1$-Lipschitz function $\lambda\mapsto \ell(\lambda)$ on $[0,B]$ with $\ell(0)=0$, 
there exists an approximation $\hat{\ell}(\lambda) = \sum_{x = 0}^{k} b_x\frac{\lambda^xe^{-\lambda}}{x!}$ such that 
\[
\sup_{\lambda\in[0,B]}|\hat{\ell}(\lambda)-\ell(\lambda)|\leq CB/k
\]
and 
$\max_x |b_x| \leq C\left(\sqrt{e}k/B\right)^{k}$,
where $C>1$ is a universal constant.
\item[(b)] Suppose $B>0$, $N\in\mathbb{N}^+$ and there exists constants $c_0,C_0>0$ such that $B\in[c_0\log N, C_0N]$.
Then, for any fixed $c_0\geq 100$ and any small $\epsilon\in(0,0.02)$
there exist constants $C(\epsilon)>0$ and $N(\epsilon)>1$ and a sequence of coefficients $\{b_x\}_{x=0}^\infty$ such that for $N\geq N(\epsilon)$
any $1$-Lipschitz function $\ell(\lambda)$ on $[0,B]$ with $\ell(0)=0$ can be approximated by 
 $\hat{\ell}(\lambda) = \sum_{x=0}^\infty b_x\frac{\lambda^xe^{-\lambda}}{x!}$ with  an uniform approximation error of
 $C(\epsilon)\sqrt{\frac{B}{\log N}}$ with $\max_x |b_x| \leq C(\epsilon)BN^\epsilon$. 
 \end{itemize}
\end{proposition}

\begin{proof}[Proof of Proposition~\ref{Proposition:BoundOnf1} (a)]
The following two facts are used in our proof.
\begin{fact}[Chapter 2.6 Equation 9 in \cite{timan2014theory}]\label{fact1:coeff-bound}
Suppose $k$ is a non-negative integer and $\lambda\mapsto p_k(\lambda)$ is a polynomial function with coefficients $c_0,\ldots,c_k$, i.e. $p_k(\lambda) := \sum_{x=0}^k c_x\lambda^x$. Then it follows that coefficients $\{c_x\}_{x=0}^k$ satisfy
\[
|c_x|\le \frac{k^x}{x!}\max_{|\lambda|\le 1}|p_k(\lambda)|.
\]
\end{fact}
\begin{fact}[Approximating $e^\lambda$ with Taylor expansion]\label{fact2:exp-approx-taylor}
Let $\lambda\in [0, B]$. For any $k\geq 2B$, it follows that
\[
e^{\lambda} - \sum_{x=0}^k \frac{\lambda^x}{x!} 
= \sum_{x=k+1}^\infty \frac{\lambda^x}{x!} 
=\frac{\lambda^k}{k!}\sum_{x=1}^\infty\left(\frac{\lambda}{k+x}\cdots\frac{\lambda}{k+1}\right)
\leq\frac{\lambda^k}{k!}\sum_{x=1}^\infty\frac{1}{2^x}
=\frac{\lambda^k}{k!}
\]
and hence
\[
|e^{\lambda} - \sum_{x=0}^k \frac{\lambda^x}{x!}|/e^\lambda \le \frac{\lambda^k}{k!e^\lambda} \leq \frac{B^k}{k!e^B}.
\]
\end{fact}
Applying Fact~\ref{fact2:exp-approx-taylor}, it holds that for any $k\geq 2B$, there exists a polynomial $q_k(\lambda)=\sum_{x=0}^k \lambda^x/x!$ of degree $k$ such that $|1 - q_k(\lambda)e^{-\lambda}|\leq B^k/(k!e^B)$ for all $\lambda\in [0,B]$. It is well known through Jackson's theorem (see Lemma~\ref{lem:jackson}) that for any $1$-Lipschitz function $\ell(\cdot)$ on $[0,B]$, there exists a polynomial $p_k(\lambda)$ of degree $k$ such that $\sup_{\lambda\in[-B,B]}|\ell(\lambda)-p_k(\lambda)|\leq C_1B/k$, where $\ell(\lambda):= -\ell(-\lambda)$ for $\lambda<0$ and $C_1>0$ is a universal constant independent of $k$ and $\ell$.
Combining $p_k(\lambda)$, $q_k(\lambda)$ and the fact that $|p_k(\lambda)|\leq B+C_1B/k\leq (1+C_1)B$, 
it follows that for $\lambda\in[0,B]$
\begin{eqnarray*}
|p_k(\lambda)q_k(\lambda)e^{-\lambda}-\ell(\lambda)| 
\leq |p_k(\lambda)(q_k(\lambda)e^{-\lambda}-1)|+|p_k(\lambda)-\ell(\lambda)|
\leq (1+C_1)\frac{B}{k}\left(\frac{kB^{k}e^k}{\sqrt{k}k^k e^B}+1\right), 
\end{eqnarray*}
where the last inequality follows from $k!\geq \sqrt{k}\left(k/e\right)^k$ for $k\geq 2$ by Stirling's approximation.
It further follows from the increasing monotonicity of $B\mapsto B^k/e^B$ for $B\leq k/2$ that
\[
\sqrt{k}B^{k}e^k/(k^k e^B)
\leq
\sqrt{k}(k/2)^{k}e^k/(k^k e^{k/2})
=\sqrt{k}\left(\sqrt{e}/2\right)^k
<1,
\]
where the last inequality holds for all $k\geq 2$,
and hence
\[
|p_k(\lambda)q_k(\lambda)e^{-\lambda}-\ell(\lambda)|
\leq
2(1+C_1)B/k
\]
for $k\geq 2(B\vee 1)$.
Therefore, we have shown that for any $k\geq 2(B\vee 1)$, there exists a function 
\[
\hat{\ell}(\lambda) = p_k(\lambda)q_k(\lambda)e^{-\lambda} = \sum_{x=0}^{2k}b_x\frac{\lambda^xe^{-\lambda}}{x!}
\]
such that $|\hat{\ell}(\lambda)-\ell(\lambda)| \leq 2(1+C_1)B/k$. 
For the bounded on the coefficients $b_x$, first let us define the polynomial $r(\lambda) := p_k(B\cdot \lambda)q_k(B\cdot \lambda) = \sum_{x=0}^{2k}b'_x\frac{\lambda^x}{x!}$.
Note that $b_x=b'_x/B^x$, 
\begin{eqnarray*}
|r(\lambda)|\leq\left(B+2(1+C_1)B/k\right)e^B
\leq2(1+C_1)Be^B
\end{eqnarray*} 
for $\lambda\in[0,1]$ and 
\begin{eqnarray*}
|r(\lambda)|
\leq|q_k(B\cdot \lambda)|(1+C_1)B
\leq\left(1+\frac{B^{k+1}}{(k+1)!}\right)(1+C_1)B
\leq \left(1+e/2e^B\right)(1+C_1)B
\leq 3(1+C_1)Be^B
\end{eqnarray*}
 for $\lambda\in[-1,0)$.
Then we can apply Fact~\ref{fact1:coeff-bound} for the polynomial $r(\lambda)$, which implies that 
\begin{eqnarray*}
\frac{|b'_x|}{x!}\le \frac{(2k)^x}{x!}\max_{|\lambda|\leq 1}|r(\lambda)| 
\leq\frac{(2k)^x}{x!}3(1+C_1)Be^B,
\end{eqnarray*}
and hence
\[
\max_x |b_x|
= \max_x \frac{|b'_x|}{B^x} 
\leq\max_x\left(\frac{2k}{B}\right)^x3(1+C_1)Be^B
=3(1+C_1)\left(\frac{2k}{B}\right)^{2k}Be^B
\leq 3(1+C_1)\left(\frac{2\sqrt{e}k}{B}\right)^{2k},
\]
where the last inequality follows from $B\leq k/2\leq \exp(k/2)$ and $e^B\leq \exp(k/2)$.
\end{proof}

\begin{proof}[Proof of Proposition~\ref{Proposition:BoundOnf1} (b)]
Since $B\geq c_0\log N$, we have $B\geq 1$ for sufficiently large $N$.
Note that $\lambda\mapsto\frac{1}{B}\ell(B\lambda)$ is a Lipschitz-$1$ function on $[0,1]$.
By Proposition~\ref{Proposition:BoundOnf1.5}, it follows that there exists a sequence of coefficients $\{b_x\}_{x=0}^\infty$ such that
\[
\left|\frac{1}{B}\ell(B\lambda)- \sum_{x=0}^\infty b_x \Pr(\text{Poi}(B\lambda) = x)\right|\le C(\epsilon)\sqrt{\frac{1}{B\log N}}, \text{ for any }\lambda \in [0, 1],
\]
where $b_x = 0$ for $x > 4B$, and 
\[
\left|b_x-\frac{1}{B}\ell\left(B\cdot\frac{x}{B}\right)\right|\leq \frac{C(\epsilon)(1+x^{1/2})N^\epsilon}{B}, \text{ for }x\leq 4B.
\]
Defining $b^*_x=Bb_x$ and replacing $B\lambda$ by $\lambda$, it follows that
\[
|\ell(\lambda)- \sum_{x=0}^\infty b^*_x \Pr(\text{Poi}(\lambda) = x)|\le C(\epsilon)\sqrt{\frac{B}{\log N}}, \text{ for any }\lambda \in [0, B],
\]
where $b^*_x = 0$ for $x > 4B$, and 
\[
\left|b^*_x-\ell(x)\right|\leq C(\epsilon)(1+x^{1/2})N^{\epsilon}, \text{ for }x\leq 4B.
\]
Moreover, It follows from the triangle inequality that $|b^*_x|\leq 4B+C(\epsilon)(1+2B^{1/2})N^{\epsilon}= O(BN^\epsilon)$.
\end{proof}

The following proposition is an extension of \citet[Lemma 3]{wu2016minimax}; see also \citet[Section 3.3]{wu2020polynomial} for a nice survey.
\begin{proposition}[Lemma 32 in~\cite{jiao2018minimax}]
\label{Lemma:EssentialLemmaInLowerBound}
Suppose $U_0$, $U_1$ are two random variables supported on $[a - M, a + M]$, where $a \geq M \geq 0$ are constants. Suppose $E\{U^j_0\}= E\{U^j_1\}, 0 \leq j \leq L$. 
Denote the marginal distribution of $X$ where $X|\lambda\sim Poi(\lambda)$, $\lambda \sim U_i$ as $F_i$, where $i=0,1$. If $L + 1 \ge (2eM)
^2/a$, then
$
\text{TV}(F_0, F_1) \leq 2(eM/\sqrt{a(L+1)})^{L+1}
$
\end{proposition}

\begin{proposition}
\label{Proposition:BoundOnf1.5}
Suppose $B>0$ and $N\in\mathbb{N}^+$ and there exists constants $c_0,C_0>0$ such that $B\in[c_0\log N, C_0N]$.
Let $\ell(\cdot)$ be any Lipschitz-$1$ function on $\mathbb{R}$ with $\ell(0)=0$.
Then, for any fixed $c_0\geq 96$ and any small $\epsilon\in(0,0.02)$ there exist positive constants $C(\epsilon)>0$ and $N(\epsilon)>1$ depending on $\epsilon$ and a sequence of coefficients $\{b_x\}_{x=0}^\infty$ such that the following inequality holds for $N\geq N(\epsilon)$, i.e.
\begin{eqnarray}
|\ell(\lambda) - \sum_{x=0}^\infty b_x \Pr(\text{Poi}(B\lambda) = x)|\leq C(\epsilon)\sqrt{\frac{1}{B\log N}}, \lambda \in [0, 1]
\label{Eqn:UpperBoundOfErrorInPoissonPoly}
\end{eqnarray}
where $b_x=0$ for $x>4B$, and
\begin{eqnarray}
\left|b_x-\ell\left(\frac{x}{B}\right)\right|\leq C(\epsilon)(1+x^{1/2})\frac{N^\epsilon}{B}, \text{ for any }x\leq 4B.\label{Eqn:UpperBoundOfCoefficientsInPoissonPoly}
\end{eqnarray}
\end{proposition}

\begin{proof}[Proof of Proposition~\ref{Proposition:BoundOnf1.5}]
Note that Proposition~\ref{Proposition:BoundOnf1.5} is an analogue of Theorem 5 in~\cite{han2020optimality} and these lemmas below will be used in the following proof.

\begin{lemma}[Jackson’s theorem, Lemma 10 of \cite{han2020optimality}]
\label{lem:jackson} 
Let $k>0$ be any integer, and $[a, b] \subseteq \mathbb{R}$ be any
bounded interval. For any Lipschitz-$1$ function $\ell(\cdot)$ on $[a, b]$, there exists a universal constant $C$ independent of $k, \ell$ such that there exists a polynomial $p_k(\cdot)$ of degree at most $k$ such that
\begin{align}
|\ell(\lambda)-p_k(\lambda)|\leq C\sqrt{(b-a)(\lambda-a)}/k, \; \forall \lambda\in [a,b].
\end{align}
In particular, the following norm bound holds:
\begin{align}
\sup_{\lambda\in [a,b]}|\ell(\lambda)-p_k(\lambda)|\leq C(b-a)/k.
\end{align}
\end{lemma}

\begin{lemma}[Lemma 11 of \cite{han2020optimality}]
\label{CoefficientBound} 
Let $p_k(\lambda)=\sum_{x=0}^ka_x \lambda^x$ be a polynomial of degree at most $k$ such that $|p_k(\lambda)|\leq A$ for $\lambda\in[a,b]$.
Then 
\begin{enumerate}
\item If $a+b\neq0$, then 
\[
|a_x|\leq 2^{7k/2}A\left|\frac{a+b}{2}\right|^{-x}\left(\left|\frac{b+a}{b-a}\right|^k+1\right),\ \ \ \ x=0,\cdots,k.
\]
\item If $a+b=0$, then 
$
|a_x|\leq Ab^{-x}(\sqrt{2}+1)^{k},\ \ \ \ x=0,\cdots,k.
$
\end{enumerate}
\end{lemma}

\begin{lemma}[Poisson tail inequality, Lemma 12 of \cite{han2020optimality}]
\label{Lemma:PoissonTailInequality} 
For $X\sim\text{Poi}(\lambda)$ and any $\delta>0$, we have
\begin{eqnarray*}
P(X\geq (1+\delta)\lambda)\leq\exp\left(-(\delta^2\wedge \delta)\lambda/3\right)\text{~~~~and~~~~}
P(X \leq (1-\delta)\lambda)\leq\exp\left(-\delta^2\lambda/3\right).
\end{eqnarray*}
\end{lemma}

\begin{lemma}[Lemma 15 of \cite{han2020optimality}]
\label{Lemma:CharlierPolynomials} 
Define
\[
g_{d,x}(z):=\sum_{d^\prime=0}^d{d \choose d^\prime}(-x)^{d-d^\prime}\prod_{d^{\prime\prime}=0}^{d^\prime-1}\left(z-\frac{2d^{\prime\prime}}{n}\right)
\]
with $d\in\mathbb{N},x\in[0,1]$.
Then for any $z\in[0,1]$ and $d\geq 1$, the following identity holds:
\[
g_{d,x}\left(z+\frac{2}{n}\right)-g_{d,x}(z)=\frac{2d}{n}g_{d-1,x}(z).
\]
Moreover, if $nz/2\in\mathbb{N}$ and $\max\{|z-x|,8d/n,\sqrt{8zd/n}\}\leq \Delta$, then 
$
|g_{d,x}(z)|\leq (2\Delta)^d.
$
\end{lemma}
\medskip

\par\noindent
Now we start our proof. Since $B\geq c_0\log N$, we have $B\geq 1$ for sufficiently large $N$.
Note that the proof here is an analog of proof of Theorem 5 in \cite{han2020optimality}, but has much more details. 
We could omit the following proof, but for the completeness of this paper, we decide to write it down.

We shall construct the following local intervals: for $c_1:=c_0/4$ and $m=1,2,\cdots,M:=\sqrt{B/(c_1\log N)}$, define
\begin{align*}
&I_m:=\left[\frac{c_1\log N}{B}\cdot (m-1)^2,\frac{c_1\log N}{B}\cdot m^2\right],
~~~~I_m^\prime:=\left[\frac{c_1\log N}{B}\cdot (m-4/3)_+^2,\frac{c_1\log N}{B}\cdot (m+1/3)^2\right],\\
&I_m^{\prime\prime}:=\left[\frac{c_1\log N}{B}\cdot (m-2)_+^2,\frac{c_1\log N}{B}\cdot (m+1)^2\right],
\end{align*}
and without loss of generality we assume that $M$ is an integer.
Note that $M\geq 2$. 
We shall also define
\[
\lambda_m:=\frac{c_1\log N}{B}\cdot \frac{(m-4/3)_+^2+(m+1/3)^2}{2}
\]
to be the center of $I_m^\prime$. 
Note that $I_m\subset I^\prime_m\subset I_m^{\prime\prime}$, and it follows from Lemma~\ref{Lemma:PoissonTailInequality} that for
$m\geq 2$,
\begin{eqnarray}
P(\text{Poi}(B\lambda)\notin BI_m^\prime|\lambda\in I_m)
\leq 2N^{-c_1/27}.\label{Formula:ExpUpperBounded1}
\end{eqnarray}
For $m=1$, it can be analogous to verify that the last display holds for $m=1$ and hence it holds for $m=1,\ldots,M$.
Analogously, we have the following inequalities: for $m=1,\ldots,M$
\begin{eqnarray}
P(\text{Poi}(B\lambda)\notin BI_m^{\prime\prime}|\lambda\in I_m^\prime)\leq 2N^{-c_1/3}
\label{Formula:ExpUpperBounded2}
\end{eqnarray}
and
\begin{eqnarray}
P(\text{Poi}(B\lambda)\notin BI_m|\lambda\in I_m-I_{m-1}^\prime-I_{m+1}^{\prime})
\leq 2N^{-c_1/12}
.\label{Formula:ExpUpperBounded3}
\end{eqnarray}

Now we use the local Poisson polynomial on each local interval $I_m^\prime$ constructed in Lemma~\ref{Lemma:LocalPoissonPolynomial} to prove Proposition~\ref{Proposition:BoundOnf1.5}.

We assume that $\sum_{x=0}^\infty b_x^{(m)}P(\text{Poi}(B\lambda/2)=x)$ is the Poisson polynomial given by Lemma \ref{Lemma:LocalPoissonPolynomial} on the $m$-th local interval $I_m^\prime$, with $B$ replaced by $B/2$. Now consider the following Poisson polynomial:
\begin{eqnarray}
p(\lambda):= \sum_{x=0}^\infty b_x P(\text{Poi}(B\lambda)=x)\text{ with }
b_x:= \frac{1}{2^x}\sum_{m=1}^M\sum_{k\in BI_m/2}{x \choose k}b_{x-k}^{(m)}.
\label{Eqn:CoefficientOfRealB}
\end{eqnarray}
We claim that the above polynomial with coefficients in (\ref{Eqn:CoefficientOfRealB}) satisfies Proposition~\ref{Proposition:BoundOnf1.5}.
We first verify the inequality (\ref{Eqn:UpperBoundOfErrorInPoissonPoly}). Using a change of variable $j=x-k$, we have
\begin{eqnarray*}
p(\lambda)
=\sum_{m=1}^MP(\text{Poi}(B\lambda/2)\in BI_m/2)\sum_{j=0}^\infty b_j^{(m)}P(\text{Poi}(B\lambda/2)=j).
\end{eqnarray*}
Since $I_m$ constitutes a partition of $[0,1]$, for $\lambda\in[0,1]$ there exists $m^*=1,2,\cdots,M$, such that $\lambda\in I_{m^*}$.
We distinguish into three cases:
\begin{itemize}
\item
Case 1: $\lambda\in I_{m^*}-I_{m^*-1}^\prime-I_{m^*+1}^\prime$. By (\ref{Formula:ExpUpperBounded3}), we have $P(\text{Poi}(B\lambda/2)\notin BI_{m^*}/2)\leq 2N^{-1}$ since $c_1=c_0/4\geq 24$, and therefore $P(\text{Poi}(B\lambda/2)\in BI_{m}/2)\leq 2N^{-2}$ for any $m\neq m^*$. Hence,
\begin{eqnarray*}
|\ell(\lambda)-p(\lambda)|
\leq C(\epsilon)\sqrt{\frac{\lambda}{B\log N}}+4N^{-2}\left(1+2C(\epsilon)N^{\epsilon}\right)
\end{eqnarray*}
where we have used (\ref{Eqn:BoundOfCoefficientInLocalPoissonPolynomial}) in the second last inequality. As a result, the desired approximation error in (\ref{Eqn:UpperBoundOfErrorInPoissonPoly}) holds.
\item Case II: $\lambda\in I_{m^*}\cap I_{m^*+1}^\prime$. In this case, Lemma~\ref{Lemma:PoissonTailInequality} 
gives $P(\text{Poi}(B\lambda/2)\in BI_m/2)\leq N^{-2}$ for any $m\notin \{m^*,m^*+1\}$.
Consequently,
\begin{eqnarray*}
|\ell(\lambda)-p(\lambda)|&\leq&
P(\text{Poi}(B\lambda/2)\in BI_{m^*}/2)\left|\ell(\lambda)-\sum_{j=0}^\infty b_j^{(m^*)}P(\text{Poi}(B\lambda/2)=j)\right|\\
& &+P(\text{Poi}(B\lambda/2)\in BI_{m^*+1}/2)\left|\ell(\lambda)-\sum_{j=0}^\infty b_j^{(m^*+1)}P(\text{Poi}(B\lambda/2)=j)\right|\\
& &+\sum_{m\neq m^*,m^*+1}P(\text{Poi}(B\lambda/2)\in BI_m/2)\left|\sum_{j=0}^\infty b_j^{(m)}P(\text{Poi}(B\lambda/2)=j)\right|,
\end{eqnarray*}
and using Lemma~\ref{Lemma:LocalPoissonPolynomial} and the same concentration bounds gives (\ref{Eqn:UpperBoundOfErrorInPoissonPoly}).
\item Case III: $\lambda\in I_{m^*}\cap I_{m^*-1}^\prime$. This case is entirely symmetric to Case II.
\end{itemize}
Combining the above three cases, we arrive at the inequality (\ref{Eqn:UpperBoundOfErrorInPoissonPoly}).

Next we verify the coefficient bound (\ref{Eqn:UpperBoundOfCoefficientsInPoissonPoly}). By Lemma~\ref{Lemma:LocalPoissonPolynomial}, it is clear from the definition that $b_x=0$ whenever $x\notin \cup_{m=1}^M BI_{m}^{\prime\prime}$ and hence $b_x=0$ for $x\geq 4B$.
Fix any $x\geq 0$ such that $b_x\neq 0$, assume that $x\in BI_{m^*}^{\prime\prime}$ (if there are multiple choices of $m^*$, pick an arbitrary one).
We claim that any other $m=1,2,\cdots,M$ such that $|m-m^*|\geq 5$ do not contribute to $b_x$ in the summation (\ref{Eqn:CoefficientOfRealB}).
In fact, if there is non-zero coefficient $b_{x-k}^{(m)}$ in (\ref{Eqn:CoefficientOfRealB}), we must have
\begin{align*}
&x\in BI_{m^*}^{\prime\prime}=c_1\log N\cdot \left[(m^*-2)_+^2,(m^*+1)^2\right],\text{~~~~}
k\in BI_{m}/2=\frac{c_1\log N}{2}\cdot \left[(m-1)^2,m^2\right],\\
&x-k\in BI_{m}^{\prime\prime}/2=\frac{c_1\log N}{2}\cdot \left[(m-2)_+^2,(m+1)^2\right].
\end{align*}
Summing up, we must have 
at least one of 
\begin{eqnarray*}
(m^*-2)^2_+\leq\frac{(m-1)^2+(m-2)^2_+}{2},\text{~~~~}
(m^*+1)^2\geq\frac{m^2+(m+1)^2}{2},
\end{eqnarray*}
 will fail whenever $|m-m^*|\geq 5$.
Hence, there exists constants $C_1,C_2$ such that 
\[
|b_x|\leq \frac{1}{2^x}\sum_{m=1}^M\sum_{k\in BI_m/2}{x\choose k}|b_{x-k}^{(m)}|
\leq
C_1\max_{m:|m-m^*|\leq 4}\max_{j\geq 0}|b_j^{(m)}|
\leq
\frac{C_2C(\epsilon)(1+x^{1/2})N^{\epsilon}}{B}
\]
establishing (\ref{Eqn:UpperBoundOfCoefficientsInPoissonPoly}).
\end{proof}

\begin{lemma}
\label{Lemma:LocalPoissonPolynomial}
Suppose $B>0$ and $N\in\mathbb{N}^+$ and there exists constants $c_0,C_0>0$ such that $B\in[c_0\log N, C_0N]$.
Let $\ell(\cdot)$ be any Lipschitz-$1$ function on $\mathbb{R}$ with $\ell(0)=0$.
Then, for any fixed $c_0\geq 16$ and any small $\epsilon\in(0,0.02)$ there exist constants $C(\epsilon)>0$ and $N(\epsilon)>1$ depending on $\epsilon$ and a sequence of coefficients $\{b_x\}_{x=0}^\infty$ such that the following inequality holds for $N\geq N(\epsilon)$, i.e.
\begin{eqnarray}
|\ell(\lambda)-\sum_{x=0}^\infty b_xP(\text{Poi}(B\lambda)=x)|\leq C(\epsilon)\sqrt{\frac{\lambda}{B\log N}}, \text{ for any }\lambda\in I^\prime_m,
\label{Eqn:LocalPoissonPolynomial}
\end{eqnarray}
where $b_x=0$ for $x\notin BI_m^{\prime\prime}$, and 
\begin{eqnarray}
\left|b_x-\ell\left(\frac{x}{B}\right)\right|\leq \frac{C(\epsilon)(1+x^{1/2})N^\epsilon}{B}, \text{ for any }x\in BI_m^{\prime\prime},
\label{Eqn:BoundOfCoefficientInLocalPoissonPolynomial}
\end{eqnarray}
where $I^\prime_m$ and $I_m^{\prime\prime}$ are defined in the proof of Proposition~\ref{Proposition:BoundOnf1.5} .
\end{lemma}

\begin{proof}[Proof of Lemma~\ref{Lemma:LocalPoissonPolynomial}]
Since $B\geq c_0\log N$, we have $B\geq 1$ for sufficiently large $N$.
Recall that $c_1=c_0/4\geq 4$.
Let $D:= c_2 \log N$
where $c_2>0$ is a small constant specified later and without loss of generality it is assumed that $D$ is an integer.
Throughout the proof we will use $C_1,C_2,\cdots$ to denote positive constants independent of $(B,c_1,c_2)$.
For $m=1$ it follows from Lemma~\ref{lem:jackson} that there exist coefficients $\{a_{1,d}\}_{d=0}^D$ such that 
\[
|\ell(\lambda)-\sum_{d=0}^Da_{1,d}(\lambda-\lambda_1)^d|
\leq C_1\frac{\sqrt{\frac{c_1\log N}{B}\cdot (4/3)^2\lambda}}{D}
= \frac{4C_1}{3c_2}\sqrt{\frac{c_1\lambda}{B\log N}}
\]
for all $\lambda\in I_1^\prime$.
If $m\geq2$, it follows from Lemma~\ref{lem:jackson} that there exist coefficients $\{a_{m,d}\}_{d=0}^D$ such that 
\[
|\ell(\lambda)-\sum_{d=0}^Da_{m,d}(\lambda-\lambda_m)^d|
\leq \frac{10C_1}{3c_2}\frac{c_1(m-\frac{1}{2})}{B}
\leq \frac{10C_1c_1}{c_2B}\left(m-\frac{4}{3}\right),
\]
where the last inequality follows from $m\geq 2$.
Then it follows from $m\leq \frac{4}{3}+\sqrt{\frac{B\lambda}{c_1\log N}}$ for all $\lambda\in I_m^\prime$ that
\[
|\ell(\lambda)-\sum_{d=0}^Da_{m,d}(\lambda-\lambda_m)^d|
\leq
\frac{10C_1}{c_2}\sqrt{\frac{c_1\lambda}{B\log N}}.
\]
Combining the above cases, it follows that for $m=1,\ldots,M$ and $\lambda\in I^\prime_m$
\[
|\ell(\lambda)-\sum_{d=0}^Da_{m,d}(\lambda-\lambda_m)^d|
\leq
10C_1\frac{\sqrt{c_1}}{c_2}\sqrt{\frac{\lambda}{B\log N}}.
\]
As a sequence, it follows that for $\lambda\in I^\prime_m$
\begin{eqnarray*}
|\ell(\lambda_m)-\sum_{d=0}^Da_{m,d}(\lambda-\lambda_m)^d|
\leq |\ell(\lambda)-\ell(\lambda_m)|+|\ell(\lambda)-\sum_{d=0}^Da_{m,d}(\lambda-\lambda_m)^d|
\leq\frac{5c_1m\log N}{B}\left(\frac{2}{3}+\frac{4C_1}{c_2\log N}\right).
\end{eqnarray*}

Moreover, applying Lemma~\ref{CoefficientBound} on the shifted interval $I_m^\prime-\lambda_m$ gives that for $d=1,2,\cdots,D$
\begin{eqnarray*}
|a_{m,d}|
\leq9\left(\frac{2}{3}+\frac{4C_1}{c_2\log N}\right)\left(\frac{5}{3}\frac{c_1m\log N}{B}\right)^{1-d}N^{c_2}.
\end{eqnarray*}
As for $d=0$, choosing $\lambda=\lambda_m$ in the above inequality gives $a_{m,0}\leq |\ell(\lambda_m)|+20C_1\frac{c_1m}{Bc_2}\leq 1+20C_1\frac{\sqrt{c_1}}{c_2}\sqrt{\frac{1}{B\log N}}$.

Next we write the above polynomial as a linear combination of Poisson polynomials.
Since
\[
\sum_{x=0}^\infty \frac{x!}{(x-d)!B^d}\cdot P(\text{Poi}(B\lambda)=x)=\lambda^d
\]
where $y!:=\infty$ for $y<0$,
we have 
\begin{eqnarray*}
\sum_{d=0}^Da_{m,d}(\lambda-\lambda_m)^d
=
\sum_{x=0}^\infty b_x^* P(\text{Poi}(B\lambda)=x),
\end{eqnarray*}
where $b_x^*:= \sum_{d=0}^Da_{m,d}\sum_{d^\prime=0}^d{d \choose d^\prime}(-\lambda_m)^{d-d^\prime} \frac{j!}{(x-d^\prime)!B^{d^\prime}}$.

In other words, the inequality holds for the coefficients $\{b_x^*\}_{x=0}^\infty$.
Now we define $\{b_x\}_{x=0}^\infty$ to be the truncated version of $\{b_x^*\}_{x=0}^\infty$:
\[
b_x=b_x^*\cdot 1(x\in B I_m^{\prime\prime}).
\]
Clearly $b_x=0$ for all $x\notin BI_m^{\prime\prime}$. 
By Lemma~\ref{Lemma:CharlierPolynomials}, for $d=1,2,\cdots, D,$
\begin{eqnarray*}
& &|\sum_{d^\prime=0}^d{d \choose d^\prime}(-\lambda_m)^{d-d^\prime}\frac{x!}{(x-d^\prime)!B^{d^\prime}}|\\
&\leq&
\left\{
\begin{matrix}
\left(8 \max\left\{\frac{c_1m\log N}{B},\frac{1+c_2\log N}{B},\frac{m\sqrt{(c_1\log N)(1+c_2\log N)}}{B}\right\}\right)^d & \text{ if }x\in BI_m^{\prime\prime}\\
\left(8\max\left\{\left|\frac{x}{B}-\lambda_m\right|/4,\frac{1+c_2\log N}{B},\frac{m\sqrt{(c_1\log N)(1+c_2\log N)}}{B}\right\}\right)^d &\text{ otherwise.}
\end{matrix}
\right.
\end{eqnarray*}
Suppose $c_2<c_1$.
Then for $N\geq \exp(1/c_2)$, it follows that $c_2\log N\geq 1$ and hence
\[
\frac{1+c_2\log N}{B}\leq \frac{2c_2\log N}{B}\leq\frac{2c_1m\log N}{B} \text{ and }
\frac{m\sqrt{(c_1\log N)(1+c_2\log N)}}{B}\leq\frac{2c_1m\log N}{B}.
\]
Therefore,
\begin{eqnarray*}
|\sum_{d^\prime=0}^d{d \choose d^\prime}(-\lambda_m)^{d-d^\prime}\frac{x!}{(x-d^\prime)!B^{d^\prime}}|
\leq
\left\{
\begin{matrix}
\left(16\frac{c_1m\log N}{B}\right)^d & \text{ if }x\in BI_m^{\prime\prime}\\
\left(8\left|\frac{x}{B}-\lambda_m\right|\right)^d &\text{ otherwise.}
\end{matrix}
\right.
\end{eqnarray*}

Hence, for $x\in BI_m^{\prime\prime}$, we have
\begin{eqnarray*}
|b_x-a_{m,0}|=|b_x^*-a_{m,0}|
\leq\frac{17}{B}\left(\frac{2}{3}+\frac{4C_1}{c_2\log N}\right)c_1\log N\cdot \left(\sqrt{\frac{x}{c_1\log N}}+2\right)N^{3c_2}.
\end{eqnarray*}
Since $c_2<c_1$ and $c_2\log N\geq1$, it follows that $c_1\log N\geq 1$ and hence
\begin{eqnarray*}
|b_x-a_{m,0}|=|b_x^*-a_{m,0}|
\leq 34c_1\left(\frac{2}{3}+4C_1\right)\cdot \frac{(\sqrt{x}+1)N^{4c_2}}{B},
\end{eqnarray*}
for $N$ sufficiently large (depending on $c_2$).

Moreover, for any $x\in BI_m^{\prime\prime}$,
\begin{eqnarray*}
|a_{m,0}-\ell\left(\frac{x}{B}\right)|
\leq \left(8c_1+40C_1\frac{c_1}{c_2}\right)\frac{(\sqrt{x}+1)\log N}{B}
\end{eqnarray*}
and therefore a triangle inequality gives the inequality (\ref{Eqn:BoundOfCoefficientInLocalPoissonPolynomial}).

As for the other inequality (\ref{Eqn:LocalPoissonPolynomial}), by triangle inequality it suffices to prove that 
\[
\sum_{x\notin BI_m^{\prime\prime}}|b_x^*|\cdot P(\text{Poi}(B\lambda)=x)=O(N^{-1}),\text{ for any }\lambda\in I_m^\prime.
\]
To prove the last display, first note that for $x\notin BI_m^{\prime\prime}$, we have
\[
|x-B\lambda_m|\geq 2c_1m\log N
\]
and
\begin{eqnarray*}
|b_x^*|
\leq C(c_1,c_2)+17\left(\frac{2}{3}+\frac{4C_1}{c_2\log N}\right)N^{c_2}\left(\frac{7|x-B\lambda_m|}{\sqrt{c_1B\lambda_m\log N}}\right)^D,
\end{eqnarray*}
where $C(c_1,c_2)$ is a constant depending on $c_1,c_2$.
Futhermore, by the Chernoff bound (Lemma~\ref{Lemma:PoissonTailInequality}), we have 
\begin{eqnarray*}
P(\text{Poi}(B\lambda)=x)
\leq\exp\left(-\frac{1}{3}|x-B\lambda|\left(\frac{|x-B\lambda|}{B\lambda}\wedge1\right)\right).
\end{eqnarray*}
Since for all $\lambda\in I_m^\prime$ and $x\notin BI_m^{\prime\prime}$ we have $|x-B\lambda|\geq 4c_1m\log N$ and $B\lambda\leq4 c_1m^2\log N$, it follows that
$
|x-B\lambda|/(B\lambda)\geq1/m
$
and hence
\begin{eqnarray*}
P(\text{Poi}(B\lambda)=x)
\leq\exp\left(-\frac{1}{6}\cdot c_1\log N\cdot \frac{|x-B\lambda|}{\sqrt{c_1B\lambda_m\log N}}\right).
\end{eqnarray*}
Moreover, the assumption $x\notin BI_m^{\prime\prime}$ implies that $|x-B\lambda_m|/\sqrt{c_1B\lambda_m\log N}\geq 2>0$.
Consequently, whenever $\lambda\in I_m^\prime$ and $x\notin BI_m^{\prime\prime}$, we have
\begin{align*}
&\sum_{x\notin BI_m^{\prime\prime}}|b_x^*|P(\text{Poi}(B\lambda)=x)\\
&\leq
2C(c_1,c_2)N^{-c_1/3}+\sum_{x\notin BI_m^{\prime\prime}}C_3
\exp\left(10c_2\log N\cdot\log\frac{|x-B\lambda_m|}{\sqrt{c_1B\lambda_m\log N}}-\frac{1}{6}\cdot c_1\log N\cdot \frac{|x-B\lambda|}{\sqrt{c_1B\lambda_m\log N}}\right),
\end{align*}
where the first term in the last display follows from (\ref{Formula:ExpUpperBounded2}) and  $C_3$ in the second term is a positive constant which doesn't depend on $c_1$.
Then
by choosing $c_2>0$ small enough we arrive at an exponent$=\left(-\frac{1}{7}c_1\log N\frac{|x-B\lambda_m|}{\sqrt{c_1B\lambda_m\log N}}\right)$.
It follows from $|x-B\lambda_m|/\sqrt{c_1B\lambda_m\log N}\geq 2>0$ that 
\begin{eqnarray*}
\sum_{x\notin BI_m^{\prime\prime}}|b_x^*|P(\text{Poi}(B\lambda)=x)
\leq2C(c_1,c_2)N^{-c_1/3}+C_3N^{-2c_1/7}=O(N^{-2c_1/7})=O(N^{-1}),\text{ for }c_1\geq 4.
\end{eqnarray*}
This completes the proof.
\end{proof}

\section{Implementation details in Section \ref{sec:app}}\label{sec:app-details}

Let $n_1=13$ and $n_2=10$ denote the number of subjects in ASD and control groups respectively and $n=23$ represent the number of total subjects.
Since 99 percent of $\{X_{ij}^{(k)}/r_{ij}^{(k)},i\in[N_{jk}],j\in[n_k],k\in[K]\}$ for 100 genes are smaller than 15.09, we choose $B=20$. We use VEM to compute NPMLEs with a stop tolerance 0.01. The testings with covariance adjustments $\hat T_Z$ and $\hat T_{h,Z}$ are conducted by R package ``\textit{ideas}'' by Sun and Zhang with $10^5$ Monte Carlo simulations.

To account for covariates, the pseudo-$F$ statistics described in Section~\ref{sec:test} has to be changed a little bit. Let $\textbf{D}_n$ be the $n$ by $n$ distance matrix corresponding to the mixing distributions, with each entry equal to the squared $W_1$ distance between the two corresponding NPMLEs, and let
 \begin{eqnarray*}
\textbf{G}_n:=\left(\textbf{I}_n-\frac{1}{n}\textbf{1}_n\textbf{1}_n^\top\right)\textbf{A}_n\left(\textbf{I}_n-\frac{1}{n}\textbf{1}_n\textbf{1}_n^\top\right),
 \end{eqnarray*}
 be the Grower's center matrix of $\textbf{A}_n$, where $\textbf{A}_n:= -(1/2)\textbf{D}_n$, $\textbf{1}_n:=(\underbrace{1,1,\ldots,1}_{n})^\top$, and $\textbf{I}_n$ stands for the $n$-dimensional identity matrix. Note that $\textbf{G}_n$ may have some negative eigenvalues, and we set those negative eigenvalues to 0. Let $\textbf{Z}$ be an $n$ by 5 matrix consisting of diagnostics (1 as ASD and 0 as control), age, sex, seqbatch, and RIN. Let $\textbf{H}_Z$ be the hat matrix $\textbf{H}_Z:=\textbf{Z}(\textbf{Z}^\top \textbf{Z})^{-1}\textbf{Z}^\top$.
 Then the new $F$-statistic accounting for covariates is 
 \begin{eqnarray}
\hat F_Z:=\frac{{\rm tr}(\textbf{H}_Z\textbf{G}\textbf{H}_Z)}{{\rm tr}((\textbf{I}-\textbf{H}_Z)\textbf{G}(\textbf{I}-\textbf{H}_Z))},
 \label{Formula:DefinitionOfpseudoFZ}
 \end{eqnarray}
 where ${\rm tr}(\cdot)$ denotes the trace of a matrix. 
To implement the permutation test, we permute the variable ``diagnostics'' with all the rest covariates fixed and accordingly generate a new data matrix $Z^\pi$. The corresponding $F$-statistic is denoted by $\hat F_Z^\pi$ and the $p$-value is 
\begin{eqnarray}
\frac{\text{the number of permutations $\pi$ such that }\hat{F}_Z^\pi\geq\hat{F}_Z}{\text{the number of all possible permutations }\pi}.
\label{Formula:p-valueMixingCov}
\end{eqnarray}

When replacing the distance matrix $\textbf{D}_n$ by the corresponding Poisson-smoothed version, the corresponding $p$-value is 
\begin{eqnarray}
\frac{\text{the number of permutations $\pi$ such that }\hat{F}_{h,Z}^\pi\geq\hat{F}_{h,Z}}{\text{the number of all possible permutations }\pi},
\label{Formula:p-valueMixtureCov}
\end{eqnarray}
where $\hat{F}_{h,Z}$ and $\hat{F}_{h,Z}^\pi$ are the Poisson-smoothed versions of $\hat{F}_{Z}$ and $\hat{F}_{Z}^\pi$.

The above two testing procedures are abbreviated as $\hat T_Z$ and $\hat T_{h,Z}$.

\section*{Acknowledgement}

The authors would like to thank Yihong Wu for his very informative remarks on the issue of uniqueness of NPMLEs and concavity of the nonparametric Poisson likelihood functions. The authors would also like to thank Jiahua Chen, Matthew Stephens, and Jon Wellner for pointing out related literature and for helpful discussions. 

{\small
\bibliographystyle{apalike}
\bibliography{npmle}
}

\end{document}